\documentclass[twocolumn,pra,showpacs,showkeys,groupedaddress,nofootinbib,longbibliography]{revtex4-1}

\input{preamble}
\usepackage{pifont}
\usepackage{xcolor}	\definecolor{darkgreen}{rgb}{0.0, 0.5, 0.0}

\begin{document}
\title{Classicality without local discriminability: decoupling entanglement and complementarity}
\author{Giacomo Mauro D'Ariano}
\email{dariano@unipv.it}
\author{Marco Erba}
\email{marco.erba@unipv.it}
\author{Paolo Perinotti}
\email{paolo.perinotti@unipv.it}
\affiliation{Universit\`a degli Studi di Pavia, Dipartimento di Fisica, QUIT Group, and INFN Gruppo IV, Sezione di Pavia, via Bassi 6, 27100 Pavia, Italy}
\begin{abstract}
An operational probabilistic theory where all systems are classical, and all pure states of composite systems are entangled, is constructed. The theory is endowed with a rule for composing an arbitrary number of systems, and with a nontrivial set of transformations. Hence, we demonstrate that the presence of entanglement is independent of the
existence of incompatible measurements. We then study a variety of phenomena occurring in the theory---some of them contradicting both Classical and Quantum Theories---including: cloning, entanglement swapping, dense coding, additivity of classical capacities, non-monogamous entanglement, hypersignaling. We also prove the existence, in the theory, of a universal processor.
The theory is causal and satisfies the no-restriction hypothesis. At the same time, it violates a number of information-theoretic principles enjoyed by Quantum Theory, most notably: local discriminability, purity of parallel composition of states, and purification. Moreover, we introduce an exhaustive procedure to construct generic operational probabilistic theories, and a sufficient set of conditions to verify their consistency. 
In addition, we prove a characterisation theorem for the parallel composition rules of arbitrary theories, and specialise it to the case of bilocal-tomographic theories. We conclude pointing out some open problems. In particular, 
on the basis of the fact that every separable state of the theory is a statistical mixture of entangled states, we formulate a no-go conjecture for the existence of a local-realistic ontological model.
\end{abstract}
\maketitle

\tableofcontents

\section{\textbf{Introduction}}
In \emph{Discussion of probability relations between separated systems}~\cite{schrodinger1935discussion}, a paper dating back to 1935, Schr{\"o}dinger provides a seminal description of the phenomenon of entanglement: ``When two systems, of which we know the states by their respective representatives, enter into temporary physical interaction due to known forces between
them, and when after a time of mutual influence the systems separate again, then they can no longer be described in the same way as before, viz.~by endowing each of them with a representative of its own. I would not call that \emph{one} but rather \emph{the} characteristic trait of quantum mechanics, the one that enforces its entire departure from classical lines of thought.'' Indeed, in Classical Theory (CT), the state of affair of every system admits of a suitable description as a statistical mixture of products of pure states. This means that in CT any bipartite state can be prepared by two experimenters just by using local operations and shared randomness. In other words, every state in CT is \emph{separable}. On the contrary, Quantum Theory (QT) allows for states which are not separable, namely, which are \emph{entangled}. Actually, in the broad landscape of probabilistic theories~\cite{barrett2007information}, entanglement is far from being \emph{the} characteristic trait of QT. Most notably, the so-called \emph{PR boxes}~\cite{Rastall:1985aa,Popescu1994} provided the first example of a probabilistic model beyond QT and featuring entangled states, initiating a fruitful research field in the scope of foundations of physics. On top of this, it has been even argued that entanglement is \emph{an inevitable feature
of any theory superseding classical theory}~\cite{PhysRevLett.119.080503}.

While entanglement can be considered a ubiquitous feature in the scope of probabilistic theories, one of the points of the present work is to question to which extent the presence of entanglement in a physical theory \emph{enforces its entire departure from classical lines of thought}. One widespread notion of \emph{classicality} for a physical theory~\cite{barrett2007information,PhysRevLett.99.240501,pfister2013information,PhysRevLett.119.080503} is based on the set of states of the theory. A theory is deemed classical if the pure states of states for every system are: (i) the vertices of a simplex, and (ii) jointly perfectly discriminable. We present a complete operational probabilistic theory, which we call Bilocal Classical Theory (BCT), that, in spite of being classical in the above sense, features entangled states. Classical theories with entanglement have been characterised in Ref.~\cite{d2019classical}. Moreover, it is interesting to notice that BCT also complies with the notion of classicality proposed in Ref.~\cite{schmid2020structure}, admitting of a noncontextual ontological model.

BCT represents the proof of concept that entanglement is compatible with the absence of \emph{complementarity}, i.e., with the existence of incompatible measurements. The theory is causal and satisfies the no-restriction hypothesis. However, BCT violates the principles of local discriminability, purity of parallel composition of states, and purification. The theory also features non-monogamous entanglement and hypersignaling~\cite{PhysRevLett.119.020401}. Furthermore, we show that in BCT it is possible to perform entanglement swapping and dense coding, and we prove a theorem of universal programmability for the theory.

The paper is organized as follows. In Sec.~\ref{sec:conc} we provide a conceptual preview of the main results of the paper in a non technical way. In Sec.~\ref{sec:OPTs} we review the framework of operational probabilistic theories (OPTs), which captures the building blocks necessary to construct a theory of physical processes. This can be done by resorting to \emph{Category Theory}~\cite{lane2013categories}. Indeed, the categorical framework provided a powerful toolbox for deriving Quantum Theory~\cite{chiribella2016quantum}, reformulating it~\cite{coecke2017picturing}, and also for modelling generic physical theories~\cite{PhysRevA.81.062348,Chiribella2015,gogioso2017categorical}, or more general kinds of theories~\cite{selby2017process,tull2019categorical}. In Sec.~\ref{sec:structure_consistency} we make a survey of the operational probabilistic structures that are relevant to the present work. We prove two first results, providing in particular a characterisation for the composition rule of the system-sizes in bilocal-tomographic theories. We then move on to considering the construction of a probabilistic theory. When a mathematical structure is constructed, one needs to provide a procedure not only to build it, but also to make sure that it is consistent. Accordingly, in Sec.~\ref{sec:constructing_an_opt} we set the problem of consistency for the construction of an operational probabilistic theory, proposing a building procedure, and identifying a sufficient set of conditions to check well-posedness and coherence. Thus, in Sec.~\ref{sec:BCT} we present Bilocal Classical Theory, also verifying its consistency, while in Sec.~\ref{sec:features} we analyse a variety of information-theoretic properties of the theory. In Sec.~\ref{sec:discussion} we discuss our findings in the light of the existing literature. Finally, in Sec.~\ref{sec:conclusions} we draw our conclusions, pointing out some open problems of foundational and interpretative relevance.

\section{\textbf{Conceptual preview}}\label{sec:conc}
In the present paper we exhibit a complete probabilistic theory that embodies features studied in Ref.~\cite{d2019classical}, where classical theories without local discriminability were investigated. Such theories, in particular, feature entanglement, and allow to study its logical dependence on other properties that are commonly associated with it. The first step that we take is to review the framework of operational probabilistic theories (OPTs). OPTs allow for the study of statistics of  events, as in generalised probabilistic theories (GPTs), but they also enable us to describe the transformations and their composition in sequence and in parallel. Unlike most of the literature on GPTs, with remarkable exceptions~\cite{schmid2020structure}, OPTs formalise in a thorough manner the compositional structure, and can be thus thought as a completion of the framework of GPTs, emphasising the role of transformations. Introducing new structure on top of the statistical content of a GPT, an OPT has to satisfy further consistency constraints, that regard not only  its coherence as a probability theory~\cite{janotta2014generalized}, but, most importantly, the interplay of compositional structures\footnote{The compositional structure of an OPT, in technical terms, corresponds to that of a {\em monoidal category}~\cite{coecke2010categories}.} among themselves, and with the probabilistic one.

Classical theories are OPTs that are \emph{locally classical}, i.e.~the state space of every system is  classical, but their composition rule is not the usual one. Constructing such a theory presents remarkable difficulties, in particular, abiding by consistency constraints. For this reason the paper presents an exhaustive discussion of such constraints, that are then checked after introducing in detail an example of classical theory with entanglement.  

We now provide a brief survey of the main results of the paper. Theorem~\ref{thm:delta3} provides necessary conditions that the parallel composition rules for arbitrary theories must satisfy,  and Theorem~\ref{thm:bilocal} provides a characterisation of those rules in the case of bilocal-tomographic theories~\cite{hardy2012limited}. Bilocal discriminability is particularly relevant in the present work, since the theory here presented is bilocal-tomographic. In Sec.~\ref{sec:constructing_an_opt} we formulate a procedure to check the consistency of a given theory, that turns out to be crucial in establishing the validity of the theory constructed in Sec.~\ref{sec:BCT}. The latter, that we call  Bilocal Classical Theory (BCT), represents the main result of the paper. It constitutes a proof of concept that incompatibility of measurements  and entanglement are independent properties. In subsequent sections we analyse other relevant features of  BCT, some in common with Quantum Theory---e.g.~dense coding, entanglement swapping---some in common with Classical Theory---e.g.~perfect clonability, full information without disturbance, no Bell nonlocality, universal programming, insecure cryptography---and finally some which supersede both theories---such as the violation of  atomicity of parallel composition and of entanglement monogamy, hypersignaling.

\section{\textbf{Operational Probabilistic Theories: a review}}\label{sec:OPTs}
The primitive notions of an operational probabilistic theory (OPT) are those of \emph{systems}, \emph{tests}, \emph{events}, and \emph{probabilities}. Systems represent the physical entities which are probed in a laboratory (e.g.~an electron, a molecule, a radiation field, etc\ldots). Tests represent the physical processes---occurring between two systems---which experiments are made up of (such as the single use of a physical device). Accordingly, an \emph{outcome space} is associated with each test, collecting all the possible outcomes of the test itself. On the other hand, an \emph{event} is associated with each outcome, representing a possible occurrence in a physical process. Finally, the goal of a physical theory is to associate some \emph{probability distributions} with each event.

In the present section, we will provide a review of the main general properties of an OPT. In particular, we will present the operational and compositional properties of the primitives, their probabilistic structure, and the resulting linear structure.

\subsection{Operational structure: compositional properties of a theory}\label{subsec:operational}
Let $\Sys{\Theta}$, $\TestA{\Theta}$ and $\EventA{\Theta}$ denote the classes of, respectively, the systems, the tests, and the events of a theory $\Theta$. Systems will be denoted using latin characters $\sys{A},\sys{B},\sys{C},\ldots\in\Sys{\Theta}$, while tests from a system $\sys{A}$ to a system $\sys{B}$ will be denoted by $\TestC{T}{X}{A}{B}\in\TestAB{A}{B}$, where $\Out{X}$ is the corresponding outcome space.
A test having a single outcome $\ast$ will be called \emph{a singleton test}, and the singleton set will be denoted by $\star\coloneqq\lbrace\ast\rbrace$. Without loss of clarity, sometimes we will use the shorthand notation $\Out{T}\equiv\Test{T}{X}\equiv\TestC{T}{X}{A}{B}$. Each test $\TestC{T}{X}{A}{B}$ is a collection of events from $\sys{A}$ to $\sys{B}$, namely:
	\begin{align*}
	\TestC{T}{X}{A}{B}=\lbrace\T{T}_x\rbrace_{x\in\Out{X}};\quad \forall x\in\Out{X},\ \T{T}_x\colon\sys{A}\rightarrow\sys{B}.
	\end{align*}
The class of events from a system $\sys{A}$ to a system $\sys{B}$ will be denoted by $\Event{A}{B}$. Events $\T{T}\in\Event{A}{B}$ are represented as wired boxes, where the source and target systems  are the labels of the input and, respectively, the output wires:
	\begin{align*}
	\Qcircuit @C=1em @R=1em {
		&\s{A}&\gate{\T{T}}&\s{B}\qw&
	}.
\end{align*}
As it is clear from the diagram above, the input-output direction is conventionally represented as going from left to right.

One requires that every test with, say, output system $\sys{B}$, can be sequentially composed with any other test having $\sys{B}$ as input. This operation represents the consecutive occurrence of two physical processes. That is, for all $\sys{A},\sys{B},\sys{C}\in\Sys{\Theta}$ and all $\T{T}_1\in\Event{A}{B},\T{T}_2\in\Event{B}{C}$, there exist an associative map $\circ$, called \emph{sequential composition}, and an event $\T{T}_2\T{T}_1\coloneqq\T{T}_2\circ\T{T}_1\in\Event{A}{C}$, such that the sequential composition of two tests $\TestC{T}{X}{A}{B},\TestC{T'}{Y}{B}{C}$ is defined as:
\begin{align*}
\TestC{\left( T'T\right)}{X\times Y}{A}{C}\equiv\TestC{T'}{Y}{B}{C}\TestC{T}{X}{A}{B} \coloneqq \lbrace \T{T}'_{y}\circ\T{T}_x \rbrace_{(x,y)\in\Out{X\times Y}}.
\end{align*}
Sequential composition is pictorially represented by the horizontal juxtaposition of boxes from left to right, connecting the two consecutive input/output wires which carry the same label:
\begin{align*}
\begin{aligned}
\Qcircuit @C=1em @R=1.3em
{
	&\s{A}&\gate{\T{T}_2\T{T}_1}&\s{C}\qw&
}
\end{aligned}
=
\begin{aligned}
\Qcircuit @C=1em @R=1.3em
{
	&\s{A}&\gate{\T{T}_1}&\s{B}\qw&\gate{\T{T}_2}&\s{C}\qw&
}
\end{aligned}.
\end{align*}
Moreover, for all $\sys{S}\in\Sys{\Theta}$ there exists a (unique) singleton test---denoted by $\TestC{I}{\star}{S}{S}=\lbrace\T{I}_\sys{S}\rbrace$ and called \emph{the identity of $\sys{S}$}---satisfying, for all systems $\sys{A}$ and $\sys{B}$, the following property:
	\begin{align}\label{eq:identity_sequential}
	\T{I}_\sys{B}\T{T}_x=\T{T}_x=\T{T}_x\T{I}_\sys{A},\quad\forall\TestC{T}{X}{A}{B},\ \forall\T{T}_x\in\TestC{T}{X}{A}{B}.
	\end{align}
In the following, we will denote the identity family also by using the symbol $\T{I}$. Identity processes are equivalent to doing nothing, and can be then conveniently represented just as an extended wire carrying the respective system's label:
\begin{align*}
	\begin{aligned}
	\Qcircuit @C=1em @R=1em {
		&\qw&\s{A}\qw&\qw&\qw&
	}
	\end{aligned}
	=
	\begin{aligned}
	\Qcircuit @C=1em @R=1em {
		&\s{A}&\gate{\T{I}_\sys{A}}&\s{A}\qw&
	}
	\end{aligned}.
\end{align*}

One can compose two systems $\sys{A}$ and $\sys{B}$ to make the new composite system $\sys{AB}$. Correspondingly, any two arbitrary events $\T{T}_1\in\Event{A}{B},\T{T}_2\in\Event{C}{D}$  can be composed in parallel. This operation corresponds to an associative map $\boxtimes$, called \emph{parallel composition}, that produces the composite event $\T{T}_1\boxtimes\T{T}_2\in\Event{AC}{BD}$. The parallel composition of tests is then straightforwardly defined as:
\begin{align*}
\TestC{\left( T\boxtimes T'\right)}{X\times Y}{AC}{BD}\equiv\TestC{T}{X}{A}{B}\boxtimes\TestC{T'}{Y}{C}{D}\!\coloneqq \lbrace \T{T}_x\boxtimes\T{T}'_{y} \rbrace_{(x,y)\in\Out{X\times Y}}.
\end{align*}
This can be thought as a composite test which is completely described by two single tests performed in different laboratories. Parallel composition is pictorially represented by vertically juxtaposing transformations from top to bottom:
\begin{align*}
	\begin{aligned}
	\Qcircuit @C=1em @R=1em
	{
		&\s{AC}&\gate{\T{T}_1\boxtimes\T{T}_2}&\s{BD}\qw&
	}
	\end{aligned}
	&=
	\begin{aligned}
	\Qcircuit @C=1em @R=1em
	{
		&\s{A}&\multigate{1}{\T{T}_1\boxtimes\T{T}_2}&\s{B}\qw&
		\\
		&\s{C}&\ghost{\T{T}_1\boxtimes\T{T}_2}&\s{D}\qw&
	}
	\end{aligned}
	=\\[2.5ex]
	&=
	\begin{aligned}
	\Qcircuit @C=1em @R=1em
	{
		&\s{A}&\gate{\T{T}_1}&\s{B}\qw&
		\\
		&\s{C}&\gate{\T{T}_2}&\s{D}\qw&
	}
	\end{aligned}.
\end{align*}
Moreover, the following properties are required for all $\sys{A},\sys{B},\sys{C},\sys{D},\sys{E},\sys{F}\in\Sys{\Theta}$ and all $\T{T}_1\in\Event{A}{B}$, $\T{T}_2\in\Event{B}{C}$, $\T{T}_3\in\Event{D}{E}$, $\T{T}_4\in\Event{E}{F}$:
\begin{align}
	\begin{aligned}\label{eq:identity}
	\Qcircuit @C=1em @R=2em {
		&\qw&\s{A}\qw&\qw&\qw&\\
		&\qw&\s{B}\qw&\qw&\qw&
	}
	\end{aligned}
	&=\ \ 
	\begin{aligned}
	\Qcircuit @C=1em @R=2em {
		&\qw&\s{AB}\qw&\qw&\qw&
	}
	\end{aligned},
	\\[2.5ex]
	\begin{aligned}\label{eq:bifunctoriality_diagrammatic}
	\Qcircuit @C=1em @R=1.3em
	{
	\s{A}&\qw&\gate{\T{T}_1}&\s{B}\qw&\gate{\T{T}_2}&\qw&\s{C}\qw&
	\\
	\s{D}&\qw&\gate{\T{T}_3}&\s{E}\qw&\gate{\T{T}_4}&\qw&\s{F}\qw&
	\relax\gategroup{1}{3}{1}{5}{.8em}{--}
	\relax\gategroup{2}{3}{2}{5}{.8em}{--}
	\relax\gategroup{1}{3}{2}{5}{2em}{--}
	}
	\end{aligned}
	&\!\!=\ \ 
\begin{aligned}
\Qcircuit @C=1.2em @R=1.3em
{
	\s{A}&\qw&\gate{\T{T}_1}&\s{B}\qw&\gate{\T{T}_2}&\qw&\s{C}\qw&
	\\
	\s{D}&\qw&\gate{\T{T}_3}&\s{E}\qw&\gate{\T{T}_4}&\qw&\s{F}\qw&
	\relax\gategroup{1}{3}{2}{3}{.8em}{--}
	\relax\gategroup{1}{5}{2}{5}{.8em}{--}
	\relax\gategroup{1}{3}{2}{5}{2em}{--}
}
\end{aligned}\!\!\!.
\end{align}
Eq.~\eqref{eq:identity} asserts that the parallel composition of the identities $\T{I}_\sys{A}$ and  $\T{I}_\sys{B}$ is the identity $\T{I}_\sys{AB}$ of the compound system $\sys{AB}$. Property~\eqref{eq:bifunctoriality_diagrammatic}  states that the operations of sequential and parallel composition commute.

Furthermore, one requires the possibility to consider physical processes having no input or output. The corresponding tests are those where the experimenter disregards everything that---from the viewpoint of the input--output direction---has occurred before or, respectively, will occur after, the physical process considered. In order to capture this notion, there exists a (unique) system $\sys{I}$, called the \emph{trivial system}, satisfying: $\sys{IA}=\sys{A}=\sys{AI}$ for all $\sys{A}\in\Sys{\Theta}$. Also, the parallel composition of any test with the identity of $\sys{I}$ amounts to doing nothing. Indeed, one can conveniently omit an explicit diagrammatic representation of both $\sys{I}$ and $\T{I}_\sys{I}$, leaving blank spaces. Accordingly, events $\rho\in\Event{I}{A},a\in\Event{A}{I}$ will be represented, respectively, as:
\begin{align*}
\begin{aligned}
\Qcircuit @C=1em @R=1em {
	&\prepareC{\rho}&\s{A}\qw&
}
\end{aligned},
\begin{aligned}
\Qcircuit @C=1em @R=1em {
	&\s{A}&\measureD{a}&
}
\end{aligned}.
\end{align*}
Such events are called, respectively, \emph{preparations} and \emph{observations}.

Finally, we introduce one last relevant feature. One can think to each system as being controlled by an agent. Accordingly, one also requires the possibility of exchanging systems between agents. This operation is captured by the notion of \emph{braiding}, namely, a family $\Out{S}$ of invertible singleton tests defined in the following way. For any two systems $\sys{X},\sys{Y}$, the braiding $\Out{S}$ contains two tests, $\TestC{S}{\star}{XY}{YX}=\lbrace\T{S}_{\sys{X},\sys{Y}}\rbrace$ and its inverse $\left(\mathsf{S}_{\star}^{-1}\right)^{\sys{YX}\!\rightarrow\!\sys{XY}}=\lbrace\T{S}^{-1}_{\sys{X},\sys{Y}}\rbrace$, whose associated events will be denoted as follows:
\begin{align*}
\begin{aligned}
\Qcircuit @C=1em @R=2em
{
	\s{X}&\multigate{1}{\T{S}_{\sys{X},\sys{Y}}}&\s{Y}\qw&
	\\
	\s{Y}&\ghost{\T{S}_{\sys{X},\sys{Y}}}&\s{X}\qw&
}
\end{aligned}
&=\begin{tikzpicture}
	\begin{pgfonlayer}{nodelayer}
		\node [style=none] (0) at (-1.25, -1) {};
		\node [style=none] (1) at (-1.25, 1) {};
		\node [style=none] (2) at (1.25, 1) {};
		\node [style=none] (3) at (1.25, -1) {};
		\node [style=none] (4) at (-0.25, -0.25) {};
		\node [style=none] (5) at (0.25, 0.25) {};
		\node [style=none] (6) at (1.25, 1) {};
		\node [style=none] (7) at (1.25, -1) {};
		\node [style={system_label}] (8) at (1.75, 1.75) {$\scriptstyle \sys Y$};
		\node [style={system_label}] (9) at (1.75, -0.25) {$\scriptstyle \sys  X$};
		\node [style=none] (10) at (-1.75, 1) {};
		\node [style=none] (11) at (-1.75, -1) {};
		\node [style=none] (12) at (-1.25, 1) {};
		\node [style=none] (13) at (-1.25, -1) {};
		\node [style={system_label}] (14) at (-1.75, 1.75) {$\scriptstyle \sys  X$};
		\node [style={system_label}] (15) at (-1.75, -0.25) {$\scriptstyle \sys  Y$};
		\node [style=none] (16) at (1.75, 1) {};
		\node [style=none] (17) at (1.75, -1) {};
	\end{pgfonlayer}
	\begin{pgfonlayer}{edgelayer}
		\draw [in=-180, out=0, looseness=0.75] (1.center) to (3.center);
		\draw [in=-135, out=0] (0.center) to (4.center);
		\draw [in=180, out=45] (5.center) to (2.center);
		\draw (10.center) to (12.center);
		\draw (11.center) to (13.center);
		\draw (6.center) to (16.center);
		\draw (7.center) to (17.center);
	\end{pgfonlayer}
\end{tikzpicture},
\\[1ex]
\begin{aligned}
\Qcircuit @C=1em @R=2em
{
	\s{Y}&\multigate{1}{\T{S}^{-1}_{\sys{X},\sys{Y}}}&\s{X}\qw&
	\\
	\s{X}&\ghost{\T{S}^{-1}_{\sys{X},\sys{Y}}}&\s{Y}\qw&
}
\end{aligned}
&=\begin{tikzpicture}
	\begin{pgfonlayer}{nodelayer}
		\node [style=none] (0) at (-1.25, -1) {};
		\node [style=none] (1) at (-1.25, 1) {};
		\node [style=none] (2) at (1.25, 1) {};
		\node [style=none] (3) at (1.25, -1) {};
		\node [style=none] (4) at (0.25, -0.25) {};
		\node [style=none] (5) at (-0.25, 0.25) {};
		\node [style=none] (6) at (1.25, 1) {};
		\node [style=none] (7) at (1.25, -1) {};
		\node [style={system_label}] (8) at (1.75, 1.75) {$\scriptstyle \sys  X$};
		\node [style={system_label}] (9) at (1.75, -0.25) {$\scriptstyle \sys  Y$};
		\node [style=none] (10) at (-1.75, 1) {};
		\node [style=none] (11) at (-1.75, -1) {};
		\node [style=none] (12) at (-1.25, 1) {};
		\node [style=none] (13) at (-1.25, -1) {};
		\node [style={system_label}] (14) at (-1.75, 1.75) {$\scriptstyle \sys  Y$};
		\node [style={system_label}] (15) at (-1.75, -0.25) {$\scriptstyle \sys  X$};
		\node [style=none] (16) at (1.75, 1) {};
		\node [style=none] (17) at (1.75, -1) {};
	\end{pgfonlayer}
	\begin{pgfonlayer}{edgelayer}
		\draw (10.center) to (12.center);
		\draw (11.center) to (13.center);
		\draw (6.center) to (16.center);
		\draw (7.center) to (17.center);
		\draw [in=-180, out=0, looseness=0.75] (13.center) to (6.center);
		\draw [in=135, out=0] (12.center) to (5.center);
		\draw [in=-180, out=-45] (4.center) to (7.center);
	\end{pgfonlayer}
\end{tikzpicture}.
\end{align*}
The above graphical representation as a twisting of systems is due to the fact that the braiding is required to obey, for all $\sys{A},\sys{B},\sys{C},\sys{D}\in\Sys{\Theta}$ and all $\T{T}_1\in\Event{A}{B},\ \T{T}_2\in\Event{C}{D}$, to the following ``sliding property'':
\begin{align}\label{eq:braiding_naturality_diagram}
\begin{tikzpicture}
	\begin{pgfonlayer}{nodelayer}
		\node [style=none] (0) at (-1, -1) {};
		\node [style=none] (1) at (-1, 1) {};
		\node [style=none] (2) at (1.5, 1) {};
		\node [style=none] (3) at (1.5, -1) {};
		\node [style=none] (4) at (0, -0.25) {};
		\node [style=none] (5) at (0.5, 0.25) {};
		\node [style=none] (29) at (1.5, 1) {};
		\node [style=none] (30) at (1.5, -1) {};
		\node [style=none] (38) at (-1, 1) {};
		\node [style=none] (39) at (-1, -1) {};
		\node [style=none] (44) at (1.5, 1) {};
		\node [style=none] (45) at (1.5, -1) {};
		\node [style=none] (46) at (2, 1) {};
		\node [style=none] (47) at (2, -1) {};
		\node [style={system_label}] (48) at (-4, 1.75) {$\scriptstyle \sys  A$};
		\node [style={system_label}] (49) at (-4, -0.25) {$\scriptstyle \sys  C$};
		\node [style=none] (50) at (-1, -1) {};
		\node [style=none] (51) at (-1, 1) {};
		\node [style=rectangle] (52) at (-2.75, 1) {$\T{T}_1$};
		\node [style=rectangle] (53) at (-2.75, -1) {$\T{T}_2$};
		\node [style={system_label}] (54) at (-1.5, -0.25) {$\scriptstyle \sys  D$};
		\node [style={system_label}] (55) at (-1.5, 1.75) {$\scriptstyle \sys  B$};
		\node [style={system_label}] (56) at (2, 1.75) {$\scriptstyle \sys  D$};
		\node [style={system_label}] (57) at (2, -0.25) {$\scriptstyle \sys  B$};
		\node [style=none] (58) at (-4, 1) {};
		\node [style=none] (59) at (-4, -1) {};
	\end{pgfonlayer}
	\begin{pgfonlayer}{edgelayer}
		\draw [in=-180, out=0, looseness=0.75] (1.center) to (3.center);
		\draw [in=-135, out=0] (0.center) to (4.center);
		\draw [in=180, out=45] (5.center) to (2.center);
		\draw (44.center) to (46.center);
		\draw (45.center) to (47.center);
		\draw (52) to (51.center);
		\draw (53) to (50.center);
		\draw (58.center) to (52);
		\draw (59.center) to (53);
	\end{pgfonlayer}
\end{tikzpicture}=\begin{tikzpicture}
	\begin{pgfonlayer}{nodelayer}
		\node [style=none] (0) at (-1, -1) {};
		\node [style=none] (1) at (-1, 1) {};
		\node [style=none] (2) at (1.5, 1) {};
		\node [style=none] (3) at (1.5, -1) {};
		\node [style=none] (4) at (0, -0.25) {};
		\node [style=none] (5) at (0.5, 0.25) {};
		\node [style=none] (25) at (4.5, -1) {};
		\node [style=none] (26) at (4.5, 1) {};
		\node [style=rectangle] (27) at (3.25, 1) {$\T{T}_2$};
		\node [style=rectangle] (28) at (3.25, -1) {$\T{T}_1$};
		\node [style=none] (29) at (1.5, 1) {};
		\node [style=none] (30) at (1.5, -1) {};
		\node [style={system_label}] (31) at (4.5, 1.75) {$\scriptstyle \sys  D$};
		\node [style={system_label}] (32) at (4.5, -0.25) {$\scriptstyle \sys  B$};
		\node [style={system_label}] (33) at (2, 1.75) {$\scriptstyle \sys  C$};
		\node [style={system_label}] (34) at (2, -0.25) {$\scriptstyle \sys  A$};
		\node [style=none] (36) at (-1.5, 1) {};
		\node [style=none] (37) at (-1.5, -1) {};
		\node [style=none] (38) at (-1, 1) {};
		\node [style=none] (39) at (-1, -1) {};
		\node [style={system_label}] (40) at (-1.5, 1.75) {$\scriptstyle \sys  A$};
		\node [style={system_label}] (41) at (-1.5, -0.25) {$\scriptstyle \sys  C$};
	\end{pgfonlayer}
	\begin{pgfonlayer}{edgelayer}
		\draw [in=-180, out=0, looseness=0.75] (1.center) to (3.center);
		\draw [in=-135, out=0] (0.center) to (4.center);
		\draw [in=180, out=45] (5.center) to (2.center);
		\draw [in=180, out=0] (29.center) to (27);
		\draw (27) to (26.center);
		\draw (30.center) to (28);
		\draw (28) to (25.center);
		\draw (36.center) to (38.center);
		\draw (37.center) to (39.center);
	\end{pgfonlayer}
\end{tikzpicture}.
\end{align}
In the following, we will denote a braiding also by using the symbol $\T{S}$, identifying it with the family of events that defines $\Out{S}$. In the case where the members of the braiding satisfy $\T{S}_{\sys{A},\sys{B}}^{-1}=\T{S}_{\sys{B},\sys{A}}$ for all $\sys{A},\sys{B}\in\Sys{\Theta}$, the OPT is called \emph{symmetric}. Notice that all the theories developed so far are symmetric. In analogy to the cases of CT and QT, symmetricity has been assumed in the literature, for both physical and process theories.

\subsection{Probabilistic structure: linear properties of a theory}\label{subsec:probabilistic}
Tests $\Test{P}{X}=\lbrace p_x \rbrace_{x\in\Out{X}}\in\TestAB{I}{I}$ are probability distributions. For all $p\in\Event{I}{I}$ and all $\T{T}\in\Event{A}{B}$, one can define an operation, called \emph{multiplication by a scalar}, as follows:
\begin{align}\label{eq:scalar_multiplication}
\begin{aligned}
\Qcircuit @C=1em @R=1em
{
	p\;&\s{A}&\gate{\T{T}}&\s{B}\qw&
}
\end{aligned}
\coloneqq
\begin{aligned}
\Qcircuit @C=1em @R=1em
{
	&&p&&
	\\
	&\s{A}&\gate{\T{T}}&\s{B}\qw&
}
\end{aligned}
=
\begin{aligned}
\Qcircuit @C=1em @R=1em
{
	&\s{A}&\gate{\T{T}}&\s{B}\qw&
	\\
	&&p&&
}
\end{aligned}.
\end{align}
Notice that the equality can be shown to hold using the properties of parallel composition and of the trivial system. In the case of preparations $\rho$ and observations $a$, one can also show that $p\boxtimes\rho=\rho\circ p$ and $p\boxtimes a =  p\circ a$.
On the other hand, the following is required:
\begin{align}\label{eq:indep}
	 p\boxtimes q\coloneqq pq,\quad\forall p,q\in\Event{I}{I},
\end{align}
where $pq$ is the usual multiplication on real numbers. Moreover, an experimenter is allowed to perform a test $\TestC{T}{X}{A}{B}$
disregarding the different outcomes within an arbitrary subset $\Out{Y}\subseteq\Out{X}$. This amounts to
merging the events in $\Out{Y}$ into a single event. Such a possibility is captured by the notion of \emph{coarse-graining}. The probability of the coarse-grained event $\Out{Y}$ amounts to the sum of the probabilities
of all the outcomes in the subset $\Out{Y}$. Then, for each test $\TestC{T}{X}{A}{B}$ and every subset $\Out{Y}\subseteq\Out{X}$, there exists a coarse-grained event formally given by:
\begin{align}\label{eq:coarse_graining}
\T{T}_{\Out{Y}}\coloneqq\sum_{y\in\Out{Y}}\T{T}_y,
\end{align}
where sequential and parallel composition distribute over sums. The coarse-graining operation~\eqref{eq:coarse_graining} clearly boils down to the usual sum over real when $\sys{A}=\sys{I}=\sys{B}$.

Now, let $\Theta$ be an OPT. For all systems $\sys{A},\sys{B}\in\Sys{\Theta}$, we define the following equivalence relation for every $\T{T}_1,\T{T}_2\in\Event{A}{B}$: $\T{T}_1\sim\T{T}_2$ if for all $\sys{E}\in\Sys{\Theta},\rho\in\Event{I}{AE},a\in\Event{BE}{I}$ one has
\begin{align}\label{eq:equivalence_relation}
	\begin{aligned}
	\Qcircuit @C=1em @R=1em {
		&\multiprepareC{1}{\rho}&\s{A}\qw&\gate{\T{T}_1}&\s{B}\qw&\multimeasureD{1}{a}&
		\\
		&\pureghost{\rho}&\qw&\s{E}\qw&\qw&\ghost{a}&
	}
	\end{aligned}
	=
	\begin{aligned}
	\Qcircuit @C=1em @R=1em {
		&\multiprepareC{1}{\rho}&\s{A}\qw&\gate{\T{T}_2}&\s{B}\qw&\multimeasureD{1}{a}&
		\\
		&\pureghost{\rho}&\qw&\s{E}\qw&\qw&\ghost{a}&
	}
	\end{aligned}.
\end{align}
The members of the quotient class
\begin{align*}
\Transf{A}{B}\coloneqq\Event{A}{B}/\sim
\end{align*}
are called the \emph{transformations from $\sys{A}$ to $\sys{B}$}. Quotient classes of preparations $\St{A}\coloneqq\Transf{I}{A}$ are called \emph{the states of $\sys{A}$}, while those of observations $\Eff{A}\coloneqq\Transf{A}{I}$ \emph{the effects of $\sys{A}$}. Finally, equivalence classes of tests are called \emph{instruments}, and are collected by the quotient class $\InstrA{\Theta}$. We will often denote the states of $\sys{A}$ as $\rket{\rho}_{\sys{A}}$, and the effects of $\sys{A}$ as $\rbra{a}_\sys{A}$. In the case of the parallel composition of states or effects, we will safely omit the symbol $\boxtimes$ without loss of clarity---simply writing $\rket{\rho_1}_{\sys{A}_1}\rket{\rho_2}_{\sys{A}_2}$ or $\rbra{a_1}_{\sys{A}_1}\rbra{a_2}_{\sys{A}_2}$---while the sequential composition between a state and an effect will be given by the pairing: $\rbraket{a}{\rho}_\sys{A}\coloneqq\rbra{a}_\sys{A}\circ\rket{\rho}_\sys{A}\in[0,1]$.
\begin{remark}\label{rem:transformation_equiv_class}
	We observe that an arbitrary transformation $\T{T}\in\Transf{A}{B}$ is indeed defined by the entire class of transformations $\lbrace\T{T}\boxtimes\T{I}_\sys{E}\rbrace_{\sys{E}\in\Sys{\Theta}}$. This means that, given two transformations $\T{T}_1,\T{T}_2\in\Transf{A}{B}$, the following holds:
	\begin{align}\label{eq:identity_extended}
	\T{T}_1=\T{T}_2\quad \Longleftrightarrow\quad \T{T}_1\boxtimes\T{I}_\sys{E}=\T{T}_2\boxtimes\T{I}_\sys{E},\ \ \forall \sys{E}\in\Sys{\Theta}.
	\end{align}
\end{remark}
\begin{remark}\label{rem:OPT}
	Two tests associated with the same probability distribution in every possible experiment are said to be \emph{operationally equivalent}. Importantly, in the construction and characterization of an OPT, one is interested in equivalence classes of tests, namely in instruments. Thus, as it will be explicitly shown in Secs.~\ref{sec:constructing_an_opt} and~\ref{sec:BCT}, an OPT can be defined by specifying: (i) the systems $\Sys{\Theta}$, (ii) a parallel composition rule $\boxtimes$ for systems and states,
	and (iii) the instruments $\InstrA{\Theta}$ and their parallel composition $\boxtimes$.
\end{remark}
Now, say that we have a preparation with output system $\sys{A}$, an observation with input system $\sys{B}$, and in between a test from $\sys{A}$ to $\sys{B}$: this scenario corresponds to a joint probability distribution of the events that may occur in the experiment. More generally, the goal of a physical theory is to associate a probability with each event of a composite test as the above-described one. In particular:
\begin{align}\label{eq:probability}
\begin{aligned}
\Qcircuit @C=1em @R=1em {
	\prepareC{\rho_x}&\s{A}\qw&\gate{\T{T}_y}&\s{B}\qw&\measureD{a_z}
}
\end{aligned}\coloneqq
p\left(x,y,z\left.|\right. \srho_{\mathsf{X}},\Test{T}{Y},\Test{a}{Z} \right)\ ,
\end{align}
where $\rho_x,\T{T}_y,a_z$ are transformations of, respectively, the instruments $\srho_{\mathsf{X}}^{\sys{I}\!\rightarrow\!\sys{A}},\TestC{T}{Y}{A}{B},\TestC{a}{Z}{B}{I}$. 
The above definition shows the parametric dependence of the joint probability distribution of outcomes on the whole closed circuit representing a test $\sys I\!\to\!\sys I$. Along with condition~\eqref{eq:indep}, Eq.~\eqref{eq:probability} amounts to say that disconnected circuits represent statistically independent processes or experiments. 

By virtue of Eq.~\eqref{eq:equivalence_relation}, effects separate states and, \emph{vice versa}, states separate effects. This means that, for every pair of states $\rket{\rho}_\sys{A},\rket{\sigma}_\sys{A}\in\St{A}$ such that $\rket{\rho}_\sys{A}\neq\rket{\sigma}_\sys{A}$, there exists an effect $\rbra{a}_\sys{A}\in\Eff{A}$ such that $\rbraket{a}{\rho}_\sys{A}\neq\rbraket{a}{\sigma}_\sys{A}$ (and \emph{vice versa} for every pair of different effects). The latter amounts to say that, given two different states (or effects), there exists an experiment producing different statistics for them. States (effects) can be seen as functionals from effects (states) to
probabilities. Consequently---in the light of definitions~\eqref{eq:scalar_multiplication},~\eqref{eq:indep}, and~\eqref{eq:coarse_graining}---one can take linear combinations of them. In particular, the pairing between states and effects defines a complete class of linearly independent vectors in $\St{A}$, spanning a real vector space $\StR{A}\coloneqq\Span_{\mathbb{R}}{\St{A}}$. Similarly, $\Eff{A}$ is a class of non-negative linear functionals on $\St{A}$, and spans the dual space $\EffR{A}\coloneqq\Span_{\mathbb{R}}{\Eff{A}}=\StR{A}^{\vee}$. The members of $\StR{A}$ and $\EffR{A}$ are called, respectively, \emph{generalised states} and \emph{generalised effects}. The dimension $\D{A}\coloneqq\dim{\StR{A}}$ is called \emph{the size (or dimension) of the system $\sys{A}$}. In the following, every system $\sys{A}$ will be always thought as accompanied by its associated size $\D{A}$. A \emph{type of system} is an infinite collection of systems with a given size $D$. In an OPT, $\D{A}$ is the minimum number of probabilities that must be ascertained to determine the state of a system $\sys{A}$. For instance, in QT, for a system $\sys{A}$ whose associated Hilbert space has dimension $d_{\sys{A}}$, one has $\D{A}=d_{\sys{A}}^2$. In the present work we will restrict to the case of \emph{finite-dimensional} OPTs, namely theories $\Theta$ where $\D{A}<+\infty$ for all $\sys{A}\in\Sys{\Theta}$. In this case, for all systems $\sys{A}$ one has $\dim{\StR{A}}=\dim{\EffR{A}}$.

Every transformation $\T{T}$ from $\sys{A}$ to $\sys{B}$ maps $\St{A}$ to $\St{B}$.
Accordingly, transformations from $\sys{A}$ to $\sys{B}$ span a real vector space $\TransfR{A}{B}\coloneqq\Span_{\mathbb{R}}{\Transf{A}{B}}$, whose members are called \emph{generalised transformations}. Given $\sys{A},\sys{B},\sys{C},\sys{D}\in\Sys{\Theta}$, the transformations of the form $\T T\boxtimes\T T'\in\Transf{AC}{BD}$ are called {\em local transformations}. The transformations of the form $\T{T}\boxtimes\T{I}_\sys{C}\in\Transf{AC}{BC}$ and $\T{I}_\sys{A}\boxtimes\T{T}'\in\Transf{AC}{AD}$ are called, respectively, \emph{local transformations from the system $\sys{A}$ to the system $\sys{B}$} and \emph{from the system $\sys C$ to the system $\sys{D}$}. A transformation $\T{R}\in\Transf{A}{B}$ is called \emph{reversible} if there exists a transformation $\T{R}^{-1}\in\Transf{B}{A}$ such that $\T{R}^{-1}\T{R}=\T{I}_\sys{A}$ and $\T{R}\T{R}^{-1}=\T{I}_\sys{B}$. The class of reversible transformations of a theory $\Theta$ will be denoted by $\RevTransfA{\Theta}$.

Now, for all $p\in\St{I}$ and $\T{T}\in\Transf{A}{B}$, we will write the following:
\begin{align}
\begin{split}\label{eq:scalar_multiplication_consistency}
\begin{aligned}
\Qcircuit @C=1em @R=1.3em
{   
	&\s{A}&\gate{p\T{T}}&\s{B}\qw&
}
\end{aligned}
=
\begin{aligned}
\Qcircuit @C=1em @R=1.3em
{   
	&p&\s{A}&\gate{\T{T}}&\s{B}\qw&
}
\end{aligned},
\end{split}
\end{align}
where $p\T{T}$ represents the scalar multiplication on the vector space $\TransfR{A}{B}$. One has that $\StR{I}\cong\mathbb{R}$, implying $\D{I}=1$. Then, for all $p\in\St{I}$, $\sys{A},\sys{B}\in\Sys{\Theta}$, and $\T{T}\in\Transf{A}{B}$, one has: $\sys{I}\boxtimes\sys{A}=\sys{I}\otimes\sys{A}=\sys{A}=\sys{A}\boxtimes\sys{I}=\sys{A}\otimes\sys{I}$ and $p\boxtimes\T{T}=p\otimes\T{T}$---where $\otimes$ is the standard tensor product. The identity of the trivial system is clearly given by the unit, i.e., $\T{I}_\sys{I}=1$. Finally, the braiding $\T{S}$ satisfies the following relation:
\begin{align}\label{eq:braiding_identical}
\T{S}_{\sys{A},\sys{I}}=\T{I}_{\sys{A}}=\T{S}^{-1}_{\sys{A},\sys{I}},\quad \forall\sys{A}\in\Sys{\Theta}.
\end{align}

The transformations from a system $\sys{A}$ to a system $\sys{B}$ are contained in a convex and bounded subset of $\TransfR{A}{B}$~\cite{bookDCP2017}. Then, it is often convenient to consider the convex cone generated by the convex combinations of transformations in $\Transf{A}{B}$. The latter will be denoted by $\TransfC{A}{B}$---clearly, $\StC{A}$ and $\EffC{A}$ will be used in the case of the convex cones generated by, respectively, $\St{A}$ and $\Eff{A}$. A transformation $\T{T}\in\Transf{A}{B}$ is called \emph{atomic} if, given $\T{T}_{1},\T{T}_{2}\in\Transf{A}{B}$, one has the following implication:
\begin{align}\label{eq:atomicity}
\T{T}=\T{T}_{1}+\T{T}_{2}\implies\T{T}_{1}\propto \T{T}_{2}.
\end{align}
A transformation $\T{T}\in\Transf{A}{B}$ is called \emph{extremal} if, given $\T{T}_{1},\T{T}_{2}\in\Transf{A}{B}$ and $p\in(0,1)$, the condition $\T{T}=p\T{T}_{1}+(1-p)\T{T}_{2}$ implies $\T{T}_{1}=\T{T}_{2}$. On the other hand, we call \emph{a refinement of $\T{T}$} a set $\{\T{T}_i\}_{i\in\mathcal{R}}\subseteq\Transf{A}{B}$ such that $\T{T}=\sum_{i\in\mathcal{R}}\T{T}_{i}$, and similarly we define \emph{a convex refinement of $\T{T}$} a set $\{\T{T}_i\}_{i\in\mathcal{R}}\subseteq\Transf{A}{B}$ such that $\T{T}=\sum_{i\in\mathcal{R}}p_i\T{T}_{i}$ for some probability distribution $\{p_i\}_{i\in\mathcal{R}}$. Finally, $\Refi{T}$ and $\ConvRef{T}$ will denote the union of, respectively, all the refinements and all the convex refinements of $\T{T}$.

We assume that $0\in\Transf{I}{I}$, so that events associated with a zero probability can be considered in a theory. Accordingly, the existence of parallel composition implies that the null generalised transformation $\varepsilon_{\sys{A}\to\sys{B}}\in\TransfR{A}{B}$, such that $\rbra{a}_{\sys{BE}}\left(\varepsilon_{\sys{A}\to\sys{B}}\boxtimes\T{I}_{\sys{E}}\right)\rket{\rho}_{\sys{AE}}=0$ for every $\sys{E}\in\Sys{\Theta},\rket{\rho}_{\sys{AE}}\in\St{AE},\rbra{a}_{\sys{BE}}\in\Eff{BE}$, is included in the transformations $\Transf{A}{B}$ for all $\sys{A},\sys{B}$. A null transformation always occurs with null marginal probability. A transformation $\T{T}\in\Transf{A}{B}$ such that there exists a singleton instrument $\Instr{T}{\star}=\{\T{T}\}\in\InstrAB{A}{B}$ is called \emph{deterministic}. In continuity with the past literature, we will often call \emph{channels} those deterministic transformations which are not states or effects. Oppositely to the case of null transformations, a deterministic physical process happens with certainty, i.e., with marginal probability $1$. For instance, a state is deterministic if and only if it gives probability 1 on every deterministic effect, or, in other words, if and only if it is normalized. We will denote by $\TransfN{A}{B}$ the class of deterministic transformations from system $\sys{A}$ to system $\sys{B}$---clearly, $\StN{A}$ and $\EffN{A}$ will be used in the case of, respectively, deterministic states and effects of system $\sys{A}$. For any given instrument $\TestC{T}{X}{A}{B}$, one has clearly that $\sum_{x\in\mathsf{X}}\T{T}_x\in\TransfN{A}{B}$, and, as a particular case, $\TransfN{I}{I}=\{\T{I}_{\sys{I}}\}=\{1\}$. In the light of the definition of null transformation, we will assume that $\TestC{T}{X}{A}{B}\in\InstrA{\Theta}$ if and only if $\varepsilon_{\sys{A}\to\sys{B}}\cup\TestC{T}{X}{A}{B}\in\InstrA{\Theta}$. The reason why this is a convenient assumption is the following characterisation of the deterministic transformations. If $\T{T}\in\Transf{A}{B}$ is a deterministic transformation, $\T{T}\boxtimes\T{I}_{\sys{E}}$ clearly maps $\StN{AE}$ to $\StN{BE}$ for all $\sys{E}$.
Conversely, let $\T{T}\in\Transf{A}{B}$ be such that $\T{T}\boxtimes\T{I}_{\sys{E}}$ maps $\StN{AE}$ to $\StN{BE}$ for all $\sys{E}$. Suppose now that there exists an instrument $\InstrC{T}{K\cup\star}{A}{B}$ of the form $\InstrC{T}{K\cup\star}{A}{B}=\{\T{T}_k\}_{k\in\mathsf{K}}\cup\{\T{T}\}$ with $\T{T}_k\in\Transf{A}{B}$ for every $k\in\mathsf{K}$. Let $\{\rket{\rho_j}_{\sys{AE}}\}_{j\in\mathsf{J}}$ and $\{\rbra{a_l}_{\sys{BE}}\}_{l\in\mathsf{L}}$ be, respectively, a preparation-instrument of $\sys{AE}$ and an observation-instrument of $\sys{BE}$. By coarse-graining, $\rket{\rho}_{\sys{AE}}\coloneqq\sum_{j\in\mathsf{J}}\rket{\rho_j}_{\sys{AE}}$ and $\rbra{\tilde{e}}_{\sys{B}}\coloneqq\sum_{l\in\mathsf{L}}\rbra{a_l}_{\sys{BE}}$ are deterministic. Then, by hypothesis, $\left(\T{T}\boxtimes\T{I}_{\sys{E}}\right)\rket{\rho}_{\sys{AE}}\in\StN{BE}$ also holds. Accordingly, $\TransfN{I}{I}\ni\rbra{\tilde{e}}_{\sys{BE}}\left(\T{T}\boxtimes\T{I}_{\sys{E}}\right)\rket{\rho}_{\sys{AE}}=1$. On the other hand, it must also be
\begin{align*}
\rbra{\tilde{e}}_{\sys{BE}}\left(\T{T}\boxtimes\T{I}_{\sys{E}}\right)\rket{\rho}_{\sys{AE}} + \sum_{k\in\mathsf{K}}\rbra{\tilde{e}}_{\sys{BE}}\left(\T{T}_k\boxtimes\T{I}_{\sys{E}}\right)\rket{\rho}_{\sys{AE}} = 1,
\end{align*}
which in turn implies:
\begin{align*}
\sum_{l\in\mathsf{L},k\in\mathsf{K},j\in\mathsf{J}}\rbra{a_l}_{\sys{BE}}\left(\T{T}_k\boxtimes\T{I}_{\sys{E}}\right)\rket{\rho_j}_{\sys{AE}} = 0,
\end{align*}
with $\rbra{a_l}_{\sys{BE}}\left(\T{T}_k\boxtimes\T{I}_{\sys{E}}\right)\rket{\rho_j}_{\sys{AE}}\geq 0$ for every $l\in\mathsf{L},k\in\mathsf{K},j\in\mathsf{J}$ by definition. Accordingly, since the previous argument does not depend on the choice of instruments $\{\rket{\rho_j}_{\sys{AE}}\}_{j\in\mathsf{J}}$ and $\{\rbra{a_l}_{\sys{BE}}\}_{l\in\mathsf{L}}$, we conclude that it must be $\T{T}_k=\varepsilon_{\sys{A}\to\sys{B}}$ for all $k\in\mathsf{K}$, namely, $\T{T}\in\TransfN{A}{B}$.
Thus, a transformation $\T{T}\in\Transf{A}{B}$ is deterministic if and only if $\T{T}\boxtimes\T{I}_{\sys{E}}$ maps $\StN{AE}$ to $\StN{BE}$ for all systems $\sys{E}$. By the same argument as above, one can show that, equivalently, a transformation $\T{T}\in\Transf{A}{B}$ is deterministic if and only if $\T{T}\boxtimes\T{I}_{\sys{E}}$ maps $\EffN{AE}$ to $\EffN{BE}$ for all systems $\sys{E}$. The property of being deterministic is clearly preserved under both sequential and parallel composition.

As a first consequence, it is clear why in Subsec.~\ref{subsec:operational} the identity process $\T{I}$ has been, by definition, associated with a family of singleton tests. In the remainder of this work, we will make extensive use of the above characterisation, which we now use in order to prove the following useful result.
\begin{proposition}\label{prop:pure_to_pure}
	Let $\Theta$ be an OPT, and $\T{R}\in\RevTransf{A}{A'}$ be a reversible transformation from a system $\sys{A}\in\Sys{\Theta}$ to a system $\sys{A}'\in\Sys{\Theta}$. Then $\T{R}\in\TransfN{A}{A'}$, and preserves both atomicity and extremality via sequential composition. In particular, $\RevTransf{A}{A}$ is a group of permutations on $\PurSt{A}$.
\end{proposition}
\begin{proof}
	By hypothesis, there exists $\T{R}^{-1}\in\RevTransf{A'}{A}$ such that $\T{R}^{-1}\circ\T{R}=\T{I}_{\sys{A}}$ holds. By definition, there also exist $\T{D}\in\TransfN{A}{A'}$ and $\T{D}'\in\TransfN{A'}{A}$ such that $\mathsf{T}\coloneqq\{\T{R},\T{D}-\T{R}\}$ and $\mathsf{T}'\coloneqq\{\T{R}^{-1},\T{D}'-\T{R}^{-1}\}$ are instruments of $\Theta$. Their sequential composition $\mathsf{T}'\circ\mathsf{T} = \{\T{I}_{\sys{A}},\T{R}^{-1}\T{D}-\T{I}_{\sys{A}},\T{D}'\T{R}-\T{I}_{\sys{A}},\T{D}'\T{D}-\T{D}'\T{R}-\T{R}^{-1}\T{D}+\T{I}_{\sys{A}}\}$ is also an instrument. Since $\T{I}_{\sys{A}}$ is associated with a singleton instrument by definition (see Subsec.~\ref{subsec:operational}), $\T{I}_{\sys{A}}\in\TransfN{A}{A}$ holds. Accordingly, by direct inspection of $\mathsf{T}'\circ\mathsf{T}$ and using the characterisation of the deterministic transformations, one must have $\T{D}=\T{R}$, namely $\T{R}\in\TransfN{A}{A'}$. We now just prove that if $\T{A}\in\Transf{C}{A}$ is atomic, then $\T{R}\T{A}$ is also atomic. The case of \emph{$\T{B}\in\Transf{A'}{C}$ atomic implies $\T{B}\T{R}$ atomic} will be then a straightforward adaptation of the former case; on the other hand, preservation of extremality via sequential composition can be proven in an analogous manner, just by using the definition of extremal transformation instead of that of atomic one. Let then $\T{A}\in\Transf{C}{A}$ be an atomic transformation. By contradiction, suppose that $\T{R}\T{A}$ is not atomic. Then, by definition, there exist $\T{A}_1,\T{A}_2\in\Transf{C}{A}$ with $\T{A}_1\not\propto\T{A}_2$ such that:
	\begin{align*}
	\T{R}\T{A} = 
	\T{A}_1+\T{A}_2\ \Longrightarrow\ \T{A} =  \T{R}^{-1}\T{A}_1+ \T{R}^{-1}\T{A}_2.
	\end{align*}
	However, by atomicity of $\T{A}$, the above equation implies $\T{R}^{-1}\T{A}_1\propto\T{R}^{-1}\T{A}_2$, i.e., $\T{A}_1\propto\T{A}_2$, that is absurd.
\end{proof}
By Proposition~\ref{prop:pure_to_pure}, it is clear why in Subsec.~\ref{subsec:operational} also the braiding $\T{S}$ has been, by definition, associated with a family of singleton tests.

When $\St{A}$ coincides with 
its convex hull for every $\sys{A}\in\Sys{\Theta}$, the theory $\Theta$ is called \emph{convex}. We will denote by $\ExtStZ{A}$ the set of extremal points of the convex hull of $\St{A}$. Since $\rket{\varepsilon}_\sys{A}\in \ExtStZ{A}$ for every system $\sys{A}$, we also define the set of non-null extremal points $\ExtSt{A}\coloneqq\ExtStZ{A}\setminus \lbrace\rket{\varepsilon}_\sys{A}\rbrace$. The deterministic extremal states are those historically called \emph{the pure states}. Deterministic states which are not extremal are those historically called \emph{mixed states}. More generally, we will call \emph{mixed} those states which are neither 
extremal nor atomic. This nomenclature is extended to arbitrary transformations: accordingly, \emph{extremal channels} will be called \emph{pure (transformations)}. We will denote by $\PurSt{A}\subseteq\ExtSt{A}$ the set of pure states of a system $\sys{A}$. In a convex theory, every mixture of pure states can be \emph{deterministically prepared}. CT and QT are examples of convex theories.
\begin{remark}\label{rem:atomicity}
	One may argue that \emph{atomic}
	should be used as a synonym of \emph{pure} (see e.g.~Ref.~\cite{chiribella2015operational}). However, in general we keep the two notions of atomicity and purity distinct. The only case where there is no need for a distinction between atomicity and purity, at least for states, is that of theories where every state is proportional to a
	deterministic one.
	Moreover, the theory presented in this work provides an explicit motivation to take the two notions distinct. Indeed, as it will be discussed in Sec.~\ref{sec:features}, in the case of transformations of the theory which are not states or effects, the property of purity is preserved by parallel composition, while the one of atomicity is not.
\end{remark}
Figs.~\ref{fig:convex1} and~\ref{fig:convex2} illustrate the main notions related to atomicity and convexity introduced in the present subsection.

\begin{figure}
	\begin{center}
	\includegraphics[scale=0.32]{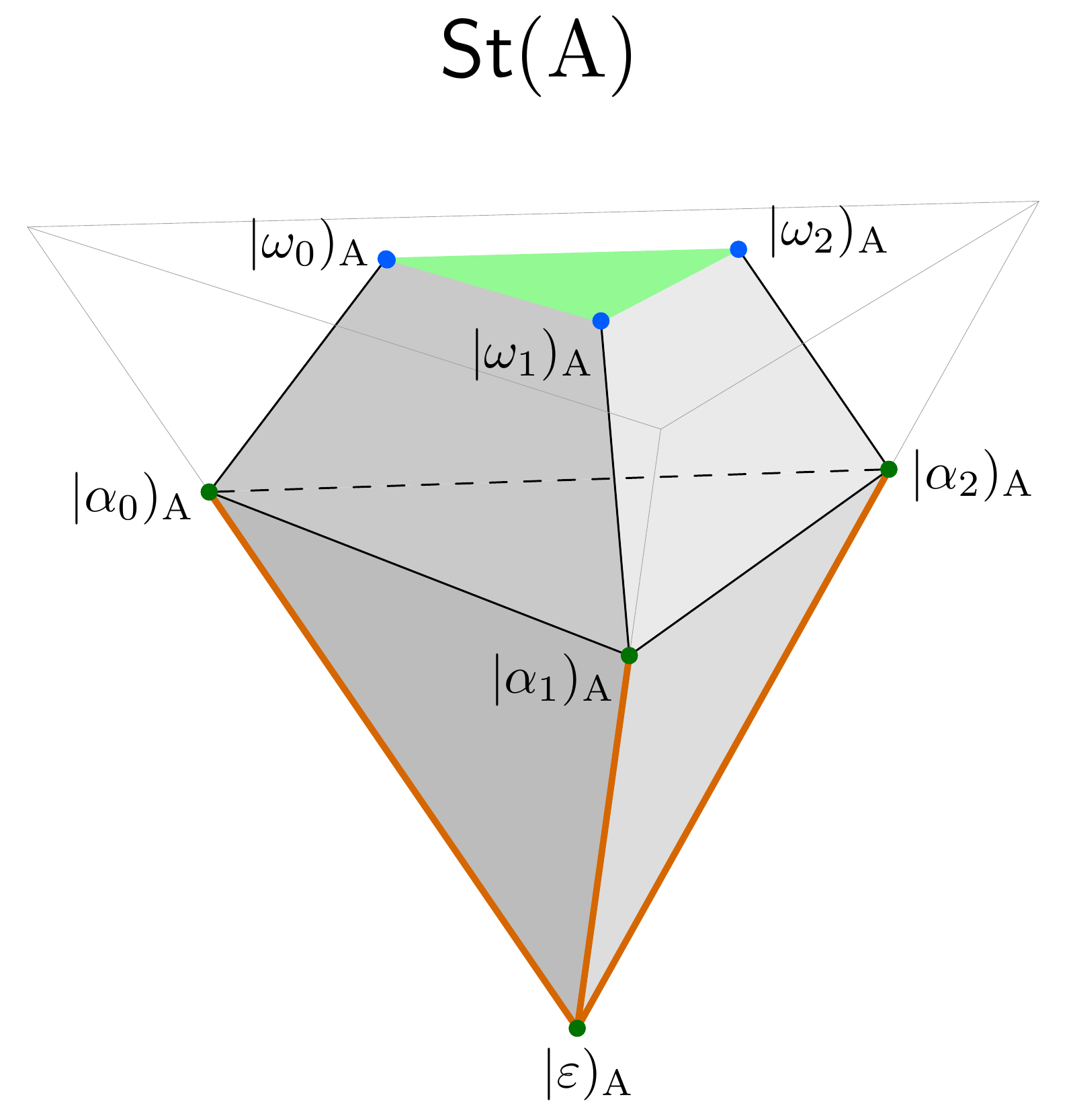}
    \end{center}
	\caption{Example of a set of states for a system $\sys{A}$ of size $3$. $\St{A}$ is convex and contained in the complete state-space of a classical trit---namely, a tetrahedron. The elements depicted in
	orange are the atomic ones, while those in dark green---the null state $\rket{\varepsilon}_{\sys{A}}$ and $\rket{\alpha_i}_{\sys{A}}$ for $i\in\{0,1,2\}$---are both atomic and extremal. The elements in blue---$\rket{\omega_i}_{\sys{A}}$ for $i\in\{0,1,2\}$---are the pure states of $\St{A}$, which are extremal but not atomic. The elements in light green---i.e.~those contained in the upper face---are the deterministic mixed states of $\sys{A}$. Finally, all the remaining states (in grey) are mixed.}
	\label{fig:convex1}
\end{figure}
\begin{figure}
	\begin{center}
	\includegraphics[scale=0.49]{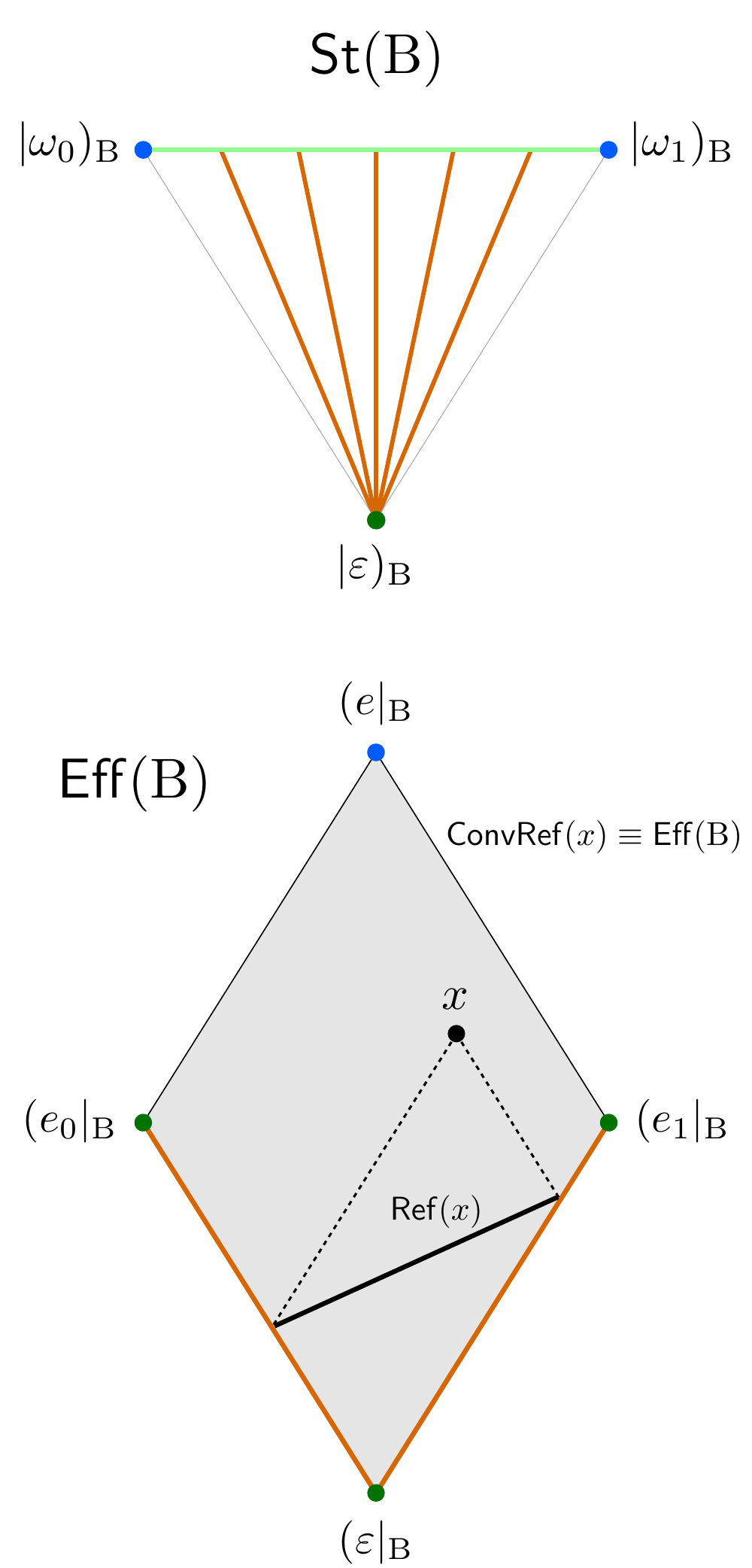}
	\end{center}
	\caption{Example of sets of states and effects of a system $\sys{B}$ of size $2$. $\St{B}$ is contained in the complete state-space of a classical bit---namely, a triangle---and is not convex. The colours are associated to the same properties as in Fig.~\ref{fig:convex1}. The atomic rays of $\St{B}$ provide the example of atomic elements which are not proportional to extremal ones. One has $\left( e_i | \omega_j \right)_{\sys{B}}=\delta_{ij}$ for every $i,j\in\{0,1\}$. $\Eff{B}$ contains a unique deterministic effect---namely, $\rbra{e}_{\sys{B}}=\rbra{e_0}_{\sys{B}}+\rbra{e_1}_{\sys{B}}$---which is pure and has a trivial convex refinement set. The refinement and convex refinement sets of a mixed (internal) point $x\in\Eff{B}$ are depicted.}
	\label{fig:convex2}
\end{figure}

\section{\textbf{The operational probabilistic structure: preliminary definitions and results}}\label{sec:structure_consistency}

In the present section, we will review some important properties in the scope of operational probabilistic theories, also proving a first relevant result related to the so-called \emph{bilocal-tomographic theories}. Then, we will set and discuss the notions of no-restriction hypothesis and causality in the OPT framework. Finally, we will focus on the family of \emph{simplicial theories}, which comprises the probabilistic theory introduced in the present work.

\subsection{Parallel composition and properties of $n$-local-tomographic theories}\label{subsec:preliminary_results}
We first recall some definitions and properties which are relevant to the present work. Let us start by defining the following map:
\begin{align*}
f\colon(\sys{A},\sys{B})\mapsto \D{AB}\coloneqq f(\sys{A},\sys{B}).
\end{align*}
The above map $f$ defines a rule for the size of composite systems. The following must always hold for all $\sys{A},\sys{B},\sys{C}\in\Sys{\Theta}$:
\begin{align}
&f(\sys I,\sys{A})=\D{IA}=\D{A}=\D{AI}=f(\sys{A},\sys I),\label{eq:fI}\\
&f(\sys{AB},\sys C)=f(\sys A,\sys{BC})\label{eq:fass}
\end{align}
Furthermore, since a theory is endowed with a braiding, one has the additional constraint:
\begin{align}
f(\sys A,\sys B)=f(\sys B,\sys A).\label{eq:fbraid}
\end{align}
We stress that when $f(\sys A,\sys B)=g(D_\sys A,D_\sys B)$, i.e., the dimension 
of the composite system depends only on the dimensions of the components, the map $f$ boils down to an operation $g$ on the dimensions of systems.
In this case, the above conditions~\eqref{eq:fI}, \eqref{eq:fass} and~\eqref{eq:fbraid} take the form:
\begin{align*}
&g(D_\sys A,1)=g(1,D_\sys A)=D_\sys A,\\
&g(g(D_{\sys A},D_{\sys B}),D_{\sys C})=g(D_{\sys A},g(D_{\sys B},D_{\sys C})),\\
&g(D_\sys A,D_\sys B)=g(D_\sys B,D_\sys A).
\end{align*}
In general, one has the following inequality~\cite{d2019classical}:
\begin{align}\label{eq:dimAB_geq}
\D{AB}\geq\D{A}\D{B},\quad \forall\sys{A},\sys{B}\in\Sys{\Theta}.
\end{align}
\begin{property}[$n$-local discriminability~\cite{d2019classical}]\label{prope:n_local}
	Let $n\leq m$. The effects obtained as a conical combination of the parallel compositions of effects $a_1,a_2,\ldots,a_l$, where $a_j$ is $k_j$-partite with $k_j\leq n$ for all values of $j$, are separating for the $m$-partite states.
\end{property}
The above property was introduced in Ref.~\cite{hardy2012limited}, where it was named \emph{$n$-local tomography}. The case of $n$-local discriminability for $n=1$ is called \emph{local discriminability}. A theory $\Theta$ satisfies local discriminability if and only if
\begin{align}\label{eq:dimensions_local}
\D{AB}=\D{A}\D{B},\quad \forall \sys{A},\sys{B}\in\Sys{\Theta}
\end{align}
(see Ref.~\cite{bookDCP2017} for a proof). In this case, for every composite system $\sys{AB}$, the vector space $\StR{AB}$ is spanned by the \emph{local (or product) states}, given by the parallel composition of single-system states. This is the case e.g.~of QT. Given two systems $\sys{A}$ and $\sys{B}$, the \emph{separable states} of the bipartite system $\sys{AB}$ are those of the form:
\begin{align}\label{eq:separable_states}
&\rket{\sigma}_{\sys{AB}} = \sum_{i\in I} \rket{\alpha_i}_{\sys{A}}\rket{\beta_i}_{\sys{B}},
\end{align}
where $\St{A}\ni\rket{\alpha_i}_\sys A\neq\rket\varepsilon_\sys A$, $\St{B}\ni\rket{\beta_i}_\sys B\neq\rket\varepsilon_\sys B$ for every $i\in I$. By negation, the \emph{entangled states} are those that are \emph{not separable}. Notice that the above definition is straightforwardly generalised to arbitrary transformations. If a theory does not satisfy local discriminability, then it necessarily has entangled states~\cite{d2019classical}. The converse is not true: QT is the example of a theory having entangled states and satisfying local discriminability.

Let now $\Theta$ be an OPT. In the light of Eq.~\eqref{eq:dimAB_geq}, we can define the non-negative integer quantity:
\begin{align*}
\Delta^{(2)}_{\sys{AB}}\coloneqq \D{AB}-\D{A}\D{B},\quad \forall \sys{A},\sys{B}\in\Sys{\Theta}.
\end{align*}
The ``excess'' $\Delta^{(2)}_{\sys{AB}}$ is the dimension of any subspace $\mathcal{B}_{\sys{AB}}$  of $\StR{AB}$ which is linearly independent of the span of the separable states---denoted by $\mathcal{S}_{\sys{AB}}$---and such that every element of $\StR{AB}$ can be uniquely decomposed as the sum of an element of $\mathcal{S}_{\sys{AB}}$ and one of $\mathcal{B}_{\sys{AB}}$. In formula:
\begin{align*}
\StR{AB}=\mathcal S_{\sys{AB}}\oplus\mathcal B_\sys{AB}.
\end{align*}
From now on, we will choose one space $\mathcal{B}_{\sys{AB}}$, where $\Delta^{(2)}_{\sys{AB}}=\dim \mathcal{B}_{\sys{AB}}$. For the sake of clarity, we will adopt the following convenient diagrammatic notation:
\begin{align*}
\begin{tikzpicture}
	\begin{pgfonlayer}{nodelayer}
		\node [style=none] (0) at (0.25, 0.75) {};
		\node [style=none] (1) at (-0.25, 0.75) {};
		\node [style=none] (2) at (0.25, -0.75) {};
		\node [style=none] (3) at (-0.25, -0.75) {};
	\end{pgfonlayer}
	\begin{pgfonlayer}{edgelayer}
		\draw [style=state] (1.center) to (0.center);
		\draw [style=state] (3.center) to (2.center);
	\end{pgfonlayer}
\end{tikzpicture}\ \coloneqq\mathcal{S}_{\sys{AB}},\quad   \begin{tikzpicture}
	\begin{pgfonlayer}{nodelayer}
		\node [style=none] (5) at (0.25, -0.75) {};
		\node [style=none] (6) at (0.25, 0.75) {};
	\end{pgfonlayer}
	\begin{pgfonlayer}{edgelayer}
		\draw [bend right=90, looseness=1.25] (6.center) to (5.center);
	\end{pgfonlayer}
\end{tikzpicture}
\ \coloneqq\mathcal{B}_{\sys{AB}}.
\end{align*}
For any tripartite system $\sys{ABC}$, we can introduce the subspaces:
\begin{align}\label{eq:separable_tripartite}
&\mathcal{S}_{\sys{ABC}}^{(1)}\coloneqq\begin{tikzpicture}
	\begin{pgfonlayer}{nodelayer}
		\node [style=none] (0) at (0.25, 1) {};
		\node [style=none] (1) at (-0.25, 1) {};
		\node [style=none] (2) at (0.25, 0) {};
		\node [style=none] (3) at (-0.25, 0) {};
		\node [style=none] (4) at (0.25, -1) {};
		\node [style=none] (5) at (-0.25, -1) {};
	\end{pgfonlayer}
	\begin{pgfonlayer}{edgelayer}
		\draw [style=state] (1.center) to (0.center);
		\draw [style=state] (3.center) to (2.center);
		\draw [style=state] (5.center) to (4.center);
	\end{pgfonlayer}
\end{tikzpicture}\ ,\  \mathcal{S}_{\sys{ABC}}^{(2)}\coloneqq\begin{tikzpicture}
	\begin{pgfonlayer}{nodelayer}
		\node [style=none] (0) at (0.25, 0) {};
		\node [style=none] (1) at (0.25, 1) {};
		\node [style=none] (2) at (0.25, -1) {};
		\node [style=none] (3) at (-0.25, -1) {};
	\end{pgfonlayer}
	\begin{pgfonlayer}{edgelayer}
		\draw [bend right=90, looseness=2.00] (1.center) to (0.center);
		\draw [style=state] (3.center) to (2.center);
	\end{pgfonlayer}
\end{tikzpicture}\ ,\ \mathcal{S}_{\sys{ABC}}^{(3)}\coloneqq\begin{tikzpicture}
	\begin{pgfonlayer}{nodelayer}
		\node [style=none] (0) at (0.25, -1) {};
		\node [style=none] (1) at (0.25, 0) {};
		\node [style=none] (2) at (0.25, 1) {};
		\node [style=none] (3) at (-0.25, 1) {};
	\end{pgfonlayer}
	\begin{pgfonlayer}{edgelayer}
		\draw [bend right=90, looseness=2.00] (1.center) to (0.center);
		\draw [style=state] (3.center) to (2.center);
	\end{pgfonlayer}
\end{tikzpicture}\ ,\ 
\mathcal{S}_{\sys{ABC}}^{(4)}\coloneqq\begin{tikzpicture}
	\begin{pgfonlayer}{nodelayer}
		\node [style=none] (0) at (0.25, 0) {};
		\node [style=none] (1) at (-0.25, 0) {};
		\node [style=none] (2) at (0.25, 1) {};
		\node [style=none] (3) at (-0.5, 0) {};
		\node [style=none] (4) at (0.25, -1) {};
	\end{pgfonlayer}
	\begin{pgfonlayer}{edgelayer}
		\draw [style=state] (1.center) to (0.center);
		\draw [in=90, out=-165] (2.center) to (3.center);
		\draw [in=-90, out=165] (4.center) to (3.center);
	\end{pgfonlayer}
\end{tikzpicture}\ .
\end{align}
Let us then define $I\coloneqq \lbrace1,2,3,4\rbrace$, and
\begin{align*}
\mathcal S^{(I)}_\sys{ABC}\coloneqq\Span_{\mathbb R}\bigcup_{i\in I}\mathcal S^{(i)}_\sys{ABC}.
\end{align*}
Then $\StR{ABC}$ is completely spanned by $\mathcal S^{(I)}_\sys{ABC}$ along with any subspace 
$\mathcal{T}_{\sys{ABC}}$ of $\StR{ABC}$ which is linearly independent of 
$\mathcal S^{(I)}_\sys{ABC}$, and such that every element of $\StR{ABC}$ can be uniquely 
decomposed as the sum of an element of 
$\mathcal{S}_{\sys{ABC}}^{(I)}$ and one of $\mathcal{T}_{\sys{ABC}}$. In formula:
\begin{align*}
\StR{ABC}=\mathcal S^{(I)}_\sys{ABC}\oplus\mathcal T_\sys{ABC}.
\end{align*}
We will arbitrarily 
choose $\mathcal{T}_{\sys{ABC}}$ once for all, and shall denote it diagrammatically as follows:
\begin{align*}
\begin{tikzpicture}
	\begin{pgfonlayer}{nodelayer}
		\node [style=none] (52) at (0.25, 0) {};
		\node [style=none] (54) at (0.25, 1) {};
		\node [style=none] (55) at (-0.5, 0) {};
		\node [style=none] (56) at (0.25, -1) {};
	\end{pgfonlayer}
	\begin{pgfonlayer}{edgelayer}
		\draw [in=90, out=-165] (54.center) to (55.center);
		\draw [in=-90, out=165] (56.center) to (55.center);
		\draw (52.center) to (55.center);
	\end{pgfonlayer}
\end{tikzpicture}\ \coloneqq\mathcal{T}_{\sys{ABC}}.
\end{align*}
The dimension of $\mathcal T_{\sys{ABC}}$ will be denoted by $\Delta^{(3)}_{\sys{ABC}}\coloneqq\dim \mathcal{T}_{\sys{ABC}}$. We are now in position to prove the following general result.
\begin{theorem}[\textbf{Composition rules for system-sizes in arbitrary theories}]\label{thm:delta3}
Let $\Theta$ be an OPT, and $\sys{ABC}$ any tripartite system in $\Theta$. Then the following 
identity holds:
\begin{align}
\begin{split}\label{eq:dimensions_tripartite}
\D{ABC} = &\D{A}\D{B}\D{C} + \Delta^{(2)}_\sys{AB} \D{C} + \Delta^{(2)}_\sys{BC} \D{A}+ \Delta^{(2)}_\sys{AC} \D{B}+ \\ &+
\Delta^{(3)}_{\sys{ABC}}.
\end{split}
\end{align}
\end{theorem}
\begin{proof}
From linear independence of $\mathcal{S}_{\sys{AB}}$ and $\mathcal{B}_{\sys{AB}}$, combined with the fact that effects separate states, we conclude that for every $\rket\sigma_{\sys{AB}}\in\mathcal S_{\sys{AB}}$ and for every $\rket{\beta}_{\sys{AB}}\in\mathcal B_{\sys{AB}}$
there exist generalised effects $\rbra{b}_\sys{AB},\rbra{b'}_\sys{AB}\in\EffR{AB}$ such that $\rbraket{b}{\sigma}_{\sys{AB}}=1$ and $\rbraket{b}{\beta'}_{\sys{AB}}=0$ for all $\rket{\beta'}_{\sys{AB}}\in\mathcal B_{\sys{AB}}$, and similarly
$\rbraket{b'}{\beta}_{\sys{AB}}=1$ and $\rbraket{b'}{\sigma'}_{\sys{AB}}=0$ for all $\rket{\sigma'}_{\sys{AB}}\in\mathcal S_{\sys{AB}}$. Using the above property,
combined with the facts that every $\rket{\sigma}_{\sys{AB}}\in\mathcal{S}_{\sys{AB}}$ is a 
linear combination of separable states, and that states separate effects, one has also that: 
$\rbra{b'}_{\sys{AB}}\rket{\rho}_{\sys{A}}=\rbra{\varepsilon}_{\sys{B}}$ and $\rbra{b'}_{\sys{AB}}\rket{\rho'}_{\sys{B}}=\rbra{\varepsilon}_{\sys{A}}$ for 
all $\rket{\rho}_{\sys{A}}\in\StR{A},\rket{\rho'}_{\sys{B}}\in\StR{B}$. Moreover, by 
construction, $\rbra{b'}_{\sys{AB}}$ cannot be a combination of separable effects. 
Diagrammatically:
\begin{align}
\begin{split}\label{eq:bipartite_effect_zero}
&\begin{tikzpicture}
	\begin{pgfonlayer}{nodelayer}
		\node [style=none] (0) at (-0.25, -0.75) {};
		\node [style=none] (1) at (-0.25, 0.75) {};
		\node [style=none] (2) at (-0.25, -0.75) {};
		\node [style=none] (4) at (-0.25, 0.75) {};
	\end{pgfonlayer}
	\begin{pgfonlayer}{edgelayer}
		\draw [style=dashed, bend left=90, looseness=1.25] (1.center) to (0.center);
	\end{pgfonlayer}
\end{tikzpicture}\coloneqq\rbra{b'}_{\sys{AB}},\quad
\begin{tikzpicture}
	\begin{pgfonlayer}{nodelayer}
		\node [style=none] (1) at (0, -0.75) {};
		\node [style=none] (2) at (0, 0.75) {};
		\node [style=none] (3) at (0, -0.75) {};
		\node [style=none] (4) at (-0.5, -0.75) {};
		\node [style=none] (5) at (0, 0.75) {};
		\node [style=none] (6) at (-0.5, 0.75) {};
	\end{pgfonlayer}
	\begin{pgfonlayer}{edgelayer}
		\draw [style=dashed, bend left=90, looseness=1.25] (2.center) to (1.center);
		\draw [style=state] (4.center) to (3.center);
		\draw [style=state] (6.center) to (5.center);
	\end{pgfonlayer}
\end{tikzpicture}=0,\\ \ \\
&\implies\  \begin{tikzpicture}
	\begin{pgfonlayer}{nodelayer}
		\node [style=none] (0) at (0.25, -0.75) {};
		\node [style=none] (1) at (0.25, 0.75) {};
		\node [style=none] (2) at (0.25, -0.75) {};
		\node [style=none] (4) at (0.25, 0.75) {};
		\node [style=none] (5) at (-0.25, 0.75) {};
	\end{pgfonlayer}
	\begin{pgfonlayer}{edgelayer}
		\draw [style=dashed, bend left=90, looseness=1.25] (1.center) to (0.center);
		\draw [style=state] (5.center) to (4.center);
	\end{pgfonlayer}
\end{tikzpicture}=\rbra{\varepsilon}_{\sys{B}},\quad \begin{tikzpicture}
	\begin{pgfonlayer}{nodelayer}
		\node [style=none] (0) at (0, -0.75) {};
		\node [style=none] (1) at (0, 0.75) {};
		\node [style=none] (2) at (0, -0.75) {};
		\node [style=none] (3) at (-0.5, -0.75) {};
		\node [style=none] (4) at (0, 0.75) {};
	\end{pgfonlayer}
	\begin{pgfonlayer}{edgelayer}
		\draw [style=dashed, bend left=90, looseness=1.25] (1.center) to (0.center);
		\draw [style=state] (3.center) to (2.center);
	\end{pgfonlayer}
\end{tikzpicture}=\rbra{\varepsilon}_{\sys{A}}.
\end{split}
\end{align}
The four subspaces $\mathcal{S}_{\sys{ABC}}^{(i)}$ for $i\in I=\{1,2,3,4\}$ are spanned by the
(tripartite) product states, and their dimensions are given by:
\begin{align*}
	&\dim\mathcal{S}_{\sys{ABC}}^{(1)}=\D{A}\D{B}\D{C},\ \dim\mathcal{S}_{\sys{ABC}}^{(2)}=\Delta^{(2)}_\sys{AB}\D{C},\\ &\dim\mathcal{S}_{\sys{ABC}}^{(3)}=\Delta^{(2)}_\sys{BC}\D{A},\ \dim\mathcal{S}_{\sys{ABC}}^{(4)}=\Delta^{(2)}_\sys{AC}\D{B} .
\end{align*}
Now, it is easy to show that the five 
subspaces $\mathcal{S}_{\sys{ABC}}^{(i)}$ for $i\in I$ and $\mathcal T_\sys{ABC}$ are 
separated by effects, and then in fact they are linearly independent. This can be done using the following argument. Let us pick a complete linearly independent set 
$\{\rket{\alpha^{(5)}_j}_\sys{ABC}\}_{j\in J_5}$ in $\mathcal T_\sys{ABC}$, where $J_5\coloneqq\{1,2,\ldots,\Delta^{(3)}_\sys{ABC}\}$.
By linear independence of $\mathcal T_\sys{ABC}$ and 
$\mathcal S^{(I)}_\sys{ABC}$, combined with the fact that effects separate states,
one can construct a set of generalised effects $\{\rbra{a^{(5)}_j}_\sys{ABC}\}_{j\in J_5}\subseteq\EffR{ABC}$ such that 
$\rbraket{a^{(5)}_j}{\alpha^{(5)}_k}_\sys{ABC}=\delta_{jk}$ while $\rbraket{a^{(5)}_j}{\sigma}_\sys{ABC}=0$ for all 
$\rket\sigma_\sys{ABC}\in\mathcal S^{(I)}_\sys{ABC}$. 
Let us now pick a complete linearly independent set 
$\{\rket{\alpha^{(i)}_j}_\sys{ABC}\}_{j\in J_i}$ 
in every subspace $\mathcal S^{(i)}_\sys{ABC}$---where $J_i\coloneqq\{1,2,\ldots,\dim\mathcal S^{(i)}_\sys{ABC}\}$.
By construction, every $\rket\rho_\sys{ABC}\in\StR{ABC}$ can be written as
\begin{align*}
\rket\rho_\sys{ABC}=\sum_{i=1}^5\sum_{j\in J_i}r^{(i)}_j\rket{\alpha^{(i)}_j}_\sys{ABC}.
\end{align*}
Let us now suppose that $\rket\rho_\sys{ABC}=\rket\varepsilon_\sys{ABC}$. Then we have
\begin{align*}
0=\rbraket{a^{(5)}_j}{\rho}_\sys{ABC}=r^{(5)}_j,\quad\forall j\in J_5,
\end{align*}
namely:
\begin{align*}
\rket\rho_\sys{ABC}=\sum_{i=1}^4\sum_{j\in J_i}r^{(i)}_j\rket{\alpha^{(i)}_j}_\sys{ABC}.
\end{align*}
Now, using formulae~\eqref{eq:separable_tripartite} and~\eqref{eq:bipartite_effect_zero}, for all $i\in I$ 
one can construct a set of generalised effects $\rbra{a^{(i)}_j}_\sys{ABC}\in\EffR{ABC}$ 
such that 
$\rbraket{a^{(i)}_j}{\alpha^{(i)}_k}_{\sys{ABC}}=\delta_{jk}$, and $\rbraket{a^{(i)}_j}{\alpha^{(l)}_{k'}}_{\sys{ABC}}=0$ for all $i,l\in I$ with $l\neq i$, all $j,k\in J_i$, and all $k'\in J_l$. Thus we have:
\begin{align*}
0=\rbraket{a^{(i)}_j}{\rho}_\sys{ABC}=r^{(i)}_j,\quad \forall i\in I,\ \forall j\in J_i.
\end{align*}
Accordingly, $\rket\rho_\sys{ABC}=\rket\varepsilon_\sys{ABC}$ if and only if $r^{(i)}_j=0$ for all $1\leq i\leq 5$ 
and all $j\in J_i$. This shows that the five subspaces $\mathcal{S}_{\sys{ABC}}^{(i)}$ 
for $i\in I$ and 
$\mathcal{T}_{\sys{ABC}}$ are not only complete, but also linearly independent. Then, 
in general Eq.~\eqref{eq:dimensions_tripartite}
holds.
\end{proof}

Notice that, for every positive integer $n$, $n$-local discriminability implies $(n+1)$-local discriminability (see Property~\ref{prope:n_local}). For this reason, we will call the property of $n$-local discriminability, in the case where $(n-1)$-local discriminability does not hold, \emph{strict $n$-local discriminability}. The case of $n$-local discriminability for $n=2$ is called \emph{bilocal discriminability}. Relevantly, strict bilocal discriminability is a property satisfied by Real Quantum Theory and Fermionic Quantum Theory~\cite{hardy2012limited,doi:10.1142/S0217751X14300257}. The next result provides the general form of the composition rules for the system-dimensions in a theory satisfying bilocal discriminability.
\begin{theorem}[\textbf{Composition rules for system-sizes in bilocal-tomographic theories}]\label{thm:bilocal}
	An OPT $\Theta$ satisfies bilocal discriminability if and only if the following equality holds for any tripartite system $\sys{ABC}\in\Sys{\Theta}$:
	\begin{align}\label{eq:dimensions_bilocal}
	\D{ABC} = \D{A}\D{B}\D{C} + \Delta^{(2)}_\sys{AB} \D{C} + \Delta^{(2)}_\sys{BC} \D{A} + \Delta^{(2)}_\sys{AC} \D{B},
	\end{align}
or, equivalently, $\Delta^{(3)}_\sys{ABC}=0$.
\end{theorem}
\begin{proof}
We recall Property~\ref{prope:n_local}. On the one hand, if Eq.~\eqref{eq:dimensions_bilocal} holds, for 
all $\sys{A},\sys{B},\sys{C}\in\Sys{\Theta}$ the space $\StR{ABC}$ has the same dimension as  
$\mathcal S^{(I)}_\sys{ABC}$. On the other hand, since $\mathcal S^{(I)}_\sys{ABC}$ is a subspace of $\StR{ABC}$, being all its elements allowed by the composition rules of OPTs, it must be $\StR{ABC}\equiv\mathcal S^{(I)}_\sys{ABC}$. This implies that the collection of effects $\bigcup_{i=1}^4\{\rbra{a^{i}_{j}}\}_{j\in J_i}$ introduced in the proof of Theorem~\ref{thm:delta3} separates states in $\StR{ABC}$, namely $\Theta$ satisfies Property~\ref{prope:n_local} with $n=1$ or 
$n=2$, depending on whether $\Delta^{(2)}_\sys{XY}=0$ or $\Delta^{(2)}_\sys{XY}\neq0$, respectively. On the other hand, if $\Theta$ satisfies Property~\ref{prope:n_local} with $n=1$ or $n=2$, 
then, by definition, the effects generated by the collection $\bigcup_{i=1}^4\{\rbra{a^{i}_{j}}\}_{j\in J_i}$ separate states. By direct inspection of the collection $\bigcup_{i=1}^4\{\rbra{a^{i}_{j}}\}_{j\in J_i}$, and recalling that for all systems $\sys{A}$ one has $\D{A}=\dim{\EffR{A}}$, then Eq.~\eqref{eq:dimensions_bilocal} holds. 
We conclude that
$\Theta$ satisfies bilocal discriminability if and only if Eq.~\eqref{eq:dimensions_bilocal} holds, or, equivalently (see Eq.~\eqref{eq:dimensions_tripartite}), if and only if $\Delta^{(3)}_{\sys{ABC}}=0$.
\end{proof}
Eq.~\eqref{eq:dimensions_bilocal} generalises Eq.~\eqref{eq:dimensions_local}, holding for theories satisfying local discriminability. Indeed, when local discriminability holds 
in Eq.~\eqref{eq:dimensions_bilocal} one has $\Delta^{(2)}_{\sys{XY}}=0$ for all $\sys{X},\sys{Y}\in\Sys{\Theta}$. Then, in the latter case, Eq.~\eqref{eq:dimensions_bilocal} represents just a restatement of Eq.~\eqref{eq:dimensions_local} for the tripartite scenario. More importantly, when strict bilocal discriminability holds---implying $\Delta^{(2)}_{\sys{XY}}\neq0$ for some $\sys{X},\sys{Y}\in\Sys{\Theta}$---the statement of Theorem~\ref{thm:bilocal} is non-trivial. Interestingly, Theorem~\ref{thm:bilocal} implies that the upper bound given in Refs.~\cite{hardy2012limited,doi:10.1142/S0217751X14300257}---holding for the dimension of tripartite systems in an OPT with bilocal discriminability---is in fact always saturated. In Ref.~\cite{doi:10.1142/S0217751X14300257}, a theory was defined to be \emph{maximally bilocal-tomographic} if it satisfies strict bilocal discriminability {\em and} Eq.~\eqref{eq:dimensions_bilocal}. However, Theorem~\ref{thm:bilocal} states that if a theory satisfies strict bilocal discriminabilty then it necessarily satisfies also Eq.~\eqref{eq:dimensions_bilocal}, namely \emph{every theory satisfying strict bilocal discriminability is in fact maximally bilocal-tomographic}.

\subsection{On the notions of no-restriction hypothesis and causality}\label{subsec:causality}
The following two important properties, called \emph{the no-restriction hypothesis} and \emph{causality}~\cite{bookDCP2017}, are extensively assumed throughout the literature on probabilistic theories.
\begin{property}[No-restriction hypothesis]\label{prope:no_restriction}	
	Let $\Theta$ be an OPT, $\sys{A},\sys{B}\in\Sys{\Theta}$, and $\Instr{T}{X}\subset\TransfR{A}{B}$ be a collection of generalised transformations. If, for all $\sys{E}\in\Sys{\Theta}$, $\Instr{T}{X}\boxtimes\InstrC{I}{\star}{E}{E}$ maps preparation-instruments of $\sys{AE}$ to preparation-instruments of $\sys{BE}$, then $\Instr{T}{X}\in\InstrA{\Theta}$, namely, $\Instr{T}{X}$ is an instrument of the theory $\Theta$.
\end{property}
For any given theory $\Theta$, one is always able to check whether $\Theta$ satisfies Property~\ref{prope:no_restriction}, by definition. If a theory satisfies the no-restriction hypothesis, this means that it cannot be extended---\emph{in a consistent way}---by adding further transformations or effects.

A brief comment on the formulation of Property~\ref{prope:no_restriction} is in order. The no-restriction hypothesis is usually formulated stating that ``all the positive functionals on states are physical'', or something along these lines. First, we explicitly included also general transformations. Second, to make a sensible statement, in general one has to be sure that a collection of generalised transformations not only maps states to states, but also that this is the case when the collection is extended to a composite system. The previous observation is relevant in the absence of local discriminability (see Remark~\ref{rem:transformation_equiv_class}).
\begin{remark}\label{rem:no-restriction}
	From a constructive perspective---namely, when a theory $\Theta$ is constructed out of a number of postulates---one can use Property~\ref{prope:no_restriction} as a postulate, and this is indeed very common in the literature on probabilistic theories. However, despite being a viable constructive postulate, in general the no-restriction hypothesis is not \emph{selective}, in the sense that, by assigning the preparation-instruments and imposing Property~\ref{prope:no_restriction}, one does not generally end up with a single theory, but with \emph{a family of different theories}.\footnote{This is due to the fact that a class of generalised instruments may fail to be \emph{uniquely} determined by the mere geometry of states. A significant example is given by the theory of superselected qubits and Fermionic Quantum Theory~\cite{doi:10.1142/S0217751X14300257}---sharing the same class of preparation-instruments and both satisfying the no-restriction hypothesis, despite being different theories.} The only case where the no-restriction hypothesis can be safely taken as a selective postulate is when the choice of the preparation-instruments satisfies local discriminability (which is in fact the case in vast majority of the literature on probabilistic theories). Indeed, in the latter case, the local states of every composite system $\sys{AB}$ generate the whole space $\StR{AB}$---parallel composition being given by the standard tensor product, i.e., $\boxtimes\equiv\otimes$. As a consequence, every transformation $\T{T}\in\Transf{A}{C}$ is solely defined by its action on $\St{A}$ (see Remark~\ref{rem:transformation_equiv_class}) and the braiding $\T{S}$ is \emph{uniquely determined} by the mere geometry of states, so that $\T{T}\in\Transf{A}{C}$ maps $\St{A}$ to $\St{C}$ if and only if $\T{T}\boxtimes\T{I}_{\sys{B}}$ maps $\St{AB}$ to $\St{CB}$ for all $\sys{B}\in\Sys{\Theta}$.
\end{remark}
\begin{property}[Causality]\label{prope:causality}
	The probability distributions of preparation-instruments do not depend on the choice of
	observation-instrument at their output.
\end{property}
We remind that the above statement of causality is equivalent to the uniqueness of the deterministic effect for every system~\cite{bookDCP2017}, namely, an OPT is causal if and only if every system has a unique deterministic effect. The unique deterministic effect is often called \emph{the unit effect (or measure)}~\cite{PhysRevA.87.052131}. The causality condition in Property~\ref{prope:causality} is known as \emph{no signalling from the future}. The consequences of this requirement have been extensively studied in Ref.~\cite{guryanova2019exploring}. Incidentally, it also implies another kind of no-signalling, namely, the \emph{no signalling at a distance (without interaction)}~\cite{bookDCP2017}, which is the standard assumption of \emph{spatial no-signalling} made in the spirit of nonlocal boxes~\cite{Popescu1994}. Indeed, in a bipartite scenario for, say, system $\sys{AB}$, the uniqueness of the deterministic effect always allows one to define a unique marginal state for both $\sys{A}$ and $\sys{B}$.
\begin{property}[{Conditional instruments}]\label{postulate:conditional}
	For all systems $\sys{A},\sys{B},\sys{C}\in\Sys{\Theta}$, all instruments $\Instr{A}{X}\in\InstrAB{A}{B}$, and all choices of maps
	\begin{align*}
	\begin{split}
	\mathsf{X}&\longrightarrow\InstrAB{B}{C}\\
	x& \longmapsto\mathsf{B}_{\mathsf{Y}^{(x)}}^{(x)}\;,
	\end{split}
	\end{align*}
	the following holds:
	\begin{align*}
	\bigcup_{x\in\mathsf{X}}\lbrace\T{B}_{y}^{(x)}\circ\T{A}_x\rbrace_{y\in\mathsf{Y}^{(x)}}\in\InstrA{\Theta}.
	\end{align*}
	In other words, the generalised conditional instrument associated to every $\Instr{A}{X}\in\InstrAB{A}{B}$ and every choice of labelled collection $\{\mathsf{B}_{\mathsf{Y}^{(x)}}^{(x)}\}_{x\in\mathsf{X}}\subset\InstrAB{B}{C}$ is an instrument of the theory.
\end{property}
Property~\ref{postulate:conditional} affects the geometry of instruments of a theory as follows. Suppose that an arbitrary instrument $\Instr{A}{X}\in\InstrAB{A}{B}$ is performed. Property~\ref{postulate:conditional} guarantees that it is always possible to use any outcome $x\in\mathsf{X}$ of $\Instr{A}{X}$ in order to freely choose a second instrument $\mathsf{B}_{\mathsf{Y}^{(x)}}^{(x)}\in\InstrAB{B}{C}$ to be performed \emph{at the output} of the transformation $\T{A}_{x}\in\Instr{A}{X}$. This has two major consequences~\cite{bookDCP2017,Perinotti2020cellularautomatain}. In any theory satisfying Property~\ref{postulate:conditional}, the existence of a probability which is not $0$ or $1$ implies that every convex combination of instruments is itself an instrument. As a consequence, the theory is convex.
Moreover, Property~\ref{postulate:conditional} selects theories where signals can propagate \emph{solely} in the input-output direction (see Property~\ref{prope:causality}), namely, causality is implied. By the above reasons, in Ref.~\cite{Perinotti2020cellularautomatain} it is argued that Property~\ref{postulate:conditional} should be considered as a (strong) notion of causality.

The no-restriction hypothesis and causality are often \emph{simultaneously} assumed throughout the literature on probabilistic theories, being a defining part of the framework itself. However, we observe that these properties might be incompatible, even in the presence of local discriminability. Indeed, assuming the no-restriction hypothesis for a class of states may give rise---depending on their geometry---to a class of effects containing more than one deterministic effect, namely, to a theory which is not causal. The possibility of formulating probabilistic theories without the no-restriction hypothesis has been explored in the literature~\cite{PhysRevA.87.052131}, whereas causality is therein implicitly assumed.

\subsection{Simplicial and classical theories}\label{subsec:simplicial}
We conclude this subsection introducing some final definitions and results which are particularly relevant to the scope of the present work.
\begin{definition}[Simplicial theories]
	A \emph{simplicial theory} $\Theta$ is a finite-dimensional OPT where the extremal states of every system $\sys{A}\in\Sys{\Theta}$ are the vertices of a $\D{A}$-simplex.
\end{definition}
Notice that a $\D{A}$-simplex is the convex hull of $\D{A}+1$ affinely independent vertices, being in the present context the elements of $\ExtStZ{A}$, which includes the null state $\rket{\varepsilon}_\sys{A}$.
\begin{property}[Joint perfect discriminability in causal theories]\label{prope:jpd}
	Let $\Theta$ be a causal OPT and $\sys{A}\in\Sys{\Theta}$. A set of states $\{\rket{\rho_i}_\sys A\}_{i=1}^n$ is jointly perfectly discriminable if there exists an observation-instrument $\{ \rbra{a_i}_{\sys{A}}\}_{i=1}^n$ such that:
	\begin{align*}
	\rbraket{a_i}{\rho_{i'}}_\sys{A} = \delta_{ii'},\quad\forall i,i'\in\lbrace 1,2,\ldots, n\rbrace.
	\end{align*}
\end{property}
Notice that, in causal theories, a necessary condition for a set of states $\{\rket{\rho_i}_\sys A\}_{i=1}^n$ to satisfy Property~\ref{prope:jpd} is that $\{\rket{\rho_i}_\sys A\}_{i=1}^n\subseteq\StN{A}$.
\begin{definition}[classical theories]\label{def:classical_theory}
	A \emph{classical theory} is a simplicial theory where the pure states of every system are jointly perfectly discriminable.
\end{definition}
Simplicial theories are not necessarily convex. However, the set of states of a simplicial theory is relatively simple to treat, since every state admits of a unique decomposition into non-null extremal states. Moreover, simplicial theories are inherently causal, and for such theories $\ExtSt{A}\equiv\PurSt{A}$ for every system $\sys{A}$ (see Theorem 1 and its proof in Ref.~\cite{d2019classical}). In the following, when dealing with simplicial theories, we will denote the unique deterministic effect of any system $\sys{A}$ by $\rbra{e}_{\sys{A}}$.
\begin{definition}[Classical Theory]
	\emph{Classical Theory (CT)} is the OPT $\Theta$ satisfying the following properties: (i) $\Theta$ is simplicial and convex, (ii) local discriminability holds, (iii) the preparation-instruments of every system $\sys{A}\in\Sys{\Theta}$ are all the collections of states of $\sys{A}$ that add up to a point in the convex hull of $\PurSt{A}$, and (iv) the no-restriction hypothesis holds.
\end{definition}
Notice that, by Remark~\ref{rem:no-restriction}, in the above case the no-restriction hypothesis singles out a unique theory. CT and QT, despite being very different theories, feature some relevant common properties, most notably causality and local discriminability. They also share convexity and the no-restriction hypothesis. Indeed, both CT and QT also satisfy a stronger property, known as \emph{(strong) self-duality}~\cite{PhysRevLett.108.130401}. This means that, for every system $\sys{A}$, one has the equality $\EffC{A}=\StC{A}^{\vee}$ (using Riesz's representation).

We finally discuss another property which also holds in both CT and QT.
\begin{property}[Atomicity of parallel composition]\label{prope:atomicity}
	The parallel composition of two atomic transformations is atomic.
\end{property}
Analogously to the case of local discriminability, if a theory does not satisfy atomicity of \emph{state-composition}, then it necessarily has entangled states~\cite{d2019classical}. In the simplicial case, local disciminability and atomicity of \emph{state-composition} (provided that $n$-local discriminability is satisfied for some $n$) are in fact equivalent~\cite{d2019classical}. Indeed, the simplicial theory presented in this manuscript satisfies strict bilocal discriminability, and not atomicity of parallel composition. In Appendix~\ref{app:simplicial} we prove a structure theorem for the parallel composition rule of strictly bilocal simplicial theories satisfying a certain property on the reversible transformations (see Property~\ref{prope:uniqueness_purification} in Subsec.~\ref{subsec:purity_preservation}).

\section{\textbf{Constructing an OPT}}\label{sec:constructing_an_opt}
We begin the present section treating the consistency of an OPT. We will provide an account of the general consistency conditions which any OPT must comply with in order to be coherent and well-posed. In the light of Remark~\ref{rem:transformation_equiv_class}, the careful assessment of coherence conditions, indeed, is particularly relevant in the absence of local discriminability, as it will be clear in the following. Then we shall present an exhaustive procedure to construct an OPT and check its consistency.

\subsection{Coherence and well-posedness}\label{subsec:coherence}
When a mathematical structure is constructed, one needs to make sure that it is 
consistent. Accordingly, this section will be devoted to present the consistency 
conditions that must be imposed on the construction of an OPT: as a category, the 
OPT must abide by coherence conditions~\cite{lane2013categories}, while the further probabilistic 
structure must be well-posed, and the two structures must be 
compatible.

Both sequential and parallel composition of tests are required to be associative operations. The reason is that when an experimenter considers composed physical tests, the choice of a particular association is just a formal action, not corresponding to a physical operation. In an OPT, transformations are defined as maps between sets: two transformations are the same if they are represented by the same function. Accordingly, sequential composition coincides with \emph{function composition}, which is always associative. Therefore, in the construction of an OPT---as long as transformations are defined as linear maps between vector spaces---the following always holds:
\begin{align*}
&\begin{aligned}
\Qcircuit @C=1.2em @R=1.3em
{
	&\s{A}&\qw&\gate{\T{T}_1}&\s{B}\qw&\gate{\T{T}_2}&\s{C}\qw&\gate{\T{T}_3}&\qw&\s{D}\qw&
	\relax\gategroup{1}{4}{1}{8}{2em}{--}
	\relax\gategroup{1}{4}{1}{6}{1em}{--}
}
\end{aligned}
=\\[2.5ex]
=&\begin{aligned}
\Qcircuit @C=1.2em @R=1.3em
{
	&\s{A}&\qw&\gate{\T{T}_1}&\s{B}\qw&\gate{\T{T}_2}&\s{C}\qw&\gate{\T{T}_3}&\qw&\s{D}\qw&
	\relax\gategroup{1}{4}{1}{8}{2em}{--}
	\relax\gategroup{1}{6}{1}{8}{1em}{--}
}
\end{aligned}.
\end{align*}
The previous argument does not apply to parallel composition. Accordingly, in general, when one chooses a rule for parallel composition, one has to assign an invertible map $\alpha$ from $\left(\Sys{\Theta}\boxtimes\Sys{\Theta}\right)\boxtimes\Sys{\Theta}$ to $\Sys{\Theta}\boxtimes\left(\Sys{\Theta}\boxtimes\Sys{\Theta}\right)$, called \emph{the associator}, such that:
\begin{align}\label{eq:naturality_associator}
\begin{aligned}
\Qcircuit @C=0.9em @R=1.3em
{
	&\s{AC}&\gate{\T{T}_1\boxtimes\T{T}_2}&\s{BD}\qw&\multigate{1}{\alpha}&\s{B}\qw&
	\\
	&\s{E}&\gate{\T{T}_3}&\s{F}\qw&\ghost{\alpha}&\s{DF}\qw&
}
\end{aligned}=
\begin{aligned}
\Qcircuit @C=0.9em @R=1.3em
{
	&\s{AC}&\multigate{1}{\alpha}&\s{A}\qw&\gate{\T{T}_1}&\s{B}\qw&
	\\
	&\s{E}&\ghost{\alpha}&\s{CE}\qw&\gate{\T{T}_2\boxtimes\T{T}_3}&\s{DF}\qw&
}
\end{aligned}.
\end{align}
This map allows one to switch between different associations respecting the associativity condition. Namely, the associator is given by the identity transformation itself:
\begin{align}\label{eq:association_rule}
\begin{aligned}
\Qcircuit @C=1.3em @R=1.3em
{
	&\qw&\s{(AB)C}\qw&\qw&\qw&
}
\end{aligned}=
\begin{aligned}
\Qcircuit @C=1em @R=1.3em
{
	&\s{AB}&\multigate{1}{\alpha}&\s{A}\qw&
	\\
	&\s{C}&\ghost{\alpha}&\s{BC}\qw&
}
\end{aligned}=\ \ 
\begin{aligned}
\Qcircuit @C=1.3em @R=1.3em
{
	&\qw&\s{A(BC)}\qw&\qw&\qw&
}
\end{aligned}.
\end{align}
Now, from a constructive perspective, when an associative parallel composition rule is chosen, not only one has to assign an associator, but also one needs to check whether it is \emph{coherent}---namely, self-consistent when extended to the composition of more than three objects. In order to verify this important requirement, it is sufficient to check the so-called \emph{pentagon identity}:\footnote{This nomenclature is due to the fact that the corresponding commutative diagram, generally holding for a non-strict monoidal category, has five vertices (see e.g.~Ref.~\cite{lane2013categories}).}
\begin{align}
\begin{split}\label{eq:pentagon_parallel}
&\begin{aligned}
\Qcircuit @C=1.6em @R=1.3em
{
	&\s{(AB)C}&\multigate{1}{\alpha}&\s{AB}\qw&\multigate{1}{\alpha}&\s{A}\qw&
	\\
	&\s{D}&\ghost{\alpha}&\s{CD}\qw&\ghost{\alpha}&\s{B(CD)}\qw&
}
\end{aligned}=
\\[2.5ex]
=&\begin{aligned}
\Qcircuit @C=1.6em @R=1.3em
{
	&\s{(AB)C}&\gate{\alpha}&\s{A(BC)}\qw&\multigate{1}{\alpha}&\qw&\s{A}\qw&\qw&
	\\
	&&\s{D}\qw&\qw&\ghost{\alpha}&\s{(BC)D}\qw&\gate{\alpha}&\s{B(CD)}\qw&
}
\end{aligned},
\end{split}
\end{align}
for all $\sys{A},\sys{B},\sys{C},\sys{D}\in\Sys{\Theta}$, where one sequentially uses the appropriate assigned association rule. Eqs.~\eqref{eq:naturality_associator} and~\eqref{eq:association_rule}, along with the pentagon identity~\eqref{eq:pentagon_parallel}, ensure that the operation of parallel composition is associative in a coherent way.

The second coherence requirement is related to the possibility of exchanging systems between agents. For instance, one demands that the single exchange
\begin{align*}
\sys{A}(\sys{B}\sys{C})\mapsto(\sys{B}\sys{C})\sys{A},
\end{align*}
is equivalent, for operational consistency, to the two sequential exchanges
\begin{align*}
\sys{A}(\sys{B}\sys{C})\!\equiv\!(\sys{A}\sys{B})\sys{C}\mapsto(\sys{B}\sys{A})\sys{C}\!\equiv\!\sys{B}(\sys{A}\sys{C})\mapsto\sys{B}(\sys{C}\sys{A})\!\equiv\!(\sys{B}\sys{C})\sys{A}.
\end{align*}
In other words, the braiding $\T{S}$ is required to compose as in the \emph{braid group}. In order to verify this, it is sufficient to check the two so-called \emph{hexagon identities}:\footnote{Analogously to the case of the pentagon identity, this terminology is related to the fact that the corresponding commutative diagrams have six vertices (see e.g.~Ref.~\cite{lane2013categories}).}
\begin{align}
\label{eq:hexagon1_diagram}
\begin{tikzpicture}
	\begin{pgfonlayer}{nodelayer}
		\node [style=none] (0) at (-3, 0) {};
		\node [style=none] (1) at (-3, 2) {};
		\node [style=none] (2) at (-0.5, 2) {};
		\node [style=none] (3) at (-0.5, 0) {};
		\node [style=none] (4) at (-2, 0.75) {};
		\node [style=none] (5) at (-1.5, 1.25) {};
		\node [style=none] (6) at (-0.5, 2) {};
		\node [style=none] (7) at (-0.5, 0) {};
		\node [style={system_label}] (8) at (1.75, 2.75) {$\scriptstyle \sys  B$};
		\node [style={system_label}] (9) at (0, 0.75) {$\scriptstyle \sys  A$};
		\node [style=none] (10) at (-3.5, 2) {};
		\node [style=none] (11) at (-3.5, 0) {};
		\node [style=none] (12) at (-3, 2) {};
		\node [style=none] (13) at (-3, 0) {};
		\node [style={system_label}] (14) at (-3.5, 2.75) {$\scriptstyle \sys  A$};
		\node [style={system_label}] (15) at (-3.5, 0.75) {$\scriptstyle \sys  B$};
		\node [style=none] (16) at (3.5, 2) {};
		\node [style=none] (17) at (0, 0) {};
		\node [style=none] (18) at (0.5, -2) {};
		\node [style=none] (19) at (0.5, 0) {};
		\node [style=none] (20) at (3, 0) {};
		\node [style=none] (21) at (3, -2) {};
		\node [style=none] (22) at (1.5, -1.25) {};
		\node [style=none] (23) at (2, -0.75) {};
		\node [style=none] (24) at (3, 0) {};
		\node [style=none] (25) at (3, -2) {};
		\node [style={system_label}] (26) at (3.5, 0.75) {$\scriptstyle \sys  C$};
		\node [style={system_label}] (27) at (3.5, -1.25) {$\scriptstyle \sys  A$};
		\node [style=none] (28) at (0, 0) {};
		\node [style=none] (29) at (-3.5, -2) {};
		\node [style=none] (30) at (0.5, 0) {};
		\node [style=none] (31) at (0.5, -2) {};
		\node [style={system_label}] (33) at (-1.75, -1.25) {$\scriptstyle \sys  C$};
		\node [style=none] (34) at (3.5, 0) {};
		\node [style=none] (35) at (3.5, -2) {};
		\node [style={system_label}] (36) at (1.75, -2.25) {};
	\end{pgfonlayer}
	\begin{pgfonlayer}{edgelayer}
		\draw [in=-180, out=0, looseness=0.75] (1.center) to (3.center);
		\draw [in=-135, out=0] (0.center) to (4.center);
		\draw [in=180, out=45] (5.center) to (2.center);
		\draw (10.center) to (12.center);
		\draw (11.center) to (13.center);
		\draw (6.center) to (16.center);
		\draw (7.center) to (17.center);
		\draw [in=-180, out=0, looseness=0.75] (19.center) to (21.center);
		\draw [in=-135, out=0] (18.center) to (22.center);
		\draw [in=180, out=45] (23.center) to (20.center);
		\draw (28.center) to (30.center);
		\draw (29.center) to (31.center);
		\draw (24.center) to (34.center);
		\draw (25.center) to (35.center);
	\end{pgfonlayer}
\end{tikzpicture}\!\!\!=\!\!\begin{tikzpicture}
	\begin{pgfonlayer}{nodelayer}
		\node [style=none] (0) at (-3, -1.25) {};
		\node [style=none] (1) at (-3, 2) {};
		\node [style=none] (6) at (3, 2) {};
		\node [style=none] (10) at (-3.5, 2) {};
		\node [style=none] (11) at (-3.5, -1.25) {};
		\node [style=none] (12) at (-3, 2) {};
		\node [style=none] (13) at (-3, -1.25) {};
		\node [style={system_label}] (14) at (-3.5, 3) {$\scriptstyle \sys  A$};
		\node [style={system_label}] (15) at (-3.5, -0.25) {$\scriptstyle \sys{BC}$};
		\node [style=none] (16) at (3.5, 2) {};
		\node [style=none] (20) at (3, 1.25) {};
		\node [style=none] (21) at (3, -2) {};
		\node [style=none] (24) at (3, 1.25) {};
		\node [style=none] (25) at (3, -2) {};
		\node [style={system_label}] (26) at (3.5, 3) {$\scriptstyle \sys{BC}$};
		\node [style={system_label}] (27) at (3.5, -1) {$\scriptstyle \sys  A$};
		\node [style=none] (29) at (-3.5, -2) {};
		\node [style=none] (31) at (-3, -2) {};
		\node [style=none] (34) at (3.5, 1.25) {};
		\node [style=none] (35) at (3.5, -2) {};
		\node [style=none] (36) at (0, -0.75) {};
		\node [style=none] (37) at (-0.5, -0.25) {};
		\node [style=none] (38) at (0, 0.75) {};
		\node [style=none] (39) at (0.5, 0.25) {};
	\end{pgfonlayer}
	\begin{pgfonlayer}{edgelayer}
		\draw (10.center) to (12.center);
		\draw (11.center) to (13.center);
		\draw (6.center) to (16.center);
		\draw (29.center) to (31.center);
		\draw (24.center) to (34.center);
		\draw (25.center) to (35.center);
		\draw [in=180, out=0, looseness=0.75] (12.center) to (25.center);
		\draw [in=-120, out=0] (13.center) to (37.center);
		\draw [in=-120, out=0] (31.center) to (36.center);
		\draw [in=-180, out=60] (39.center) to (24.center);
		\draw [in=180, out=60] (38.center) to (6.center);
	\end{pgfonlayer}
\end{tikzpicture}\!\!\!,
\\[1ex]\label{eq:hexagon2_diagram}
\begin{tikzpicture}
	\begin{pgfonlayer}{nodelayer}
		\node [style={system_label}] (14) at (-1.75, 2.75) {$\scriptstyle \sys  A$};
		\node [style={system_label}] (27) at (1.75, -1.25) {$\scriptstyle \sys  B$};
		\node [style={system_label}] (36) at (3.5, -2.5) {};
		\node [style=none] (37) at (0.5, 0) {};
		\node [style=none] (38) at (0.5, 2) {};
		\node [style=none] (39) at (3, 2) {};
		\node [style=none] (40) at (3, 0) {};
		\node [style=none] (41) at (1.5, 0.75) {};
		\node [style=none] (42) at (2, 1.25) {};
		\node [style=none] (43) at (3, 2) {};
		\node [style=none] (44) at (3, 0) {};
		\node [style={system_label}] (45) at (3.5, 0.75) {$\scriptstyle \sys  X$};
		\node [style=none] (46) at (0, 2) {};
		\node [style=none] (47) at (0, 0) {};
		\node [style=none] (48) at (0.5, 2) {};
		\node [style=none] (49) at (0.5, 0) {};
		\node [style={system_label}] (50) at (3.5, 2.75) {$\scriptstyle \sys  C$};
		\node [style=none] (52) at (3.5, 0) {};
		\node [style=none] (53) at (3.5, 0) {};
		\node [style={system_label}] (54) at (3.5, 0.75) {$\scriptstyle \sys  A$};
		\node [style=none] (55) at (-3.5, 0) {};
		\node [style=none] (56) at (-3, -2) {};
		\node [style=none] (57) at (-3, 0) {};
		\node [style=none] (58) at (-0.5, 0) {};
		\node [style=none] (59) at (-0.5, -2) {};
		\node [style=none] (60) at (-2, -1.25) {};
		\node [style=none] (61) at (-1.5, -0.75) {};
		\node [style=none] (62) at (-0.5, 0) {};
		\node [style=none] (63) at (-0.5, -2) {};
		\node [style=none] (65) at (-3.5, 0) {};
		\node [style=none] (66) at (-3, 0) {};
		\node [style=none] (67) at (-3, -2) {};
		\node [style=none] (68) at (0, 0) {};
		\node [style=none] (69) at (0, -2) {};
		\node [style={system_label}] (70) at (0, -2.5) {};
		\node [style=none] (71) at (-3.5, 2) {};
		\node [style=none] (72) at (-3.5, -2) {};
		\node [style=none] (73) at (3.5, 2) {};
		\node [style=none] (74) at (3.5, -2) {};
		\node [style={system_label}] (75) at (-3.5, 0.75) {$\scriptstyle \sys  B$};
		\node [style={system_label}] (76) at (-3.5, -1.25) {$\scriptstyle \sys  C$};
		\node [style={system_label}] (77) at (0, 0.75) {$\scriptstyle \sys  C$};
	\end{pgfonlayer}
	\begin{pgfonlayer}{edgelayer}
		\draw [in=-180, out=0, looseness=0.75] (38.center) to (40.center);
		\draw [in=-135, out=0] (37.center) to (41.center);
		\draw [in=180, out=45] (42.center) to (39.center);
		\draw (46.center) to (48.center);
		\draw (47.center) to (49.center);
		\draw (44.center) to (52.center);
		\draw [in=-180, out=0, looseness=0.75] (57.center) to (59.center);
		\draw [in=-135, out=0] (56.center) to (60.center);
		\draw [in=180, out=45] (61.center) to (58.center);
		\draw (65.center) to (66.center);
		\draw (62.center) to (68.center);
		\draw (63.center) to (69.center);
		\draw (71.center) to (46.center);
		\draw (72.center) to (67.center);
		\draw (69.center) to (74.center);
		\draw (43.center) to (73.center);
	\end{pgfonlayer}
\end{tikzpicture}\!\!\!=\!\!\begin{tikzpicture}
	\begin{pgfonlayer}{nodelayer}
		\node [style=none] (0) at (-3, 1.25) {};
		\node [style=none] (1) at (-3, 2) {};
		\node [style=none] (6) at (3, 2) {};
		\node [style=none] (10) at (-3.5, 2) {};
		\node [style=none] (11) at (-3.5, 1.25) {};
		\node [style=none] (12) at (-3, 2) {};
		\node [style=none] (13) at (-3, 1.25) {};
		\node [style={system_label}] (14) at (-3.5, 3) {$\scriptstyle \sys{AB}$};
		\node [style={system_label}] (15) at (-3.5, -1) {$\scriptstyle \sys  C$};
		\node [style=none] (16) at (3.5, 2) {};
		\node [style=none] (21) at (3, -2) {};
		\node [style={system_label}] (26) at (3.5, 3) {$\scriptstyle \sys  C$};
		\node [style={system_label}] (27) at (3.25, -0.25) {$\scriptstyle \sys{AB}$};
		\node [style=none] (29) at (-3.5, -2) {};
		\node [style=none] (31) at (-3, -2) {};
		\node [style=none] (35) at (3.5, -2) {};
		\node [style=none] (40) at (3.5, -1.25) {};
		\node [style=none] (41) at (3, -1.25) {};
		\node [style=none] (42) at (3.5, -1.25) {};
		\node [style=none] (43) at (3, -1.25) {};
		\node [style=none] (45) at (-3, 1.25) {};
		\node [style=none] (46) at (-0.5, -0.75) {};
		\node [style=none] (47) at (0.5, 0.5) {};
		\node [style=none] (48) at (0.25, 0.25) {};
		\node [style=none] (51) at (-0.25, -0.25) {};
	\end{pgfonlayer}
	\begin{pgfonlayer}{edgelayer}
		\draw (10.center) to (12.center);
		\draw (11.center) to (13.center);
		\draw (6.center) to (16.center);
		\draw (29.center) to (31.center);
		\draw (21.center) to (35.center);
		\draw (41.center) to (42.center);
		\draw [in=-135, out=0, looseness=1.25] (31.center) to (46.center);
		\draw [in=45, out=180, looseness=1.25] (6.center) to (47.center);
		\draw [in=120, out=0] (45.center) to (51.center);
		\draw [in=120, out=0] (12.center) to (48.center);
		\draw [in=180, out=-60] (48.center) to (43.center);
		\draw [in=-60, out=180] (21.center) to (51.center);
	\end{pgfonlayer}
\end{tikzpicture}\!\!\!.
\end{align}
for all $\sys{A},\sys{B},\sys{C}\in\Sys{\Theta}$ (where we omitted an explicit graphical representation of the associator). However, in the case where the theory is {\em symmetric}, it is easy to verify that the two hexagon identities~\eqref{eq:hexagon1_diagram} and~\eqref{eq:hexagon2_diagram} are in fact equivalent~\cite{lane2013categories}.

In a general category, one also needs to assign two suitably defined invertible maps $\lambda\colon \sys{IA}\mapsto\sys{A}$ and $\rho\colon \sys{AI}\mapsto\sys{A}$, called, respectively, \emph{the left} and \emph{the right unitors} (see e.g.~Ref.~\cite{lane2013categories}). Accordingly, $\lambda$ and $\rho$ must be not only well-posed, but also abide by some coherence conditions involving the associator $\alpha$ and the braiding $\T{S}$. In fact, in the light of the linear structure derived in Subsec.~\ref{subsec:probabilistic}---in particular, by virtue of Eqs.~\eqref{eq:scalar_multiplication},~\eqref{eq:scalar_multiplication_consistency} and what follows---well-posedness and coherence for the unitors are respected, and do not need to be checked.

We conclude with some final requirement of compatibility with the probabilistic structure. First, states must be separating for effects and \emph{vice versa}. Instruments $\Test{P}{X}\in\InstrAB{I}{I}$ must be probability distributions, and the null transformation $\varepsilon_{\sys{A}\to\sys{B}}$ must be included in $\Transf{A}{B}$ for every $\sys{A},\sys{B}$. The coarse-graining operation must be well-posed for every instrument. Finally, for compatibility with sequential composition, every instrument $\TestC{T}{X}{A}{B}\boxtimes\InstrC{I}{\star}{E}{E}$ must map preparation-instruments of $\sys{AE}$ to preparation-instruments of $\sys{BE}$ for every system $\sys{A},\sys{B},\sys{E}$.

We are now in position to present a possible procedure for the construction of an OPT $\Theta$, and for the check of its well-posedness.

\subsection{Setting the postulates}\label{subsec:postulates}
In the present subsection, we elaborate the procedure sketched in Remark~\ref{rem:OPT}. First, we can specify the class $\Sys{\Theta}$, endowed with the binary composition operation $(\sys{A},\sys{B})\mapsto\sys{A}\sys{B}\in\Sys{\Theta}$ for each pair $\sys{A},\sys{B}\in\Sys{\Theta}$. A dimension $\D{A}$ for the real vector space $\StR{A}$ and a class of states $\St{A}\subset\StR{A}$ are associated with each system $\sys{A}\in\Sys{\Theta}$.
Now we can assign a composition rule for systems, choosing a map $f\colon(\sys{A},\sys{B})\mapsto\D{AB}$ for each pair $\sys{A},\sys{B}\in\Sys{\Theta}$.

We can then proceed by choosing, for all $\sys{A},\sys{B},\sys{E}\in\Sys{\Theta}$, the action on $\St{AE}$ of the local transformations from $\sys{A}$ to $\sys{B}$:
\begin{align*}
\begin{aligned}
\Qcircuit @C=1em @R=1em
{
	&\s{AE}&\gate{\T{T}\boxtimes\T{I}_\sys{E}}&\s{BE}\qw&
}
\end{aligned}
=
\begin{aligned}
\Qcircuit @C=1em @R=1.5em
{
	&\s{A}&\gate{\T{T}}&\s{B}\qw&
	\\
	&&\s{E}\qw&\qw&
}
\end{aligned}.
\end{align*}
These span a real vector space of linear functions from $\StR{AE}$ to $\StR{BE}$. On the one hand, one should always include the identity family $\T{I}$. On the other hand, the two elementary cases $\sys{E}=\sys{I}$ and $\sys{A}=\sys{B}=\sys{I}$, corresponding to scalar multiplication, have been already set in Subsec.~\ref{subsec:probabilistic} (see, in particular, Eq.~\eqref{eq:scalar_multiplication_consistency}).
Then, we can proceed with the cases $\sys{A}=\sys{I}$ and $\sys{B}\neq\sys{I}$, specifying the action of all the local transformations of the form $\rket{\sigma}_{\sys{B}}\boxtimes\T{I}_\sys{E}\in\Transf{E}{BE}$ on all $\rho\in\St{E}$, namely a rule for composing states in parallel:
\begin{align*}
\begin{aligned}
\Qcircuit @C=1em @R=1em
{   
	&&\prepareC{\sigma}&\s{B}\qw&
	\\
	&\prepareC{\rho}&\s{E}\qw&\qw
}
\end{aligned}\in\St{BE}.
\end{align*}
This is done by choosing a decomposition of every product state into linearly independent vectors of the composite system. Since it must be $\St{((AB)C)}=\St{(A(BC))}$, we need also to specify---by choosing an associator $\alpha$ as in Eq.~\eqref{eq:association_rule}---the map identifying the linearly independent vectors in $\StR{(\left(AB\right)C)}$ with those in $\StR{(A\left(BC\right))}$. Also, we specify the action of all the local transformation of the form $\rbra{a}_{\sys{A}}\boxtimes\T{I}_\sys{E}\in\TransfR{AE}{E}$ on all $\Sigma\in\St{AE}$, namely a rule allowing one to construct the conditioning of bipartite states:
\begin{align*}
\begin{aligned}
\Qcircuit @C=1em @R=1em
{   \multiprepareC{1}{\Sigma}&\s{A}\qw&\measureD{a}\\
	\pureghost{\Sigma}&\qw&\s{E}\qw&\qw
}
\end{aligned}\in\StR{E}.
\end{align*}
This defines how the local generalised effects of every $\sys{AB}$ embed in $\EffR{AB}$, but not yet the class of effects, nor that of the observation-instruments, of $\Theta$. A very common choice throughout the literature is to assume the no-restriction hypothesis in the form ``all the positive functionals on states are physical'' (see Remark~\ref{rem:no-restriction}).

Then, we can define the action of arbitrary $\T{T}\boxtimes\T{I}_\sys{E}\in \Transf{AE}{BE}$ on $\St{AE}$ for all systems $\sys{A},\sys{B},\sys{E}\neq\sys{I}$, and thus set a braiding $\Out{S}$ for the theory. Finally, we are left with specifying the instruments of $\Theta$. One possibility is to choose the preparation-instruments of the theory, and then postulate that the class $\InstrA{\Theta}$ includes a collection $\InstrC{T}{X}{A}{B}$ of transformations if and only if $\InstrC{T}{X}{A}{B}\boxtimes\InstrC{I}{\star}{E}{E}$ maps preparation-instruments of $\sys{AE}$ to preparation-instruments $\sys{BE}$ for all $\sys{E}\in\Sys{\Theta}$. Notice that such a requirement does not amount to postulating the no-restriction hypothesis (see Property~\ref{prope:no_restriction} and Remark~\ref{rem:no-restriction}).

We are now in position to show that the above construction is sufficient to check whether $\Theta$ is in fact a consistent OPT.

\subsection{Checking the consistency}\label{subsec:consistency}
By having set the postulates, we are now provided with the action of the associator $\alpha$ and of all transformations---including the braiding $\T{S}$. Accordingly, we can first derive some fundamental expressions which are needed to check whether $\Theta$ is in fact a consistent OPT.

Sequential composition distributes over sums and is defined in the following way:
\begin{align}
\label{eq:sequential_extended}
\begin{aligned}
\Qcircuit @C=1em @R=1.5em
{
	&\s{A}&\gate{\T{T}_2\T{T}_1}&\s{C}\qw&\\
	&&\s{E}\qw&\qw&
}
\end{aligned}
\coloneqq
\begin{aligned}
\Qcircuit @C=1em @R=1.5em
{
	&\s{A}&\gate{\T{T}_1}&\s{B}\qw&\gate{\T{T}_2}&\s{C}\qw&\\
	&&\s{E}\qw&\qw&\s{E}\qw&\qw&
}
\end{aligned},
\end{align}
for all systems $\sys{A},\sys{B},\sys{C},\sys{E}$ and all $\T{T}_1\in \Transf{A}{B},\T{T}_2\in \Transf{B}{C}$. The rule for extending local transformations can be derived by resorting to the action of the associator $\alpha$, namely, via the following identification:
\begin{align}\label{eq:associativity_extension}
\begin{aligned}
\Qcircuit @C=1em @R=1.5em
{
	&\s{AE}&\gate{\T{T}\boxtimes\T{I}_{\sys{E}}}&\s{BE}\qw&
	\\
	&&\s{D}\qw&\qw&
}
\end{aligned}
\coloneqq
\begin{aligned}
\Qcircuit @C=1em @R=1.5em
{
	&\s{AE}&\multigate{1}{\alpha}&\s{A}\qw&\gate{\T{T}}&\s{B}\qw&\multigate{1}{\alpha^{-1}}&\s{BE}\qw&
	\\
	&\s{D}&\ghost{\alpha}&\qw&\s{ED}\qw&\qw&\ghost{\alpha^{-1}}&\s{D}\qw&
}
\end{aligned}.
\end{align}
Notice that the above expression is a particular instance of Eq.~\eqref{eq:naturality_associator}. The action of the local transformations of the form $\T{I}_\sys{E}\boxtimes\T{T}\in \Transf{EA}{EB}$ can be derived by posing the following expression:
\begin{align}\label{eq:low_transformations}
\begin{aligned}
\Qcircuit @C=1em @R=2em
{
	&&\s{E}\qw&\qw&
	\\
	&\s{A}&\gate{\T{T}}&\s{B}\qw&
}
\end{aligned}
\coloneqq
\begin{tikzpicture}
	\begin{pgfonlayer}{nodelayer}
		\node [style=none] (0) at (-1, -1) {};
		\node [style=none] (1) at (-1, 1) {};
		\node [style=none] (2) at (1.5, 1) {};
		\node [style=none] (3) at (1.5, -1) {};
		\node [style=none] (4) at (0, -0.25) {};
		\node [style=none] (5) at (0.5, 0.25) {};
		\node [style=none] (29) at (1.5, 1) {};
		\node [style=none] (30) at (1.5, -1) {};
		\node [style=none] (38) at (-1, 1) {};
		\node [style=none] (39) at (-1, -1) {};
		\node [style=none] (44) at (1.5, 1) {};
		\node [style=none] (45) at (1.5, -1) {};
		\node [style=none] (46) at (2, 1) {};
		\node [style=none] (47) at (2, -1) {};
		\node [style={system_label}] (48) at (-4, 1.75) {$\scriptstyle \sys  A$};
		\node [style=none] (50) at (-1, -1) {};
		\node [style=none] (51) at (-1, 1) {};
		\node [style=rectangle] (52) at (-2.75, 1) {$\T{T}$};
		\node [style={system_label}] (54) at (-2.75, -0.25) {$\scriptstyle \sys  E$};
		\node [style={system_label}] (55) at (-1.5, 1.75) {$\scriptstyle \sys  B$};
		\node [style={system_label}] (56) at (2, 1.75) {$\scriptstyle \sys  E$};
		\node [style={system_label}] (57) at (2, -0.25) {$\scriptstyle \sys  B$};
		\node [style=none] (58) at (-4.5, 1) {};
		\node [style=none] (59) at (-4.5, -1) {};
		\node [style=none] (60) at (-7, -1) {};
		\node [style=none] (61) at (-7, 1) {};
		\node [style=none] (62) at (-4.5, 1) {};
		\node [style=none] (63) at (-4.5, -1) {};
		\node [style=none] (64) at (-5.5, -0.25) {};
		\node [style=none] (65) at (-6, 0.25) {};
		\node [style=none] (66) at (-4.5, 1) {};
		\node [style=none] (67) at (-4.5, -1) {};
		\node [style=none] (70) at (-7.5, 1) {};
		\node [style=none] (71) at (-7.5, -1) {};
		\node [style=none] (72) at (-7, 1) {};
		\node [style=none] (73) at (-7, -1) {};
		\node [style={system_label}] (74) at (-7.5, 1.75) {$\scriptstyle \sys  E$};
		\node [style={system_label}] (75) at (-7.5, -0.25) {$\scriptstyle \sys  A$};
	\end{pgfonlayer}
	\begin{pgfonlayer}{edgelayer}
		\draw [in=-180, out=0, looseness=0.75] (1.center) to (3.center);
		\draw [in=-135, out=0] (0.center) to (4.center);
		\draw [in=180, out=45] (5.center) to (2.center);
		\draw (44.center) to (46.center);
		\draw (45.center) to (47.center);
		\draw (52) to (51.center);
		\draw (58.center) to (52);
		\draw (70.center) to (72.center);
		\draw (71.center) to (73.center);
		\draw [in=-180, out=0, looseness=0.75] (73.center) to (66.center);
		\draw [in=135, out=0] (72.center) to (65.center);
		\draw [in=-180, out=-45] (64.center) to (67.center);
		\draw (67.center) to (50.center);
	\end{pgfonlayer}
\end{tikzpicture}.
\end{align}
We notice that the latter is a particular instance of Eq.~\eqref{eq:braiding_naturality_diagram}. Parallel composition of transformations can be derived by exploiting Eqs.~\eqref{eq:sequential_extended},~\eqref{eq:associativity_extension}, and~\eqref{eq:low_transformations}, and then posing the following expression:
\begin{align}
\begin{aligned}\label{eq:parallel_composition}
\Qcircuit @C=1em @R=1.5em
{
	&\s{AC}&\gate{\T{T}\boxtimes\T{T}'}&\s{BD}\qw&
}
\end{aligned}
\coloneqq
\begin{aligned}
\Qcircuit @C=1em @R=1.5em
{
	&&\s{A}\qw&\gate{\T{T}}&\s{B}\qw&\\
	&\s{C}&\gate{\T{T}'}&\s{D}\qw&\qw&
}
\end{aligned}.
\end{align}
We can now proceed by verifying whether a theory $\Theta$ is a consistent OPT.
\paragraph*{\textbf{Identity process.}} As long as $\T{I}\subseteq\RevTransfA{\Theta}$, in the light of Eqs.~\eqref{eq:low_transformations} and~\eqref{eq:parallel_composition}, also Eq.~\eqref{eq:identity} is always satisfied.

\paragraph*{\textbf{Equations~\eqref{eq:bifunctoriality_diagrammatic} and~\eqref{eq:braiding_naturality_diagram}.}} Using Eq.~\eqref{eq:low_transformations}, one can check whether the following hold:
\begin{align}\label{eq:local_commute}
&\begin{aligned}
\Qcircuit @C=1em @R=1em
{
	&\s{A}&\gate{\T{T}}&\s{B}\qw&\qw&
	\\
	&&\s{C}\qw&\gate{\T{T}'}&\s{D}\qw&
}
\end{aligned}
=
\begin{aligned}
\Qcircuit @C=1em @R=1em
{
	&&\s{A}\qw&\gate{\T{T}}&\s{B}\qw&\\
	&\s{C}&\gate{\T{T}'}&\s{D}\qw&\qw&
}
\end{aligned},\\
&\begin{tikzpicture}
	\begin{pgfonlayer}{nodelayer}
		\node [style={system_label}] (48) at (-4, 1.75) {$\scriptstyle \sys  A$};
		\node [style=rectangle] (52) at (-2.75, 1) {$\T{T}$};
		\node [style={system_label}] (54) at (-2.75, -0.25) {$\scriptstyle \sys  E$};
		\node [style={system_label}] (55) at (-1.5, 1.75) {$\scriptstyle \sys  B$};
		\node [style={system_label}] (56) at (2, 1.75) {$\scriptstyle \sys  E$};
		\node [style={system_label}] (57) at (2, -0.25) {$\scriptstyle \sys  B$};
		\node [style={system_label}] (74) at (-7.5, 1.75) {$\scriptstyle \sys  E$};
		\node [style={system_label}] (75) at (-7.5, -0.25) {$\scriptstyle \sys  A$};
		\node [style=none] (76) at (1.5, 1) {};
		\node [style=none] (77) at (1.5, -1) {};
		\node [style=none] (78) at (-1, -1) {};
		\node [style=none] (79) at (-1, 1) {};
		\node [style=none] (80) at (1.5, 1) {};
		\node [style=none] (81) at (1.5, -1) {};
		\node [style=none] (82) at (0.5, -0.25) {};
		\node [style=none] (83) at (0, 0.25) {};
		\node [style=none] (84) at (1.5, 1) {};
		\node [style=none] (85) at (1.5, -1) {};
		\node [style=none] (87) at (-1, -1) {};
		\node [style=none] (88) at (-1, 1) {};
		\node [style=none] (89) at (-1, -1) {};
		\node [style=none] (90) at (-7, -1) {};
		\node [style=none] (91) at (-7, 1) {};
		\node [style=none] (92) at (-4.5, 1) {};
		\node [style=none] (93) at (-4.5, -1) {};
		\node [style=none] (94) at (-6, -0.25) {};
		\node [style=none] (95) at (-5.5, 0.25) {};
		\node [style=none] (96) at (-4.5, 1) {};
		\node [style=none] (97) at (-4.5, -1) {};
		\node [style=none] (98) at (-7, 1) {};
		\node [style=none] (99) at (-7, -1) {};
		\node [style=none] (100) at (-4.5, 1) {};
		\node [style=none] (101) at (-4.5, -1) {};
		\node [style=none] (103) at (-1, -1) {};
		\node [style=none] (104) at (-7, -1) {};
		\node [style=none] (105) at (-7, 1) {};
		\node [style=none] (106) at (-7.5, -1) {};
		\node [style=none] (107) at (-7.5, 1) {};
		\node [style=none] (108) at (2, -1) {};
		\node [style=none] (109) at (2, 1) {};
	\end{pgfonlayer}
	\begin{pgfonlayer}{edgelayer}
		\draw (87.center) to (89.center);
		\draw [in=-180, out=0, looseness=0.75] (89.center) to (84.center);
		\draw [in=135, out=0] (88.center) to (83.center);
		\draw [in=-180, out=-45] (82.center) to (85.center);
		\draw [in=-180, out=0, looseness=0.75] (91.center) to (93.center);
		\draw [in=-135, out=0] (90.center) to (94.center);
		\draw [in=180, out=45] (95.center) to (92.center);
		\draw (101.center) to (103.center);
		\draw (100.center) to (52);
		\draw (52) to (88.center);
		\draw (107.center) to (105.center);
		\draw (106.center) to (104.center);
		\draw (85.center) to (108.center);
		\draw (84.center) to (109.center);
	\end{pgfonlayer}
\end{tikzpicture}
=
\begin{aligned}\label{eq:low_transformations2}
\Qcircuit @C=1em @R=2em
{
	&&\s{E}\qw&\qw&
	\\
	&\s{A}&\gate{\T{T}}&\s{B}\qw&
}
\end{aligned}.
\end{align}
Thus, Eqs.~\eqref{eq:bifunctoriality_diagrammatic} and~\eqref{eq:braiding_naturality_diagram} straightforwardly follow from Eqs.~\eqref{eq:local_commute} and~\eqref{eq:low_transformations2}, using Eqs.~\eqref{eq:identity},~\eqref{eq:sequential_extended},~\eqref{eq:associativity_extension},~\eqref{eq:low_transformations}, and~\eqref{eq:parallel_composition}, along with the associativity of sequential composition. When the OPT is symmetric, Eq.~\eqref{eq:low_transformations2} does not need to be checked, being equivalent to Eq.~\eqref{eq:low_transformations}.

\paragraph*{\textbf{Associativity.}} After verifying whether the associator $\alpha$ is invertible, we can check the pentagon identity~\eqref{eq:pentagon_parallel}.
This is done by using the associator $\alpha$ as the rule for identifying the linearly independent vectors in $\StR{(\left(AB\right)C)D}$ with those in $\StR{(A\left(BC\right))D}$. Since the OPT is required to be associative, the following must hold (see Eq.~\eqref{eq:naturality_associator}):
\begin{align}\label{eq:associativity_transformations}
\begin{aligned}
\Qcircuit @C=1em @R=1em
{
	&\s{A}&\qw&\gate{\T{T}_1}&\qw&\s{B}\qw&\\
	&\s{C}&\qw&\gate{\T{T}_2}&\qw&\s{D}\qw&\\
	&\s{E}&\qw&\gate{\T{T}_3}&\qw&\s{F}\qw&
	\relax\gategroup{1}{4}{3}{4}{2em}{--}
	\relax\gategroup{1}{4}{2}{4}{.8em}{--}
}
\end{aligned}
=
\begin{aligned}
\Qcircuit @C=1em @R=1em
{
	&\s{A}&\qw&\gate{\T{T}_1}&\qw&\s{B}\qw&\\
	&\s{C}&\qw&\gate{\T{T}_2}&\qw&\s{D}\qw&\\
	&\s{E}&\qw&\gate{\T{T}_3}&\qw&\s{F}\qw&
	\relax\gategroup{1}{4}{3}{4}{2em}{--}
	\relax\gategroup{2}{4}{3}{4}{.8em}{--}
}
\end{aligned},
\end{align}
for all $\sys{A},\sys{B},\sys{C},\sys{D},\sys{E},\sys{F}\in\Sys{\Theta}$ and all $\T{T}_1\in\Transf{A}{B},\T{T}_2\in\Transf{C}{D},\T{T}_3\in\Transf{E}{F}$. In fact, using Eqs.~\eqref{eq:identity},~\eqref{eq:association_rule},~\eqref{eq:sequential_extended},~\eqref{eq:associativity_extension},~\eqref{eq:low_transformations}, and~\eqref{eq:parallel_composition}, one verifies that Eq.~\eqref{eq:associativity_transformations} holds by construction.

\paragraph*{\textbf{Braiding.}} One can verify whether the transformations associated with the braiding $\Out{S}$ are invertible and satisfy the two hexagon identities~\eqref{eq:hexagon1_diagram} and~\eqref{eq:hexagon2_diagram}. When the OPT is symmetric, the two hexagon identities~\eqref{eq:hexagon1_diagram} and~\eqref{eq:hexagon2_diagram} are equivalent, and then just one of them needs to be checked.

\paragraph*{\textbf{Compatibility with the probabilistic structure.}} Finally, one can check whether: (i) states separate effects and \emph{vice versa}, (ii) the instruments $\Test{P}{X}\in\InstrAB{I}{I}$ are probability distributions, (iii) the null transformation $\varepsilon_{\sys{A}\to\sys{B}}$ is included in $\Transf{A}{B}$ for every $\sys{A},\sys{B}$, (iv) the coarse-graining operation is well-posed for every instrument, and (v) every instrument $\TestC{T}{X}{A}{B}\boxtimes\InstrC{I}{\star}{E}{E}$ maps preparation-instruments of $\sys{AE}$ to preparation-instruments of $\sys{BE}$ for every system $\sys{A},\sys{B},\sys{E}$.

Once we have done the above checks, we are done. Indeed, the above conditions are exhaustive, being compliant with the definitions and the coherence results given in Sec.~\ref{sec:OPTs} and Subsec.~\ref{subsec:coherence}. We remind at this point that all of the above consistency checks need to be performed extending both sides of every equation with the identity wire of an arbitrary system (see Remark~\ref{rem:transformation_equiv_class}). In the next section, we shall first introduce a theory $\Theta$ in the axiomatic way described in Subsec.~\ref{subsec:postulates}. Subsequently, we show that $\Theta$ is indeed a consistent OPT by resorting to the above-described consistency checks.

\section{\textbf{Bilocal Classical Theory}}\label{sec:BCT}
We now present a classical theory $\tilde{\Theta}$ with entanglement, which we call Bilocal Classical Theory (BCT).
%
%
%

\subsection{Postulates}\label{subsec:postulates_BCT}
\begin{postulate}[{Classicality, convexity, and types of systems}]\label{postulate:classical}
	The theory $\tilde{\Theta}$ is classical and convex. For every integer $D>1$, $\Sys{\tilde{\Theta}}$ contains a type of system having dimension $D$.
\end{postulate}
Classical theories have been defined in Subsec.~\ref{subsec:simplicial}. For every system $\mathrm{A}$, the elements of $\ExtStZ{A}$ are $D_\sys{A}+1$ vertices of a simplex. Moreover, the set of the non-null extremal states $\ExtSt{A}$ coincides with $\PurSt{A}=\lbrace\rket{i}_{\sys{A}}\rbrace_{i=1}^{\D{A}}$. The pure states of every systems are jointly perfectly discriminable (see Property~\ref{prope:jpd}). In addition, by convexity, the set of deterministic states $\StN{A}$ is the convex hull of the pure states.
\begin{postulate}[{Parallel composition of systems and states, associator}]\label{postulate:parallel}
	For any two systems $\sys{A},\sys{B}\in\Sys{\tilde{\Theta}}$, the dimension of the composite system $\sys{AB}$ is given by the following rule:
	\begin{align}\label{eq:dimensions}
	\D{AB} =\D{BA}= \left\{
	\begin{aligned}
	&2\D{A}\D{B},&\mbox{if}\ \sys A\neq\sys I\neq\sys B,\\
	&\D A,&\mbox{if}\ \sys B=\sys I.
	\end{aligned}
	\right.
	\end{align}
	Let $\sys{I}\neq\sys{A},\sys{B},\sys{C}\in\Sys{\tilde{\Theta}}$. Denoting the pure states of any composite system $\sys{AB}$ as $\PurSt{AB}=\lbrace\rket{(ij)_-}_{\sys{AB}},\rket{(ij)_+}_{\sys{AB}} \left.|\right. 1\leq i\leq\D{\sys{A}},1\leq j\leq\D{\sys{B}}\rbrace$, for all states $\rket{i}_\sys{A}\in\PurSt{A},\rket{j}_{\sys{B}}\in\PurSt{B}$ the following parallel composition rule holds:
	\begin{align}\label{eq:state_composition}
	\begin{aligned}
	\Qcircuit @C=1em @R=1em
	{   
		&&\prepareC{i}&\s{A}\qw&
		\\
		&\prepareC{j}&\s{B}\qw&\qw
	}
	\end{aligned}
	=
	\frac{1}{2}\sum_{s=-,+}
	\begin{aligned}
	\Qcircuit @C=1em @R=1em
	{   \multiprepareC{1}{(ij)_{s}}&\s{A}\qw&\\
		\pureghost{(ij)_{s}}&\s{B}\qw&
	}
	\end{aligned}.
	\end{align}
	The associator map is given by the following identification:
	\begin{align}\label{eq:quadripartite}
	((ij)_{s_1}k)_{s_2} = (i(jk)_{s_1s_2})_{s_1},
	\end{align}
	for all local indices $i,j,k$ and signs $s_1,s_2$.
\end{postulate}
Notice that Postulate~\ref{postulate:parallel} complies with the classification of the sets of states in simplicial theories with $n$-local discriminability provided in Ref.~\cite{d2019classical}. In the light of Postulate~\ref{postulate:parallel}, one sees that the information carriers of the theory, despite being classical, compose differently from the ones of CT. In Sec.~\ref{sec:features} we will discuss some consequences of this fact. Accordingly, we will call \emph{bibits} the elementary information carriers of BCT---as a shorthand for \emph{bilocal bits}.
\begin{postulate}[{Reversible transformations}]\label{postulate:reversible}
	For all $\sys{I}\neq\sys{A},\sys{A}'\in\Sys{\tilde{\Theta}}$ with $\D{A}=\D{A'}$, let $\T{R}\in\TransfR{A}{A'}$. Then $\T{R}\in\RevTransf{A}{A'}$ if and only if there exist a permutation $\pi$ of $\D{A}$ elements and a sign $\sigma_i$ such that the following holds for all $\sys{I}\neq\sys{E}\in\Sys{\tilde{\Theta}}$ and $\rket{(ij)_s}\in\PurSt{AB}$:
	\begin{align}\label{eq:reversible_transformations}
	&\begin{aligned}
	\Qcircuit @C=1em @R=1.2em
	{
		\multiprepareC{1}{(ij)_{s}}&\s{A}\qw&\gate{\T{R}}&\s{A'}\qw&
		\\
		\pureghost{(ij)_{s}}&\qw&\s{E}\qw&\qw&
	}
	\end{aligned}
	=
	\begin{aligned}
	\Qcircuit @C=1em @R=1.2em
	{
		&\multiprepareC{1}{\left(\pi(i)j\right)_{\sigma_{i}s}}&\s{A'}\qw&
		\\
		&\pureghost{\left(\pi(i)j\right)_{\sigma_{i}s}}&\s{E}\qw&
	}
	\end{aligned}.
	\end{align}
\end{postulate}
In accordance with Proposition~\ref{prop:pure_to_pure}, the class of reversible transformations is well-posed, since every member---being a permutation of pure states---has an inverse, as it should be by definition. Moreover, as required, $\T{I}\subseteq\RevTransfA{\tilde{\Theta}}$.
\begin{postulate}[{Reversible dilation for arbitrary transformations}]\label{postulate:dilation}
For all $\sys{A},\sys{B}\in\Sys{\tilde{\Theta}}$, a map $\T{T}\in\TransfR{A}{B}$ is contained in $\TransfA{\tilde{\Theta}}$ if and only if $\T{T}$ admits of a reversible dilation as follows:
	\begin{align}\label{eq:dilation}
	\begin{aligned}
	\Qcircuit @C=1em @R=1.3em
	{
		&\s{A}&\gate{\T{T}}&\s{B}\qw&
	}
	\end{aligned}
	=
	\begin{aligned}
	\Qcircuit @C=1em @R=1.3em
	{
		&\prepareC{\Sigma}&\s{B'}\qw&\multigate{1}{\T{R}}&\s{A'}\qw&\measureD{H}&
		\\
		&\s{A}\qw&\qw&\ghost{\T{R}}&\qw&\s{B}\qw&\qw&
	}
	\end{aligned},
	\end{align}
	for some $\T R\in\RevTransf{B'A}{A'B}$ and some $\rket\Sigma_\sys{B'}\in\St{B'},\rbra{H}_{\sys{A'}}\in \Eff{A'}$.
\end{postulate}
The effects of the theory will be defined in Postulate~\ref{postulate:instruments}. In Subsec.~\ref{subsec:characterisation}, we will verify that Postulate~\ref{postulate:dilation} is compatible with Postulate~\ref{postulate:reversible}. Notice that Postulate~\ref{postulate:dilation} is also satisfied by both CT and QT.
\begin{postulate}[{Braiding}]
	Let $\sys{I}\neq\sys{A},\sys{B},\sys{E}\in\Sys{\tilde{\Theta}}$. The braiding $\Out{S}$ of $\tilde{\Theta}$ is given by the family of transformations $\T{S}\subset\RevTransfA{\tilde{\Theta}}$ defined as follows:
	\begin{align}\label{eq:swap}
	\begin{tikzpicture}
	\begin{pgfonlayer}{nodelayer}
		\node [style=none] (0) at (-5, 1.75) {};
		\node [style=none] (1) at (-6, 1) {};
		\node [style=none] (2) at (-6, -1) {};
		\node [style=none] (3) at (-5, -1.75) {};
		\node [style=none] (4) at (-2, -1.75) {};
		\node [style=none] (5) at (-2, 1.75) {};
		\node [style=none] (6) at (-0.75, 1.25) {};
		\node [style=none] (7) at (-0.75, 0) {};
		\node [style=none] (8) at (-2, -1.25) {};
		\node [style=none] (9) at (1.75, 0) {};
		\node [style=none] (10) at (1.75, 1.25) {};
		\node [style=none] (11) at (2.5, -1.25) {};
		\node [style=none] (12) at (2.5, 0) {};
		\node [style=none] (13) at (2.5, 1.25) {};
		\node [style=none] (14) at (0.25, 0.25) {};
		\node [style=none] (15) at (0.75, 1) {};
		\node [style=none] (16) at (-2, 1.25) {};
		\node [style=none] (17) at (-2, 0) {};
		\node [style=none] (18) at (-4, 0) {$((ij)_{s_1}k)_{s_2}$};
		\node [style=none] (19) at (-1.25, 0.5) {$\scriptstyle \sys  B$};
		\node [style=none] (20) at (-1.25, 1.75) {$\scriptstyle \sys  A$};
		\node [style=none] (24) at (0.5, -0.75) {$\scriptstyle \sys  E$};
		\node [style=none] (25) at (2.5, 1.75) {$\scriptstyle \sys  B$};
		\node [style=none] (26) at (2.5, 0.5) {$\scriptstyle \sys  A$};
	\end{pgfonlayer}
	\begin{pgfonlayer}{edgelayer}
		\draw (5.center) to (0.center);
		\draw [in=90, out=180] (0.center) to (1.center);
		\draw (2.center) to (1.center);
		\draw [in=180, out=-90] (2.center) to (3.center);
		\draw (3.center) to (4.center);
		\draw (5.center) to (4.center);
		\draw [in=180, out=0] (6.center) to (9.center);
		\draw [in=-135, out=0, looseness=0.75] (7.center) to (14.center);
		\draw [in=45, out=-180, looseness=0.75] (10.center) to (15.center);
		\draw (16.center) to (6.center);
		\draw (17.center) to (7.center);
		\draw (8.center) to (11.center);
		\draw (9.center) to (12.center);
		\draw (10.center) to (13.center);
	\end{pgfonlayer}
\end{tikzpicture}\ \ 
	=\  \ \begin{tikzpicture}
	\begin{pgfonlayer}{nodelayer}
		\node [style=none] (0) at (-1.75, 1.75) {};
		\node [style=none] (1) at (-2.75, 1) {};
		\node [style=none] (2) at (-2.75, -1) {};
		\node [style=none] (3) at (-1.75, -1.75) {};
		\node [style=none] (4) at (1.75, -1.75) {};
		\node [style=none] (5) at (1.75, 1.75) {};
		\node [style=none] (6) at (2.5, 1.25) {};
		\node [style=none] (7) at (2.5, 0) {};
		\node [style=none] (8) at (1.75, -1.25) {};
		\node [style=none] (9) at (1.75, 1.25) {};
		\node [style=none] (10) at (1.75, 0) {};
		\node [style=none] (11) at (-0.5, 0) {$((ji)_{s_1}k)_{s_1s_2}$};
		\node [style=none] (12) at (2.5, 0.5) {$\scriptstyle \sys A$};
		\node [style=none] (13) at (2.5, 1.75) {$\scriptstyle \sys B$};
		\node [style=none] (14) at (2.5, -1.25) {};
		\node [style=none] (15) at (2.5, -0.75) {$\scriptstyle \sys E$};
	\end{pgfonlayer}
	\begin{pgfonlayer}{edgelayer}
		\draw (5.center) to (0.center);
		\draw [in=90, out=180] (0.center) to (1.center);
		\draw (2.center) to (1.center);
		\draw [in=180, out=-90] (2.center) to (3.center);
		\draw (3.center) to (4.center);
		\draw (5.center) to (4.center);
		\draw (9.center) to (6.center);
		\draw (10.center) to (7.center);
		\draw (8.center) to (14.center);
	\end{pgfonlayer}
\end{tikzpicture}\ \ .
	\end{align}
\end{postulate}
It is clear that, for every pair $\sys{A},\sys{B}\in\Sys{\tilde{\Theta}}$, $\tS_{\sys{B},\sys{A}}=\tS^{-1}_{\sys{A},\sys{B}}$ holds, namely, the theory $\tilde{\Theta}$ is symmetric.
\begin{postulate}[{Preparation- and observation-instruments}]\label{postulate:instruments}
	The preparation-instruments of every system $\sys A\in\Sys{\tilde{\Theta}}$ contains a collection $\{\rket{\rho_i}_\sys{A}\}_{i\in\Out{I}}$ of states of $\sys A$ if and only if $\sum_{i\in\Out{I}}\rbraket{e}{\rho_i}=1$. The observation-instruments of every system $\sys A\in\Sys{\tilde{\Theta}}$ are all the collections $\{\rbra{a_x}_{\sys{A}}\}_{x\in\Out{X}}\subset\EffR{A}$ of generalised effects of $\sys{A}$ such that $\{\rbra{a_x}_{\sys{A}}\boxtimes\T{I}_{\sys{E}}\}_{x\in\Out{X}}$ maps preparation-instruments of $\sys{AE}$ to preparation-instruments of $\sys{E}$ for all $\sys{E}\in\Sys{\tilde{\Theta}}$.
\end{postulate}
Notice that the first part of Postulate~\ref{postulate:instruments} is well-posed, since, by Postulate~\ref{postulate:classical}, $\St{A}$ is defined for every system $\sys{A}\in\Sys{\tilde{\Theta}}$. As we will see in Subsec.~\ref{subsec:characterisation}, also the nature of effects of a classical system is compatible with the second part of Postulate~\ref{postulate:instruments}. For every system $\sys{A}$, $\EffR{A}$ is defined by Postulate~\ref{postulate:classical} (via Property~\ref{prope:jpd}).

\subsection{Characterisation}\label{subsec:characterisation}
We will use some of the following characterisation results to prove the coherence of the theory. Accordingly, in order to prove them, we will solely make use of the postulates and of the linear structure.

\subsubsection*{{Bilocal tomography and entangled states}}
\begin{proposition}[\textbf{BCT is strictly bilocal-tomographic}]\label{prop:bilocal}
	BCT satisfies Property~\ref{prope:n_local} if and only if $n\geq 2$.
\end{proposition}
\Proof
We shall prove that BCT satisfies bilocal discriminability but not local discriminability---namely, that BCT satisfies strict bilocal discriminability---via Theorem~\ref{thm:bilocal}. On the one hand, the composition rule~\eqref{eq:dimensions} on the dimensions clearly violates Eq.~\eqref{eq:dimensions_local}. On the other hand, it is straightforward to verify, by direct inspection, that the rule~\eqref{eq:dimensions} satisfies Eq.~\eqref{eq:dimensions_bilocal}.
\qed
As it has been already recalled, any OPT without local discriminability necessarily has entangled states~\cite{d2019classical}. Indeed, all the pure states of a composite system in our theory are entangled.

\subsubsection*{{Conditioning of entangled states and effects, classification of the effects of the theory}}
For given systems $\sys{A},\sys{B}\in\Sys{BCT}$, we did not postulate the action of the local effects of $\sys{A}$ on the bipartite entangled states of $\sys{AB}$. The reason is that this action is not independent from Postulates~\ref{postulate:classical} and~\ref{postulate:parallel}, and can be actually derived from them. This is done by using two facts following from the definition of a classical theory (see Def.~\ref{def:classical_theory}): (i) uniqueness of the decomposition of states into pure states in a simplicial theory, and (ii) the joint perfect discriminability of the pure states (see Property~\ref{prope:jpd}). Furthermore, using the fact that states are separating for effects, we can also derive the action of the local states of $\sys{A}$ on the bipartite entangled effects of $\sys{AB}$. Let $\rbra{i'}_{\sys A}\in\Eff{A},\rbra{(i'j')_{s'}}_{\sys{AB}}\in\Eff{AB}$ denote the effects such that, for any $\rket{i}_{\sys A}\in\PurSt{A},\rket{(ij)_s}_{\sys{AB}}\in\PurSt{AB}$, one has:
\begin{align*}
&\rbraket{i'}{i}_{\sys{A}}=\delta_{i'i},\\
&\rbraket{(i'j')_{s'}}{(ij)_s}_{\sys{AB}}=\delta_{i'i}\delta_{j'j}\delta_{s's}.
\end{align*}
Their existence is guaranteed by Postulate~\ref{postulate:classical} via Property~\ref{prope:jpd}. Then, for all systems $\sys{A},\sys{B}\in\Sys{BCT}$ and states $\rket{i}_\sys{A}\in\PurSt{A},\rket{(ij)_s}_\sys{AB}\in\PurSt{AB}$, the following holds for every $i'$, $j'$ and $s'$:
\begin{align}
\label{eq:steering}\begin{aligned}
\Qcircuit @C=1em @R=1em
{   &\multiprepareC{1}{(ij)_{s}}&\s{A}\qw&\measureD{i'}&\\
	&\pureghost{(ij)_{s}}&\qw&\s{B}\qw&\qw&
}
\end{aligned}
&=
\ \delta_{ii'}\begin{aligned}
\Qcircuit @C=1em @R=1em
{   \prepareC{j}&\s{A}\qw&
}
\end{aligned},
\\[2.5ex]
\label{eq:steering_state}\begin{aligned}
\Qcircuit @C=1em @R=1em
{
	&\prepareC{i}&\s{A}\qw&\multimeasureD{1}{(i'j')_{s'}}&\\
	&\s{B}\qw&\qw&\ghost{(i'j')_{s'}}&
}
\end{aligned}
&=
\ \delta_{ii'}\frac{1}{2}\ 
\begin{aligned}
\Qcircuit @C=1em @R=1em
{   
	\s{A}&\measureD{j'}&
}
\end{aligned}.
\end{align}
A functional on the states of a classical system in CT is an effect if and only if it is a conic combination of those functionals, that perfectly discriminate pure states (see Property~\ref{prope:jpd}), which maps states to (generally subnormalized) probability distributions. As a consequence of Eq.~\eqref{eq:steering}, one can easily verify that all the functionals that would correspond to effects of a classical system comply with the requirement of Postulate~\ref{postulate:instruments}. It is also easy to verify that actually these are the only effects of BCT. Thus, Postulate~\ref{postulate:classical} 
and Postulate~\ref{postulate:instruments} are compatible.

\subsubsection*{{Classification of transformations, operational realisation scheme for arbitrary instruments}}
We did not explicitly provide the class $\InstrA{BCT}$. The reason is that the class of BCT's instruments can be actually derived from the postulates of the theory. The following results---whose proof is given in Appendices~\ref{app:atomic} and~\ref{app:BCT_no-restriction}---classify the transformations of the theory, and provide an operational reversible dilation scheme for arbitrary instruments.
\begin{proposition}[\textbf{Transformations in BCT}]\label{prop:atomic_conclusive}
	Let $\sys{I}\neq\sys{A}\in\Sys{BCT}$ and $\sys{B}\in\Sys{BCT}$. Then the following holds.
	\begin{itemize}
		\item \textbf{Atomic transformations}. $\T{A}\in\TransfR{A}{B}$ is an atomic transformation if and only if $\T{A}\boxtimes\T{I}_{\sys{E}}$ is of the following form for every $\sys{E}\in\Sys{BCT}$:
		\begin{align}\label{eq:atomic2}
		\begin{aligned}
		\Qcircuit @C=1em @R=1.3em
		{
			\multiprepareC{1}{(ij)_{s}}&\s{A}\qw&\gate{\T{A}}&\s{B}\qw&\\
			\pureghost{(ij)_{s}}&\qw&\s{E}\qw&\qw&
		}
		\end{aligned}
		=
		\lambda\delta_{ii_0}\begin{aligned}
		\Qcircuit @C=1em @R=1.3em
		{
			\multiprepareC{1}{(lj)_{\tau s}}&\s{B}\qw&
			\\
			\pureghost{(lj)_{\tau s}}&\s{E}\qw&
		}
		\end{aligned},
		\end{align}
		for some $\lambda\in[0,1]$, $1\leq i_0\leq D_\sys A$, $1\leq l\leq D_\sys B$, and $\tau=\pm$, when $\sys B\neq\sys I$, and
		\begin{align}\label{eq:atomiceff}
		\begin{aligned}
		\Qcircuit @C=1em @R=1.3em
		{
			\multiprepareC{1}{(ij)_{s}}&\s{A}\qw&\measureD{a}&\\
			\pureghost{(ij)_{s}}&\qw&\s{E}\qw&\qw&
		}
		\end{aligned}
		=
		\lambda\delta_{ii_0}\begin{aligned}
		\Qcircuit @C=1em @R=1.3em
		{
			\prepareC{j}&\s{E}\qw&
		}
		\end{aligned},
		\end{align}
		for some $\lambda\in[0,1]$, $1\leq i_0\leq D_\sys A$ when $\sys B=\sys I$.
		\item \textbf{Arbitrary transformations.} Let $\T{T}\in\TransfR{A}{B}$. Then 	$\T{T}\in\Transf{A}{B}$ if and only if, for every $\sys E\in\Sys{BCT}$, $\T{T}\boxtimes\T{I}_{\sys{E}}$ is a conical combination of elements $\T{A}\boxtimes\T{I}_{\sys{E}}$ of the form~\eqref{eq:atomic2} when $\sys B\neq\sys I$, or of elements $a\boxtimes\T{I}_{\sys E}$ of the form~\eqref{eq:atomiceff} when $\sys B=\sys I$,  that maps $\St{AE}$ to $\St{BE}$.
		\item\textbf{Deterministic transformations.} $\T{D}\in\TransfN{A}{B}$ if and only if there exists a reversible dilation for $\T{D}$ of the form~\eqref{eq:dilation} with $\rket\Sigma_\sys{B'}\in\StN{B'}$ and $\rbra{H}_{\sys{A'}}=\rbra{e}_{\sys{A'}}$.
	\end{itemize}
\end{proposition}
First of all, by the classification given in Proposition~\ref{prop:atomic_conclusive}, one can now easily verify that Postulate~\ref{postulate:dilation} is compatible with Postulate~\ref{postulate:reversible}. Indeed, if a transformation $\T{R}$ is reversible 
(namely it is invertible and its inverse is itself a transformation),
using Eq.~\eqref{eq:arbitrary_transformation} in Appendix~\ref{app:atomic} and remembering that a reversible transformation is deterministic and preserves purity and atomicity (see Proposition~\ref{prop:pure_to_pure}), one can prove that $\T{R}$ is of the form~\eqref{eq:reversible_transformations}. On the other hand, every reversible transformation trivially admits a dilation of the form~\eqref{eq:dilation}. Notice that, in the light of Eqs.~\eqref{eq:atomic2} and~\eqref{eq:atomiceff}, the atomic transformations of BCT retain atomicity under sequential composition, and that those transformations which are not effects preserve entanglement whenever their action is not vanishing. In a sense, they are the bilocal-tomographic counterpart of those of CT---however, they are not measure-and-prepare as in CT. We will see some relevant consequences of this fact in Sec.~\ref{sec:features}.

The following result provides a classification of the instruments of BCT.
\begin{proposition}[\textbf{Reversible dilation for the instruments in BCT}]\label{prop:dilation_tests}
	For every $\sys{A},\sys{B}\in\Sys{BCT}$ there exist some systems $\sys{A'},\sys{B'}\in\Sys{BCT}$ and a reversible transformation $\tilde{\T{R}}_{\sys{A,B}}\in\RevTransf{B'A}{A'B}$ such that the following holds. Let $\InstrC{T}{X}{A}{B}\subset\TransfR{A}{B}$. The following conditions are equivalent:
	\begin{enumerate}[label=(\roman*)]
		\item\label{item:1} $\InstrC{T}{X}{A}{B}\boxtimes\InstrC{I}{\star}{E}{E}$ maps preparation-instruments of $\sys{AE}$ to preparation-instruments of $\sys{BE}$ for all $\sys{E}\in\Sys{BCT}$;
		\item\label{item:2} $\InstrC{T}{X}{A}{B}\in\InstrAB{A}{B}$;
		\item\label{item:3} 
		There exist a deterministic state $\rket{\Sigma}_{\sys{B'}}\in\StN{B'}$ and an observation-instrument $\InstrC{a}{X}{A'}{I}$,
		such that:
		\begin{align}\label{eq:dilation_instruments}
		\begin{aligned}
		\Qcircuit @C=1em @R=1.3em
		{
			&\s{A}&\gate{\Test{T}{X}}&\s{B}\qw&
		}
		\end{aligned}
		=
		\begin{aligned}
		\Qcircuit @C=1em @R=1.3em
		{
			&\prepareC{\Sigma}&\s{B'}\qw&\multigate{1}{\tilde{\T{R}}_{\sys{A,B}}}&\s{A'}\qw&\measureD{\Test{a}{X}}&
			\\
			&\s{A}\qw&\qw&\ghost{\tilde{\T{R}}_{\sys{A,B}}}&\qw&\s{B}\qw&\qw&
		}
		\end{aligned}.
		\end{align}
	\end{enumerate}
\end{proposition}
Consider Eq.~\eqref{eq:dilation_instruments}: $\rket{\Sigma}_{\sys{B'}}\in\StN{B'}$ and $\T{R}$ is deterministic by Proposition~\ref{prop:pure_to_pure}. Accordingly, every instrument of BCT can be realised by a channel followed by an observation-instrument. Notice that a reversible dilation scheme similar to~\eqref{eq:dilation_instruments} holds in CT and QT as well. However, differently from the case of CT and BCT, in QT the system-sizes of the ancillae $\sys{B'}$ and $\sys{A'}$ depend not only on the input and output systems $\sys{A}$ and $\sys{B}$, but also on the instrument $\Test{T}{X}$ itself.

As a first corollary of Proposition~\eqref{prop:dilation_tests}, BCT satisfies the following general property for a probabilistic theory.
\begin{property}[{Unrestricted class of instruments}]\label{prope:unrestricted_instruments}
	Let $\Theta$ be an OPT. For all systems $\sys{A},\sys{B}\in\Sys{\Theta}$, the class $\InstrA{\Theta}$ includes a collection $\InstrC{T}{X}{A}{B}\subset\TransfA{\Theta}$ of transformations if and only if $\InstrC{T}{X}{A}{B}\boxtimes\InstrC{I}{\star}{E}{E}$ maps preparation-instruments of $\sys{AE}$ to preparation-instruments of $\sys{BE}$ for all $\sys{E}\in\Sys{\Theta}$.
\end{property}

\subsubsection*{{No-restriction hypothesis, causality and conditional instruments}}
Both Postulate~\ref{postulate:instruments} and Property~\ref{prope:unrestricted_instruments} have a different content from the (possible formulations of the) no-restriction hypothesis (see Property~\ref{prope:no_restriction} and what follows). However, an immediate consequence of Proposition~\ref{prop:dilation_tests}---in particular, by implication $\ref{item:1}\Rightarrow\ref{item:2}$---is that BCT satisfies also the no-restriction hypothesis (see Property~\ref{prope:no_restriction}). In fact, Property~\ref{prope:unrestricted_instruments} follows from Property~\ref{prope:no_restriction}.

Moreover, as we already mentioned (see Subsec.~\ref{subsec:simplicial}), in Ref.~\cite{d2019classical} it is proven that every simplicial theory is causal (see Property~\ref{prope:causality}). It follows that BCT is also causal.
From now on, we will denote the unique deterministic effect of any system $\sys{A}\in\Sys{BCT}$ by $\rbra{e}_{\sys{A}}\equiv\sum_{i=1}^{\D{A}}\rbra{i}_{\sys{A}}$, where $\lbrace\rbra{i}_{\sys{A}}\rbrace_{i=1}^{\D{A}}$ is an observation-instrument jointly perfectly discriminating the pure states of $\sys{A}$. Furthermore, BCT is a convex theory by Postulate~\ref{postulate:classical}. Accordingly, $\StN{A}$ coincides with the convex hull of $\PurSt{A}$ for every system $\sys{A}\in\Sys{BCT}$, namely, every mixture of pure states can be \emph{deterministically prepared} in the theory. More importantly, by Theorem~\ref{thm:conditional} (see Appendix~\ref{app:conditional}), every causal theory satisfying Postulate~\ref{postulate:instruments} and Property~\ref{prope:unrestricted_instruments} contains all conditional instruments (see Property~\ref{postulate:conditional}). Thus, BCT also enjoys this important property.

\subsection{Consistency check}
In order to check the consistency of the theory, we follow the procedure established in Subsec.~\ref{subsec:coherence}. We remind that, by Remark~\ref{rem:transformation_equiv_class}, all the consistency equations must be verified extending both sides of every equation with the identity transformation of an arbitrary system, namely, Eq.~\eqref{eq:identity_extended} holds.
%
%

By Eq.~\eqref{eq:quadripartite}, the associator of the theory is invertible. We now verify the pentagon identity~\eqref{eq:pentagon_parallel} for state-composition via consecutive applications of Eq.~\eqref{eq:quadripartite} on the pure states of a pentapartite system. On the one hand, one has:
\begin{align*}
	((((ij)_{s_1}k)_{s_2}l)_{s_3}m)_{s_4}=((((ij)_{s_1}(kl)_{s_2s_3})_{s_2}m)_{s_4}=\\=((i(j(kl)_{s_2s_3})_{s_1s_2})_{s_1}m)_{s_4}.
\end{align*}
On the other hand, one also has:
\begin{align*}
	((((ij)_{s_1}k)_{s_2}l)_{s_3}m)_{s_4}&
	=(((i(jk)_{s_1s_2})_{s_1}l)_{s_3}m)_{s_4}=\\
	&=((i((jk)_{s_1s_2}l)_{s_1s_3})_{s_1}m)_{s_4}=\\
	&=((i(j(kl)_{s_2s_3})_{s_1s_2})_{s_1})m)_{s_4}.
\end{align*}
In the case of states and effects, Eq.~\eqref{eq:local_commute} can be verified using Eqs.~\eqref{eq:steering} and~\eqref{eq:steering_state}. In the case of arbitrary transformations,
Eq.~\eqref{eq:local_commute} can be easily verified just for the atomic transformations~\eqref{eq:atomic2}, and then extended by linearity in the light of Proposition~\ref{prop:atomic_conclusive}. The family of transformations $\T{S}$ defined in Eq.~\eqref{eq:swap} is manifestly invertible. Moreover, since the theory is symmetric, Eq.~\eqref{eq:low_transformations2} is equivalent to Eq.~\eqref{eq:low_transformations}, and then it does not need to be checked. Finally, for the same reason, one can verify just one of the hexagon identities, e.g.~Eq.~\eqref{eq:hexagon1_diagram}. This is simply done by an iterative application of the associator~\eqref{eq:quadripartite} and of the braiding~\eqref{eq:swap} on the tetrapartite pure states of the theory.

By Postulates~\ref{postulate:classical} and~\ref{postulate:instruments}, and by implications $\ref{item:2}\Leftrightarrow\ref{item:3}$ and $\ref{item:2}\Rightarrow\ref{item:1}$ in Proposition~\ref{prop:dilation_tests}, the following final requirements hold: (i) states separate effects and \emph{vice versa}, (ii) the instruments $\Test{P}{X}\in\InstrAB{I}{I}$ are probability distributions, (iii) the null transformation $\varepsilon_{\sys{A}\to\sys{B}}$ is included in $\Transf{A}{B}$ for every $\sys{A},\sys{B}$, (iv) the coarse-graining operation is well-posed, and (v) every instrument $\TestC{T}{X}{A}{B}\boxtimes\InstrC{I}{\star}{E}{E}$ maps preparation-instruments of $\sys{AE}$ to preparation-instruments of $\sys{BE}$ for every system $\sys{A},\sys{B},\sys{E}$.

Accordingly, the postulates of the theory, along with the classification of the atomic transformations, lead to a straightforward check of its consistency.

\section{\textbf{Features of the theory}}\label{sec:features}
\subsection{Entanglement is independent of complementarity}
As a consequence of Proposition~\ref{prop:atomic_conclusive} and Property~\ref{prope:unrestricted_instruments}, BCT satisfies the \emph{full-information without disturbance principle} (FIWD)~\cite{d2019information}, i.e., every test can be simulated via a non-disturbing test. This happens in spite of the presence of entanglement. In a theory with FIWD the identity transformation cannot be atomic~\cite{d2019information}. Moreover, a theory satisfies FIWD only if the pure states of every system are jointly perfectly discriminable (see Property~\ref{prope:jpd})~\cite{d2019information}. Accordingly, BCT does not admit the existence of incompatible observables---or, in other words,  the theory does not satisfy the \emph{principle of complementarity}. As a first consequence, in BCT it is clearly impossible to violate Bell's Inequalities. However, the theory admits of entangled states: every pure state of a composite system is in fact entangled. The above two features allow one to show that complementarity and entanglement are two independent properties. On the one hand, BCT provides the explicit example of a theory without complementarity, but endowed with entangled states. Conversely, take a modified version of QT whose parallel composition is given by the minimal tensor product on both states and effects---i.e.~the theory where every system is quantum, except that the set of states for every composite system is the convex hull of the product states only. The latter is an example of a theory with complementarity but without entanglement (indeed, complementarity is a single-system property). This shows that complementarity and entanglement are in fact two independent properties in a probabilistic theory. Furthermore, the existence of a unique joint probability distribution for the outcomes of every possible set of measurements implies that BCT is also \emph{noncontextual}.\footnote{For the definition of a generalised-noncontextual ontological model, see Ref.~\cite{schmid2020structure}. Therein, the existence of such a model, using a frame representation which is not overcomplete~\cite{Ferrie_2009}, is proved---under the hypotheses of convexity, causality, and local discriminability---to be equivalent to a property called \emph{simplex-embeddability}, and to the existence of a non-negative quasiprobabilistic model. In fact, the authors also exhibit a counterexample showing that in general (in particular, without local discriminability) the result does not hold.} It is indeed already known that noncontextuality, being a single-system property, is decoupled from entanglement. On the one hand, Spekkens toy model~\cite{PhysRevA.75.032110} is noncontextual, but has entanglement. On the other hand, the abovementioned version of QT with minimal tensor product is contextual, but has no entanglement. However, we remind that Spekkens toy model is not a simplicial theory.

\subsection{Violation of purity of state-composition and of purification, absence of superpositions}\label{subsec:purity_preservation}
In Ref.~\cite{d2019classical} it is proven that, if a simplicial theory violates local discriminability but $n$-local discriminability holds for some $n$, then the theory violates atomicity of parallel composition of states, as well. In BCT, this fact is manifest from the state-composition rule~\eqref{eq:state_composition}, where one sees that the parallel composition of any pair of pure states is not pure, namely, also the principle of \emph{purity of parallel composition of states} is violated. Notice that the principle called \emph{purity preservation} in Ref.~\cite{chiribella2015operational} coincides in fact with atomicity of both parallel and sequential composition (see Remark~\ref{rem:atomicity}). Moreover, state-composition rule~\eqref{eq:state_composition} is clearly a particular instance of the classification of composite states for simplicial theories with $n$-local discriminability provided in Ref.~\cite{d2019classical}. The atomic transformations of BCT (see Proposition~\ref{prop:atomic_conclusive}) are the bilocal-tomographic counterpart of the ones of CT. However, resorting to Proposition~\ref{prop:atomic_conclusive}---in particular using expression~\eqref{eq:atomic2}---one can easily verify that atomicity of parallel composition in BCT is not even satisfied for arbitrary transformations. Interestingly, using again Proposition~\ref{prop:atomic_conclusive}, it is also easy to verify that, differently for the case of states, purity of parallel composition of channels is satisfied. Finally, Proposition~\ref{prop:atomic_conclusive} guarantees that in BCT atomicity and purity are both preserved under \emph{sequential} composition.

Moreover, in BCT no mixed state has a purification, and superpositions of states (in an operational sense) are not admitted. The latter features are a consequence of the no-go results proven in Ref.~\cite{d2019classical} for general simplicial theories.\footnote{However, notice that, in principle, the existence of simplicial theories where at least some mixed states can be purified~\cite{winczewski2018no}, or satisfying an operational formulation of the superposition principle given in Ref.~\cite{d2019classical}, cannot be excluded.} Interestingly, this means that entanglement implies neither the purification nor the superposition principles (nor weaker formulations of these, holding just for a finite number of states). However, BCT satisfies the \emph{(essential) uniqueness of purification principle}~\cite{bookDCP2017}:
\begin{property}[Essential uniqueness of purification]\label{prope:uniqueness_purification}
Let $\Theta$ be an OPT, and $\sys{A},\sys{B}\in\Sys{\Theta}$. If there exist $\rket{\Sigma_1}_{\sys{AB}},\rket{\Sigma_2}_{\sys{AB}}\in\PurSt{AB}$ and $\rbra{\tilde{e}_1}_{\sys{B}},\rbra{\tilde{e}_2}_{\sys{B}}\in\EffN{B}$ such that
\begin{align*}
	&\begin{aligned}
	\Qcircuit @C=1em @R=1em
	{   \multiprepareC{1}{\Sigma_1}&\qw&\s{A}\qw&\qw
	\\
	\pureghost{\Sigma_1}&\s{B}\qw&\measureD{\tilde{e}_1}
	}
	\end{aligned}\ =\ \begin{aligned}
	\Qcircuit @C=1em @R=1em
	{   \multiprepareC{1}{\Sigma_2}&\qw&\s{A}\qw&\qw
	\\
	\pureghost{\Sigma_2}&\s{B}\qw&\measureD{\tilde{e}_2}
	}
	\end{aligned},
\end{align*}
then there exists $\T{R}\in\RevTransf{B}{B}$ such that:
\begin{align*}
\begin{aligned}
\Qcircuit @C=1em @R=1em
{   \multiprepareC{1}{\Sigma_2}&\s{A}\qw&\\
	\pureghost{\Sigma_2}&\s{B}\qw&
}
\end{aligned}=\ \ \begin{aligned}
\Qcircuit @C=1em @R=1em
{   \multiprepareC{1}{\Sigma_1}&\qw&\s{A}\qw&\qw&\\
	\pureghost{\Sigma_1}&\s{B}\qw&\gate{\T{R}}&\s{B}\qw&
}
\end{aligned}.
\end{align*}
\end{property}
In other words, \emph{if there exists} a purification $\rket{\Sigma}_{\sys{AB}}$ of some given deterministic state, then such $\rket{\Sigma}_{\sys{AB}}$ is \emph{essentially unique} in $\sys{AB}$, namely, unique up to reversible local transformations of $\sys{B}$. Notice that Property~\ref{prope:uniqueness_purification}---choosing $\sys{A}=\sys{I}$---implies \emph{transitivity of reversible channels on pure states}. We observe that Property~\ref{prope:uniqueness_purification} is independent of the purification principle, since, given any OPT with local discriminability, one can always restrict: (i) the class of reversible transformations to the identity $\T{I}$ and the braiding $\T{S}$ families (and compositions thereof), and (ii) the class of effects to the separable ones---and this can be done in a consistent way, ending up with a mere prepare-and-measure scenario. Thus, such a prepare-and-measure version of e.g.~QT satisfies the purification principle, but violates essential uniqueness.

\subsection{Dense coding and additivity of classical capacities}
An important communication task that can be performed in BCT is the \emph{dense coding}~\cite{PhysRevLett.69.2881,Werner_2001}. In quantum dense coding, Alice and Bob share a maximally entangled state. Bob detains his qubit and Alice, after performing  on her qubit some local operations in order to encode a two-bit message, sends her local qubit to Bob. The decoding part is then performed by Bob performing a joint measurement in a suitable basis. Via this protocol, Bob manages to decode the message and gains two classical bits of information receiving just one qubit---whose classical capacity, according to Holevo's bound, is one bit. The protocol in BCT retraces the quantum one in the following way. Alice and Bob share a known entangled state, say $\rket{(0b)_{-}}_{\sys{AB}}$, where $\D{A}=\D{B}=2$ and $b$ is an arbitrary local state of the \bibit\  detained by Bob. Now, Alice can encode her two-bit message. She performs the following encoding via entanglement-preserving local operations (recall the reversible transformations in Eq.~\eqref{eq:reversible_transformations}):
\begin{align*}
	&00\mapsto\rket{(0b)_{-}}_{\sys{AB}},\quad 01\mapsto\rket{(0b)_{+}}_{\sys{AB}},\\ &10\mapsto\rket{(1b)_{-}}_{\sys{AB}},\quad 11\mapsto\rket{(1b)_{+}}_{\sys{AB}}.
\end{align*}
Then Alice sends her local \bibit\  to Bob, who can measure the global state, and thus directly decode Alice's message. Notice that, as far as the success of the protocol is concerned, the specific value $b$ of Bob's local state is not relevant. In our BCT dense coding protocol, Alice and Bob share an entangled state, but differently from the quantum case, their local marginal states are always pure. Then, the two-bit message is encoded by Alice into her local bit and into the global degree of freedom (the sign). 
In this way, by sending one \bibit , 
Alice is able communicate two bits of information to Bob. We now have a look at the case when dense coding is realized optimally, i.e.,
with minimal resources. 
In QT, when $\log_2M$ qubits, with $M=d$, are sent from Alice to Bob, a bipartite system with local dimensions $\D{A}=\D{B}=d^2$ is optimal, and the latter case is called \emph{tight}~\cite{Werner_2001}.

In BCT, regardless of Bob's local system's size, the maximum attainable number of distinguished signals is always given by $M=2^{2n-1}$, where $n$ is the number of \bibit s that Alice sends to Bob. Accordingly, by receiving $n$ elementary information carriers, Bob is able to distinguish $2n$ bits, both in the QT and in the BCT protocol.  This means that BCT achieves the same performances as QT in dense coding. 
We remark that, apparently,
BCT seems to exhibit superadditivity of classical capacities, since a bipartite system consisting of two \bibit s---each of which, alone, carries at most one bit---is able to carry three bits. However, if one sticks to the asymptotic definition of classical capacity, every \bibit\ has classical capacity of 2 bits. Indeed, the bits carried by a system of $n$ \bibit s are $2n-1$, and thus the asymptotic capacity is 2 bits per \bibit . Moreover, the same analysis reveals that classical capacity of BCT systems is additive.

The above analysis shows that BCT does not exhibit \emph{hyperdense coding}, i.e., a coding protocol able to overcome the dense coding limit given by QT. Indeed, the possibility of performing a hyperdense coding would imply \emph{superadditive classical capacities}~\cite{massar2015hyperdense}.

\subsection{Violation of entanglement monogamy and of the no-hypersignaling principle}
A quite peculiar feature is that entanglement, in BCT, is not \emph{monogamous}, namely, BCT violates \emph{monogamy of entanglement}. This means that there exist (maximally entangled) states $\rket{\Sigma}_{\sys{ABC}}$ such that the marginal states $\rbra{e}_{\sys{B}}\rket{\Sigma}_{\sys{ABC}},\rbra{e}_{\sys{C}}\rket{\Sigma}_{\sys{ABC}}$ are both maximally entangled. Equivalently, a system can be entangled with more than one other system at the same time. On the contrary, in QT entanglement is monogamous. A violation of entanglement monogamy is known to hold in other theories, such as Real Quantum Theory and Fermionic Quantum Theory~\cite{D_Ariano_2014}. However, in the latter cases those states which violate entanglement monogamy are mixed, while in BCT the violation holds \emph{for every $m$-partite pure state with $m\geq 3$}.

In Ref.~\cite{PhysRevLett.119.020401}, the \emph{no-hypersignaling principle} in probabilistic theories is introduced and analysed. Roughly speaking, the principle states that
any input-output correlation which can be obtained by
transmitting a composite system should also be obtainable by independently transmitting its constituents. For instance, hypersignaling is exhibited if a probabilistic theory, while not contradicting CT and QT at the level of space-like
correlations, displays an anomalous behaviour in its time-like correlations. Both CT and QT satisfy the no-hypersignaling principle. In the case of simplicial theories, a theory $\Theta$ is hypersignaling if and only if there exist $\sys{A},\sys{B}\in\Sys{\Theta}$ such that $\D{AB}>\D{A}\D{B}$ (the technical definition of hypersignaling can be found in Ref.~\cite{PhysRevLett.119.020401}). Namely, a simplicial theory is hypersignaling if and only if it has entanglement. BCT, to the best of our knowledge, is the first example of a complete hypersignaling theory (Ref.~\cite{PhysRevLett.119.020401} analyses the correlations of a model consisting of a two-system scenario). Differently from the model presented in Ref.~\cite{PhysRevLett.119.020401}, in BCT the existence of anomalous behaviours in time-like correlations requires the action of bipartite effects on entangled states. We notice that the no-hypersignaling principle is neither sufficient nor necessary for local discriminability~\cite{PhysRevLett.119.020401}.

\subsection{Entanglement swapping,  cloning/teleportation, and non-null discord states}\label{subsec:entanglement_swapping}
In BCT it is possible to perform \emph{entanglement swapping}~\cite{PhysRevLett.71.4287}, namely the task of transferring entanglement to two remote systems---which are initially uncorrelated---\emph{without interaction}. It is easy to see how this can be done using Eqs.~\eqref{eq:low_transformations},~\eqref{eq:state_composition}, and~\eqref{eq:steering}. Let us pick an arbitrary pure entangled state $\rket{(ij)_s}_{\sys{AB}}$. We want to transfer the entanglement from system $\sys{B}$ to a remote, uncorrelated system $\sys{D}$. Now, fix a chosen pure entangled state $\rket{(kl)_t}_{\sys{CD}}$ shared by system $\sys{D}$ and an ancillary system $\sys{C}$, and perform the following instrument:
\begin{align}\label{eq:entanglement_swapping}
\begin{split}
&\begin{aligned}
\Qcircuit @C=1.5em @R=2em
{
	\multiprepareC{1}{(ij)_{s}}&\qw&\qw&\qw&\s{A}\qw&\qw&\qw&\qw&     \\
	\pureghost{(ij)_{s}}&\s{B}\qw&\qw&\s{B}\qw&\multimeasureD{1}{\lbrace(j'k')_{r}\rbrace_{r=+,-}^{j',k'}}&&&&      \\
	\multiprepareC{1}{(kl)_{t}}&\s{C}\qw&\qw&\s{C}\qw&\ghost{\lbrace(j'k')_{r}\rbrace_{r=+,-}^{j',k'}}&&&&     \\
	\pureghost{(kl)_{t}}&\qw&\qw&\qw&\s{D}\qw&\qw&\qw&\qw&
	\relax\gategroup{2}{4}{3}{6}{2.6em}{--}
}
\end{aligned}
\!\!\!=\\[2.5ex]&=
\lbrace \frac{1}{2}\sum_{s'=+,-}\delta_{r,ss'}\ \begin{aligned}
\Qcircuit @C=1em @R=1em
{
	\prepareC{(il)_{s't}}&\s{AD}\qw&
}
\end{aligned}\rbrace_{r=+,-}\ =
\\[2.5ex]&=
\lbrace \frac{1}{2}\ \begin{aligned}
\Qcircuit @C=1em @R=1em
{
	\prepareC{(il)_{rst}}&\s{AD}\qw&
}
\end{aligned}\rbrace_{r=+,-}\ ,
\end{split}
\end{align}
where the possible outcomes of the instrument are $j,k$ with probability $1$, and $r=+,-$, each with probability $1/2$. As a result, system $\sys{A}$ and system $\sys{D}$ become entangled. This entanglement swapping protocol in BCT clearly retraces the quantum one.

Despite the existence of an entanglement-swapping protocol, the analysis of the teleportation scenario arising from it would make poor sense. Indeed, by the FIWD principle and Property~\ref{postulate:conditional} it is possible to arbitrarily clone every unknown state, namely, teleportation can be reduced to making a copy of a state and sending it to the receiver (e.g.~using a measure-and-prepare channel). For an in-depth discussion about the dependence of teleportation on other properties (such as nonlocality) of a probabilistic theory, we refer the reader to Refs.~\cite{hardy1999disentangling,barnum2012teleportation}.

Finally, the mere existence of entangled states in BCT implies that the theory contains states having \emph{non-null discord}~\cite{PhysRevLett.108.120502} (in an operational sense). However, one may be wondering whether BCT has also some non-null discord separable states. This is actually not the case, due to the FIWD principle and to the absence of delocalized information in any separable state.

\subsection{Programming and information-theoretically insecure cryptography}
In CT it is possible to perform the task of \emph{programming} any desired channel from a system to any another, due to the existence of a \emph{universal processor}. More precisely, CT satisfies the following property.
\begin{property}[\textbf{Programming}]\label{property:programming}
	Let $\Theta$ be a causal OPT. For every pair of systems $\sys{A},\sys{B}\in\Sys{\Theta}$, there exist a system $\sys{P}$ and a channel $\T{P}_{\sys{A},\sys{B}}\in\TransfN{PA}{PB}$ such that the following holds. For every target channel $\T{C}\in\TransfN{A}{B}$, there exists a program state $\rket{\sigma}_{\sys{P}}\in\StN{P}$ such that:
	\begin{align}\label{eq:programming}
	\begin{aligned}
	\Qcircuit @C=1em @R=1.3em
	{
		&\s{A}&\gate{\T{C}}&\s{B}\qw&
	}
	\end{aligned}
	=
	\begin{aligned}
	\Qcircuit @C=1em @R=1.3em
	{
		&\prepareC{\sigma}&\s{P}\qw&\multigate{1}{\T{P}_{\sys{A},\sys{B}}}&\s{P}\qw&\measureD{e}&
		\\
		&\s{A}\qw&\qw&\ghost{\T{P}_{\sys{A},\sys{B}}}&\qw&\s{B}\qw&\qw&
	}
	\end{aligned}.
	\end{align}
\end{property}
As a straightforward corollary of Proposition~\ref{prop:dilation_tests}---in particular, see the dilation scheme~\eqref{eq:dilation_instruments} in~\ref{item:3}---in BCT the task of programming is possible as well.
In QT, despite the \emph{No-programming theorem}, the task of \emph{probabilistic programming} is possible---where the error probability can be made arbitrarily small, provided that \emph{$\D{P}$ (see Eq.~\eqref{eq:programming}) becomes arbitrarily large}~\cite{Kitaev_1997,PhysRevLett.79.321,PhysRevLett.94.090401}.

Finally, one may be wondering whether entanglement in BCT grants information-theoretical security, in a cryptographic scenario, against the attacks of malicious adversaries. For instance, one can think of the possibility to implement protocols such as secure key generation or distribution. It is easy to see why this is not the case. Despite the existence of entangled states with global information, which remains inaccessible unless all the parties collaborate, such protocols in BCT are intrinsically not secure due to the FIWD principle.

\section{\textbf{Discussion}}\label{sec:discussion}
Table~\ref{tab:properties} provides a survey of the operational features of BCT, along with a list of tasks which can and cannot be performed in the theory.
\setlength{\tabcolsep}{10pt}
\renewcommand{\arraystretch}{1}
\begin{table*}[t]
	\begin{tabular}{|p{0.44\linewidth}|p{0.44\linewidth}|}
		\hline
		\textbf{BCT features \color{darkgreen}\ding{51}} & \textbf{BCT does not feature \color{red}\ding{55}}  \\
		\hline
		Causality & Local tomography\\
		 Atomicity and purity of sequential composition
		 &  Atomicity of parallel composition
        \\
        Purity of parallel composition for channels& Purity of parallel composition for states  \\
		Essential uniqueness of purification   & Purification  \\
		Bilocal tomography & Complementarity \\
		Noncontextuality & Monogamy of entanglement \\
		Full-information without disturbance  & No-hypersignaling \\
		Non-null discord states & Secure key generation and distribution \\
		Dense coding  & Hyperdense coding  \\
		Entanglement swapping / Teleportation & Superpositions  \\
		Universal programmability & No-cloning \\
		Transitivity of local reversible channels & Superadditive classical capacity \\
		Joint perfect discriminability of pure states  & \\
		No-restriction hypothesis &   \\
		Strong self-duality & \\
		Pure conditionalization &   \\
		\hline
	\end{tabular}
	\caption{Survey of the operational properties of BCT.}\label{tab:properties}
\end{table*}

We adopted the notion of classicality defined from the perspective of states: a theory is classical if and only if the sets of states for every system are those of the systems of CT---namely, simplicial sets with jointly perfectly discriminable pure states. However, the theory presented also abides by the notion of classicality proposed in Ref.~\cite{schmid2020structure}, admitting of a noncontextual ontological model. In a recent work~\cite{aubrun2019entangleability}, it is proved the non-trivial result that, under the hypothesis of local tomography, given two sets of states $\St{A}$ and $\St{B}$, the composite space $\St{AB}$ can admit of entangled states if and only if neither $\St{A}$ nor $\St{B}$ is a simplex. On the other hand, a theory without local tomography necessarily features entangled states~\cite{d2019classical}. BCT provides the concrete example that indeed two classical systems can give rise to entangled states (at the expense of local tomography). In principle, it is not obvious that this could be done in a consistent way. Our result shows that entanglement and incompatibility of measurements (complementarity) are two independent properties in a theory. Interestingly, BCT thus also shows that a violation of local discriminability does not necessarily imply nonlocality in a Bell-like scenario, namely, stronger-than-classical space-like correlations. As we will discuss in Sec.~\ref{sec:conclusions}, however, it is not clear at the moment whether BCT is local, i.e.~admits of a local ontological model.

We have seen that in BCT entanglement is not monogamous. The same happens in other non-local-tomographic theories, such as Real Quantum Theory and Fermionic Quantum Theory~\cite{Caves:2001aa,D_Ariano_2014}. One may conjecture this property to be a consequence of the violation of local tomography. Another common trait shared by BCT and Fermionic Quantum Theory is the possibility of activating local discrimination of states with entangled ancillary systems~\cite{D_Ariano_2014,lugli2020fermionic},
as can be easily derived from Eq.~\eqref{eq:entanglement_swapping}.

From an axiomatic viewpoint, the existence of BCT allows one to draw some interesting consequences. One can verify that atomicity of parallel composition---being violated by BCT---is independent from: causality, ideal compression, perfect discriminability, and bilocal discriminability~\cite{bookDCP2017}. The same holds for purity of parallel composition of states, since its violation, in this case, is equivalent to the violation of atomicity of state-composition. Remarkably, purity of parallel composition is satisfied for channels and effects (i.e.~excluding states), while atomicity of parallel composition is violated for any kind of transformation. This fact provides a concrete example motivating the distinction between the two notions of purity and atomicity (see Remark~\ref{rem:atomicity}). In passing by, we notice that the principle of \emph{pure conditionalization}~\cite{Wilce:2010aa} is satisfied by BCT. As a consequence, pure conditionalization implies neither atomicity nor purity of parallel state-composition.
Moreover, atomicity (or purity) of parallel state-composition is neither necessary nor sufficient for the FIWD principle. BCT also shows that the presence of entanglement is compatible with the absence of purification for \emph{any} mixed state, or of \emph{any} kind of superposition~\cite{d2019classical}. Finally, BCT shows that the presence of entanglement is compatible with the FIWD principle, namely, entanglement does not imply information-theoretically secure key generation or distribution.

Interestingly, looking at BCT as a \emph{process theory}~\cite{selby2017process}, one realises that the proposed process-theoretic definition of purity~\cite{chiribella2014distinguishability,selby2017leaks} would imply that the theory has no pure state. One may then argue that the above definition is not tenable as a notion of purity in a probabilistic scope---at least when purity of state-composition is violated.

One may be wondering whether, in principle, BCT is the only classical theory satisfying strict bilocal tomography. In Appendix~\ref{app:simplicial} [Theorem~\ref{thm:classification} and Corollary~\ref{cor:classification_bilocal}], we prove that, for any simplicial theory, rules~\eqref{eq:dimensions} and~\eqref{eq:state_composition} come as a consequence of the following assumptions: (a) Homogeneous strict bilocal tomography: the theory is strictly bilocal-tomographic, with no composite system being local-tomographic; (b) Essential uniqueness of purification (Property~\ref{prope:uniqueness_purification}). In particular, assumption (a) alone is not sufficient to select rule~\eqref{eq:dimensions}, since, as a counterexample, one can verify (using Theorem~\ref{thm:bilocal}) that also the rule
\begin{align*}
\begin{split}
\D{AB}&=\D{A}\D{B}+(\D{A}-1)(\D{B}-1)
\end{split}
\end{align*}
satisfies (a), being also associative. How the reversible dynamics [assumption (b), see Property~\ref{prope:uniqueness_purification}] affects probabilities in a theory has been recently explored in Ref.~\cite{galley2020dynamics}. However, there is strong evidence that assumptions (a) and (b) may be not sufficient to single out BCT among classical theories. Indeed, the reversible transformations admitted by the simplicial structure are not in principle exhausted by those defined in Postulate~\ref{postulate:reversible}, and one can actually define disjoint families of transformations which both are reversible and obey to Eq.~\eqref{eq:local_commute}. On the other hand, coherence (see Subsec.~\ref{subsec:coherence}) is also key to single out BCT. For instance, it is easy to verify that some choices of the associator (see Postulate~\ref{postulate:parallel}), despite being invertible, lead to violations of the pentagon identity. As for the choices of the associator $\alpha$ and the braiding $\T{S}$, modifying Postulate~\ref{postulate:reversible} would lead in principle to conceivable choices differing from~\eqref{eq:quadripartite} and~\eqref{eq:swap}. Nevertheless, Postulate~\ref{postulate:reversible}, along with assumptions (a) and (b), would single out precisely those $\alpha$ and $\T{S}$ postulated for BCT. 


Finally, as well as CT and QT, also BCT satisfies both the no-restriction hypothesis (see Property~\ref{prope:no_restriction}) and strong self-duality~\cite{PhysRevLett.108.130401}. Notice that BCT also satisfies the usual formulation of the no-restriction hypothesis---the one which is valid in a context where local tomography holds, see Subsec.~\ref{subsec:causality}---namely, ``all the positive functionals on states are effects''.


\section{\textbf{Conclusions and outlook}}\label{sec:conclusions}
We conclude the present work drawing some final remarks and pointing out some open problems.

In Sec.~\ref{sec:OPTs} we provided a broad-scope presentation of the operational probabilistic framework. In Subsec.~\ref{subsec:preliminary_results} we proved Theorem~\ref{thm:delta3}, providing a characterisation of the composition rules for system-sizes in arbitrary theories, and then we specialised it to strictly bilocal-tomographic theories, proving Theorem~\ref{thm:bilocal}. In Sec.~\ref{sec:constructing_an_opt}, we set the problem of the consistency of an operational probabilistic theory, identifying a set of sufficient conditions to construct a theory, and to verify its well-posedness and coherence. Finally, we have presented BCT---a classical theory featuring entanglement and complete with a non-trivial set of transformations---providing an extensive characterisation of it.

BCT violates the principle of local tomography. However, the degree of holism required by the theory in the task of \emph{state tomography} is still limited, due to the property of (strict) bilocal tomography. What is more, notice that in BCT, in order to perform \emph{process tomography} of any $\T{T}\in\Transf{A}{B}$ [Proposition~\ref{prop:atomic_conclusive}], one just needs a limited set of bipartite pure states, namely, a set $\{\rket{\rho_i}_{\sys{AE}}\}_{i=1}^{\D{A}}\subset\PurSt{AE}$ such that the set of reduced states $\{\rbra{e}_{\sys{E}}\rket{\rho_i}_{\sys{AE}}\}_{i=1}^{\D{A}}=\PurSt{A}$, while the ancilla $\sys{E}\neq\sys{I}$ and all reduced states $\{\rbra{e}_{\sys{A}}\rket{\rho_i}_{\sys{AE}}\}_{i=1}^{\D{A}}$ can be arbitrarily chosen. One may still argue that the principle of $n$-local discriminability [Property~\ref{prope:n_local}] does not sufficiently bind the degree of holism of a physical theory. Nevertheless, there are examples of theories which, while violating local discriminability, have a strong physical motivation, such as Fermionic Quantum Theory. In addition, generic non-local-tomographic theories provide a sandbox to investigate the logical interdependence of physical properties, which is a key aspect in view of formulating new physical theories.

Indeed, BCT provides two important proofs of concept. First, entanglement and complementarity are decoupled, i.e., they are independent properties of a physical theory. Second, the set of states---or of correlations---is not sufficient to determine the full theory. This is by the way close to the spirit of the no-hypersignaling principle~\cite{PhysRevLett.119.020401}. Interestingly, since BCT enjoys the same sets of states of CT, it looks like it cannot violate any device-independent principle~\cite{scarani2012device,Navascues:2015aa}. This leads to the following question: is it even logically possible to formulate device-independent principles ruling out BCT, but not CT?

As the correlations attainable by a physical theory do not determine the theory itself---and, in fact, may even give rise to quite different theories---one may be questioning \emph{what does it mean to be classical}, or, more specifically along the lines of Ref.~\cite{schmid2020structure}, \emph{what does it mean to be classically explainable}. As we extensively discussed in Sec.~\ref{sec:discussion}, we left the problem of an information-theoretic axiomatisation of BCT open. Such an axiomatisation may shed some light on a meaningful notion of classicality even in a context allowing for non-local tomography. For instance, both Real and Fermionic Quantum Theory are superselected versions of QT: might it be the case that BCT arises as the superselection of CT? We conjecture that this is not the case, as we argue in the following.

By direct inspection of the state-composition rule in BCT [Postulate~\ref{postulate:parallel}, Eq.~\eqref{eq:state_composition}], one sees that every product of pure states is the flat statistical mixture of two entangled states carrying a ``delocalized'' variable, i.e.~the sign $s$ in Eq.~\eqref{eq:state_composition}. 
One is tempted to assign an \emph{element of reality} to the global degree of freedom $s$. However, is this intuition correct? In fact, there are two ontologically distinct---although operationally equivalent---preparation procedures for the parallel composition of preparation-tests in BCT: one is to generate the mixture using classical randomness for the composite system; the other one is to simply prepare two local pure states in parallel. Therefore, is it meaningful to say that an entangled classical measurement on a product state reveals a pre-existing value $s$ in both cases? This would amount to say that the product state was entangled in the first place. Surprisingly, such a paradox is reminiscent of the question whether measurements reveal a pre-defined value, albeit in a measurement context which is, in all respects, classical. Being the value measurable without disturbance, the issue seems to be irrelevant: there is no contradiction in thinking of $s$ as an ``element of reality'' in the sense of EPR. However, this position opens a deeper question: can the global degree of freedom $s$ be interpreted as a function of suitably defined local (possibly hidden) variables? We have evidence that the primary difficulty in devising such a model is having hidden variables of finite sizes, i.e., hidden variables which can be stored in a limited memory. This would mean to require that each system cannot retain the values of global data for all the possible choices of remote systems and experiments considered. In Ref.~\cite{PhysRevA.94.052127}, a similar situation was investigated as far as QT is concerned. However, in that case, the hidden register would store the values of the intrinsic properties for all possible sequences of measurements that the observer can perform. The motivation to demand such a property is that abandoning it would allow
for some form of \emph{superdeterminism}, where e.g.~all that (infinite amount of) information, which is relevant to predict the outcomes of any experiment, is possibly stored in each system. Nevertheless, such a requirement seemingly clashes with the action of the associator and braiding of the theory. Might it be that BCT does not admit of a local-realistic hidden variable model, in spite of having classical correlations?\footnote{On top of this, at this stage it is not clear whether a (noncontextual) ontological model for BCT may possibly satisfy \emph{diagram preservation}---namely, the property of preserving the compositional structure of the OPT---or whether it may preserve the linear (convex and coarse-graining) relations of the OPT~\cite{schmid2020structure}.}

Finally, the present study provides an adequate toolbox for the comprehensive construction of complete and consistent physical theories. In particular, Sec.~\ref{sec:constructing_an_opt} sets an exploitable constructive procedure which can be used in a generic context, while in Sec.~\ref{sec:BCT} we provided an explicit concrete application of the latter. This paves the way for a direct employment of the developed techniques in enhancing the investigations on post-quantum theories.

\section*{Acknowledgements}
This paper was made possible through the support of a grant from the John Templeton Foundation, Grant No. 60609 ``Quantum Causal Structures.'' The opinions expressed in this publication are those of the authors and do not necessarily reflect the views of the John Templeton Foundation. M. Erba wishes to thank A. Tosini and M. Pl\'avala for stimulating and insightful discussions.

\bibliography{opt_bib2}

\appendix

\section{Classification of BCT's transformations [Proof of Proposition~\ref{prop:atomic_conclusive}]}\label{app:atomic}
In the present appendix, in order to make sense of expressions of the form $\rket{(l_1l_2)_{u}}_{\sys{S}_1\sys{S}_2}$ with $\sys{S}_1=\sys{I}$ or $\sys{S}_2=\sys{I}$, these can be safely replaced by the expressions $\rket{(*l_2)_{+}}_{\sys{S}_2}$ or $\rket{(l_1*)_{+}}_{\sys{S}_1}$, respectively. Moreover, we will extensively exploit the classification of BCT's effects given in Subsec.~\ref{subsec:characterisation}.
\begin{lemma}\label{lem:reversible}
	Let $\sys{I}\neq\sys{A}\in\Sys{BCT}$, $\sys{B}\in\Sys{BCT}$, and $\T{T}\in\Transf{A}{B}$ be defined as
	\begin{align}\label{eq:atomic_dilation}
	\begin{aligned}
	\Qcircuit @C=1em @R=1.3em
	{
		&\s{A}&\gate{\T{T}}&\s{B}\qw&
	}
	\end{aligned}
	\coloneqq 
	\begin{aligned}
	\Qcircuit @C=1em @R=1.3em
	{
		&\prepareC{k}&\s{B'}\qw&\multigate{1}{\T{R}}&\s{A'}\qw&\measureD{\tilde{k}}&
		\\
		&\s{A}\qw&\qw&\ghost{\T{R}}&\qw&\s{B}\qw&\qw&
	}
	\end{aligned},
	\end{align}
	with $\rket{k}_{\sys{B'}}\in \PurSt{B'},\rbra{\tilde{k}}_{\sys{A'}}\in \Eff{A'}$ such that $\rbraket{k'}{\tilde{k}}_{\sys{A'}}=\delta_{k'\tilde{k}}$ for any $\rket{k'}_{\sys{A'}}\in\PurSt{A'}$, and $\T{R}\in\RevTransf{B'A}{A'B}$ such that for all $\sys{E}\in\Sys{BCT}$:
	\begin{align}\label{eq:reversible_action}
	\begin{aligned}
	\Qcircuit @C=1em @R=1.3em
	{
		\multiprepareC{1}{((ki)_{s_1}j)_{s_2}}&\s{B'A}\qw&\gate{\T{R}}&\s{A'B}\qw&\\
		\pureghost{((ki)_{s_1}j)_{s_2}'}&\qw&\s{E}\qw&\qw&
	}
	\end{aligned}
	\coloneqq
	\begin{aligned}
	\Qcircuit @C=1em @R=1.3em
	{
		\multiprepareC{1}{((\kappa_{s_1}^{ki}\iota_{s_1}^{ki})_{r_{s_1}^{ki}}j)_{t_{s_1}^{ki}s_2}}&\s{A'B}\qw&
		\\
		\pureghost{((\kappa_{s_1}^{ki}\iota_{s_1}^{ki})_{r_{s_1}^{ki}}j)_{t_{s_1}^{ki}s_2}}&\s{E}\qw&
	}
	\end{aligned}.
	\end{align}
	[In Eq.~\eqref{eq:reversible_action}, there is no dependence on $k$ and $s_1$ in the case $\sys{B'}=\sys{I}$, while in the case $\sys{A'}=\sys{I}$ or $\sys{B}=\sys{I}$ we just set $r_{s_1}^{ki}=+$ for all $k,i,s_1$.] Then, for all $\rket{(ij)_s}_{\sys{AE}}\in\St{AE}$, the following holds:
	\begin{align}\label{eq:dilation_action}
	\begin{split}
	&\begin{aligned}
	\Qcircuit @C=1em @R=1.3em
	{
		\multiprepareC{1}{(ij)_{s}}&\s{A}\qw&\gate{\T{T}}&\s{B}\qw&\\
		\pureghost{(ij)_{s}}&\qw&\s{E}\qw&\qw&
	}
	\end{aligned}
	= \\[2.5ex] =&
	\frac{1}{2}\sum_{s'=+,-}\delta_{\tilde{k},\kappa_{s'}^{ki}}
	\begin{aligned}
	\Qcircuit @C=1em @R=1.3em
	{
		\multiprepareC{1}{(\iota_{s'}^{ki}j)_{r_{s'}^{ki}t_{s'}^{ki} s's}}&\s{B}\qw&
		\\
		\pureghost{(\iota_{s'}^{ki}j)_{r_{s'}^{ki}t_{s'}^{ki} s's}}&\s{E}\qw
	}
	\end{aligned}.
	\end{split}
	\end{align}
\end{lemma}
\begin{proof}
	By direct computation, using Postulate~\ref{postulate:parallel} and formula~\eqref{eq:steering}, one has:
	\begin{align*}
	&\begin{aligned}
	\Qcircuit @C=1em @R=1.3em
	{
	\multiprepareC{1}{(ij)_{s}}&\s{A}\qw&\gate{\T{T}}&\s{B}\qw&\\
	\pureghost{(ij)_{s}}&\qw&\s{E}\qw&\qw&
	}
	\end{aligned}
	=\\[2.5ex]=&
	\frac{1}{2}\sum_{s'=+,-}\begin{aligned}
	\Qcircuit @C=1em @R=1.3em
	{
		\multiprepareC{1}{(k(ij)_{s})_{s'}}&\s{B'}\qw&\multigate{1}{\T{R}\boxtimes\T{I}_{\sys{E}}}&\s{A'}\qw&\measureD{\tilde{k}}&\\
		\pureghost{(k(ij)_{s})_{s'}}&\s{AE}\qw&\ghost{\T{R}\boxtimes\T{I}_{\sys{E}}}&\qw&\s{BE}\qw&\qw&
	}
	\end{aligned}
	=\\[2.5ex]=&
	\frac{1}{2}\sum_{s'=+,-}\begin{aligned}
	\Qcircuit @C=1em @R=1.3em
	{
		\multiprepareC{1}{((ki)_{s'}j)_{s's}}&\s{B'A}\qw&\multigate{1}{\T{R}\boxtimes\T{I}_{\sys{E}}}&\s{A'}\qw&\measureD{\tilde{k}}&\\
		\pureghost{((ki)_{s'}j)_{s's}}&\s{E}\qw&\ghost{\T{R}\boxtimes\T{I}_{\sys{E}}}&\qw&\s{BE}\qw&\qw&
	}
	\end{aligned}
	=\\[2.5ex]=&
	\frac{1}{2}\sum_{s'=+,-}\begin{aligned}
	\Qcircuit @C=1em @R=.85em
	{
		\multiprepareC{2}{((\kappa_{s'}^{ki}\iota_{s'}^{ki})_{r_{s'}^{ki}}j)_{t_{s'}^{ki}s's}}&\s{A'}\qw&\measureD{\tilde{k}}&
		\\
		\pureghost{((\kappa_{s'}^{ki}\iota_{s'}^{ki})_{r_{s'}^{ki}}j)_{t_{s'}^{ki}s's}}&\qw&\s{B}\qw&\qw&
		\\
		\pureghost{((\kappa_{s'}^{ki}\iota_{s'}^{ki})_{r_{s'}^{ki}}j)_{t_{s'}^{ki}s's}}&\qw&\s{E}\qw&\qw&
	}
	\end{aligned}
	=\\[2.5ex]=&
	\frac{1}{2}\sum_{s'=+,-}\begin{aligned}
	\Qcircuit @C=1em @R=1.3em
	{
		\multiprepareC{1}{(\kappa_{s'}^{ki}(\iota_{s'}^{ki}j)_{r_{s'}^{ki}t_{s'}^{ki}s's})_{r_{s'}^{ki}}}&\s{A'}\qw&\measureD{\tilde{k}}&
		\\
		\pureghost{(\kappa_{s'}^{ki}(\iota_{s'}^{ki}j)_{r_{s'}^{ki}t_{s'}^{ki}s's})_{r_{s'}^{ki}}}&\qw&\s{BE}\qw&\qw&
	}
	\end{aligned}
	=\\[2.5ex]=&
	\frac{1}{2}\sum_{s'=+,-}\delta_{\tilde{k},\kappa_{s'}^{ik}}
	\begin{aligned}
	\Qcircuit @C=1em @R=1.3em
	{
	\multiprepareC{1}{(\iota_{s'}^{ki}j)_{r_{s'}^{ki}t_{s'}^{ki} s's}}&\s{B}\qw&\\
	\pureghost{(\iota_{s'}^{ki}j)_{r_{s'}^{ki}t_{s'}^{ki} s's}}&\s{E}\qw&
	}
	\end{aligned}.
	\end{align*}
\end{proof}
From Postulate~\ref{postulate:dilation} and Eqs.~\eqref{eq:atomic_dilation},~\eqref{eq:reversible_action}, and \eqref{eq:dilation_action}, posing $\rket{\Sigma}_{\sys{B'}}\coloneqq\sum_{k=1}^{\D{B'}}\alpha_{k}\rket{k}_{\sys{B'}},\rbra{H}_{\sys{A'}}\coloneqq\sum_{\tilde{k}=1}^{\D{A'}}\beta_{\tilde{k}}\rbra{\tilde{k}}_{\sys{A'}}$---where for all $k,\tilde{k}$  one has $0\leq\alpha_{k},\beta_{\tilde{k}}\leq 1$ such that $\sum_{k=1}^{\D{B'}}\alpha_{k}\in[0,1]$ and $\sum_{\tilde{k}=1}^{\D{A'}}\beta_{\tilde{k}}\in[0,\D{A'}]$---the action of an arbitrary transformation $\T{T}\in\Transf{A}{B}$ with $\sys{A}\neq\sys{I}$ has the following form:
\begin{align}\label{eq:arbitrary_transformation}
\begin{split}
&\begin{aligned}
\Qcircuit @C=1em @R=1.3em
{
\multiprepareC{1}{(ij)_{s}}&\s{A}\qw&\gate{\T{T}}&\s{B}\qw&\\
\pureghost{(ij)_{s}}&\qw&\s{E}\qw&\qw&
}
\end{aligned}
=\\[2.5ex]=&
\frac{1}{2}\sum_{k=1}^{\D{B'}}\sum_{s'=+,-}\alpha_{k}\beta_{\kappa_{s'}^{ki}}
\begin{aligned}
\Qcircuit @C=1em @R=1.3em
{
	\multiprepareC{1}{(\iota_{s'}^{ki}j)_{r_{s'}^{ki}t_{s'}^{ki} s's}}&\s{B}\qw&\\
	\pureghost{(\iota_{s'}^{ki}j)_{r_{s'}^{ki}t_{s'}^{ki} s's}}&\s{E}\qw&
}
\end{aligned}.
\end{split}
\end{align}

Let $\sys{I}\neq\sys{A}\in\Sys{BCT}$, $\sys{B}\in\Sys{BCT}$, and $\T{A}$ be a map defined by the following action:
\begin{align}\label{eq:atomic}
\begin{aligned}
\Qcircuit @C=1em @R=1.3em
{
	\multiprepareC{1}{(ij)_{s}}&\s{A}\qw&\gate{\T{A}}&\s{B}\qw&\\
	\pureghost{(ij)_{s}}&\qw&\s{E}\qw&\qw&
}
\end{aligned}
=
\lambda\delta_{ii_0}\begin{aligned}
\Qcircuit @C=1em @R=1.3em
{
	\multiprepareC{1}{(lj)_{\tau s}}&\s{B}\qw&
	\\
	\pureghost{(lj)_{\tau s}}&\s{E}\qw&
}
\end{aligned},
\end{align}
for some $\lambda\in[0,1]$, $1\leq i_0\leq D_\sys A$, $1\leq l\leq D_\sys B$, and $\tau=\pm$. Notice that such functions $\T{A}$ map states to states.
\begin{lemma}[Characterisation of transformations]\label{lem:char_transf}
	Let $\sys{I}\neq\sys{A}\in\Sys{BCT}$, $\sys{B}\in\Sys{BCT}$, and $\T{A}$ be a map defined as in Eq.~\eqref{eq:atomic}. Then every $\T{T}\in\Transf{A}{B}$ is a conical combination of elements of the form~\eqref{eq:atomic}.
\end{lemma}
\begin{proof}
	Let $\sys{I}\neq\sys{A}\in\Sys{BCT}$ and $\sys{B}\in\Sys{BCT}$.
	On the one hand, we know that the action of an arbitrary transformation $\T{T}\in\Transf{A}{B}$ has the form~\eqref{eq:arbitrary_transformation}. On the other hand, the action on $\PurSt{AE}$ of an arbitrary conical combination of generalised transformations of the form~\eqref{eq:atomic} is given by:
	\begin{align}\label{eq:atomic_sum}
	\begin{split}
	&\begin{aligned}
	\Qcircuit @C=1em @R=1.3em
	{
	\multiprepareC{1}{(ij)_{s}}&\s{A}\qw&\gate{\sum_n\T{A}_n}&\s{B}\qw&\\
	\pureghost{(ij)_{s}}&\qw&\s{E}\qw&\qw&
	}
	\end{aligned}
	=\\[2.5ex]=&
	\sum_{m=1}^{\D{B}}\sum_{\tau=+,-}\lambda_{m,\tau}^{(i)}\begin{aligned}
	\Qcircuit @C=1em @R=1.3em
	{
		\multiprepareC{1}{(mj)_{\tau s}}&\s{B}\qw&\\
		\pureghost{(mj)_{\tau s}}&\s{E}\qw&
	}
	\end{aligned},
	\end{split}
	\end{align}
	where the $\lambda_{m,\tau}^{(i)}$ are non-negative real numbers such that $\sum_{m,\tau}\lambda_{m,\tau}^{(i)}\in[0,1]$. Now,  
	for every $1\leq i\leq D_\sys A$ let us define a set $I_{m,\tau}(i)$ such that: $(k,s')\in I_{m,\tau}(i)$ if and only if $\iota_{s'}^{ki}=m$ and $r_{s'}^{ki}t_{s'}^{ki} s'=\tau$.
	For all $i,m,\tau$ one can pose:
	\begin{align*}
	\lambda_{m,\tau}^{(i)}\coloneqq \frac{1}{2}\sum_{(k,s')\in I_{m,\tau}(i)}\alpha_{k}\beta_{\kappa_{s'}^{ki}}.
	\end{align*}
	Clearly, if $I_{m,\tau} (i)= \emptyset$ one has $\lambda_{m,\tau}^{(i)}=0$.
	This proves that Eq.~\eqref{eq:arbitrary_transformation} can be rewritten in form~\eqref{eq:atomic_sum}, namely, every $\T{T}\in\Transf{A}{B}$ is a conical combination of elements of the form~\eqref{eq:atomic}.
\end{proof}
Let $\sys{I}\neq\sys{A}\in\Sys{BCT}$ and $\sys{B}\in\Sys{BCT}$, $h$ be any chosen function from $\D{A}$ elements to $\D{B}$ elements, and $\xi$ any chosen function from $\D{A}$ to $\{+,-\}$ [$\xi$ is set to be identically $+$ if $\sys{B}=\sys{I}$]. Moreover, let $\sys{B''}\in\Sys{BCT}$ such that $\D{B''}=2\D{B}^{\D{A}}$ and $\sys{B}'\coloneqq\sys{B}''\sys{B}$, $\sys{A}'\coloneqq\sys{B}''\sys{A}$. Define $\T{\tilde{R}}_{\sys{A,B}}\in\RevTransf{B'A}{A'B}$ in the following way:
\begin{align}\label{eq:atomic_atomic}
\begin{split}
&\begin{aligned}
\Qcircuit @C=1em @R=1.3em
{
	&\multiprepareC{3}{((\sigma_{h,\xi}k)_{s_1}(ij)_{s_2})_{s_3}}&\s{B''}\qw&\multigate{2}{\T{\tilde{R}}_{\sys{A,B}}}&\s{\sys{B''}}\qw&\\
	&\pureghost{((\sigma_{h,\xi}k)_{s_1}(ij)_{s_2})_{s_3}}&\s{B}\qw&\ghost{\T{\tilde{R}}_{\sys{A,B}}}&\s{\sys{A}}\qw&\\
	&\pureghost{((\sigma_{h,\xi}k)_{s_1}(ij)_{s_2})_{s_3}}&\s{A}\qw&\ghost{\T{\tilde{R}}_{\sys{A,B}}}&\s{B}\qw&\\
	&\pureghost{((\sigma_{h,\xi}k)_{s_1}(ij)_{s_2})_{s_3}}&\qw&\qw&\s{E}\qw&
}
\end{aligned}
\coloneqq\\[2.5ex]
\coloneqq&\begin{aligned}
\Qcircuit @C=1em @R=1em
{
	&\multiprepareC{3}{((\sigma_{h,\xi}i)_{s_1}([h(i)\oplus k]j)_{\xi(i) s_2})_{s_3}}&\s{\sys{B''}}\qw&\\
	&\pureghost{((\sigma_{h,\xi}i)_{s_1}([h(i)\oplus k]j)_{\xi(i) s_2})_{s_3}}&\s{A}\qw&\\
	&\pureghost{((\sigma_{h,\xi}i)_{s_1}([h(i)\oplus k]j)_{\xi(i) s_2})_{s_3}}&\s{B}\qw&\\
	&\pureghost{((\sigma_{h,\xi}i)_{s_1}([h(i)\oplus k]j)_{\xi(i) s_2})_{s_3}}&\s{E}\qw&
}
\end{aligned},
\end{split}
\end{align}
where ${\oplus}$ denotes the sum modulo $\D{B}$. The states $\sigma_{h,\tau}$ in Eq.~\eqref{eq:atomic_atomic} implement every possible pair of functions $(h,\xi)$. It is easy to realise that the transformation $\tilde{\T{R}}_{\sys{A,B}}$ complies indeed with Postulate~\ref{postulate:reversible}. In the remainder of the present appendix and in Appendix~\ref{app:BCT_no-restriction}, we will use $\tilde{\T{R}}_{\sys{A,B}}$ as a universal processor for BCT's transformations.
		\begin{lemma}[Realisation of deterministic transformations]\label{lem:realisation_channels}
		 Let $\sys{I}\neq\sys{A}\in\Sys{BCT}$, $\sys{B}\in\Sys{BCT}$, and $\T{D}\in\TransfR{A}{B}$. Then the following are equivalent:
		\begin{enumerate}[label=(\Alph*)]
			\item\label{item:22} For every $i\in\{1,2,\ldots,\D{A}\}$ there exists a probability distribution $\{\lambda_{m,\tau}^{(i)}\}_{(m,\tau)\in I}$, with $I=\{1,2,\ldots,\D{B}\}\times\{+,-\}$, such that the following holds for all $\rket{(ij)_s}_{\sys{AE}}\in\PurSt{AE}$:
			\begin{align}\label{eq:channel2}
			\begin{aligned}
			\Qcircuit @C=1em @R=1.3em
			{
				\multiprepareC{1}{(ij)_{s}}&\s{A}\qw&\gate{\T{D}}&\s{B}\qw&\\
				\pureghost{(ij)_{s}}&\qw&\s{E}\qw&\qw&
			}
			\end{aligned}=
			\sum_{(m,\tau)\in I}\lambda_{m,\tau}^{(i)}
			\begin{aligned}
			\Qcircuit @C=1em @R=1.3em
			{
				\multiprepareC{1}{(mj)_{\tau s}}&\s{B}\qw&\\
				\pureghost{(mj)_{\tau s}}&\s{E}\qw&
			}
			\end{aligned}.
			\end{align}
			\item\label{item:33} $\T{D}\in\TransfN{A}{B}$;
			\item\label{item:11} There exists $\rket{\Sigma}_{\sys{B'}}\in\StN{B'}$ such that the following holds:
			\begin{align}\label{eq:channel}
			\begin{aligned}
			\Qcircuit @C=1em @R=1.3em
			{
				&\s{A}&\gate{\T{D}}&\s{B}\qw&
			}
			\end{aligned}
			=
			\begin{aligned}
			\Qcircuit @C=1em @R=1.3em
			{
				&\prepareC{\Sigma}&\s{B'}\qw&\multigate{1}{\tilde{\T{R}}_{\sys{A,B}}}&\s{A'}\qw&\measureD{e}&
				\\
				&\s{A}\qw&\qw&\ghost{\tilde{\T{R}}_{\sys{A,B}}}&\qw&\s{B}\qw&\qw&
			}
			\end{aligned};
			\end{align}
		\end{enumerate}
		\end{lemma}
	\begin{proof}
	\textbf{$\ref{item:11}\Rightarrow\ref{item:33}\Rightarrow\ref{item:22}$.} The chain of implications holds by definition and by the characterisation of deterministic transformations given in Subsec.~\ref{subsec:probabilistic}.
	
	\textbf{$\ref{item:22}\Rightarrow\ref{item:11}$.} Suppose that Eq.~\eqref{eq:channel2} holds. We provide below an explicit construction of a family of states $\rket{\Sigma}_{\sys{B'}}\in\StN{B'}$ such that Eq.~\eqref{eq:channel} holds. First, let $\rket{\Sigma}_{\sys{B'}}\in\StN{B'}$ be of the form $\sum_{(h,\xi)\in J}\mu_{h,\xi}\rket{\sigma_{h,\xi}}_{\sys{B''}}\rket{0}_{\sys{B}}$. We provide a step-by-step construction for suitable families of sets $J,\{\mu_{h,\xi}\}_{(h,\xi)\in J}$ such that Eq.~\eqref{eq:channel} reads as Eq.~\eqref{eq:channel2}. Start by taking the minimum nonvanishing value $\lambda_0\coloneqq\lambda_{m_0,\tau_0}^{(i_0)}$ over all the probability distributions $\{\lambda_{m,\tau}^{(i)}\}_{(m,\tau)\in I}$ for $i\in\{1,2,\ldots,\D{A}\}$. In the case where the minimum is not unique, just arbitrarily pick one of them. Define $h_0,\xi_0$ as those (families of) functions such that $h_0(i_0)=m_0,\xi_0(i_0)=\tau_0$ and  $h_0(i)=m_0^{(i)},\xi_0(i)=\tau_0^{(i)}$, where $m_0^{(i)},\tau_0^{(i)}$ are any chosen values such that $\lambda_{m_0^{(i)},\tau_0^{(i)}}^{(i)}$ is nonvanishing for every $i$. In the collections $\{\lambda_{m,\tau}^{(i)}\}_{(m,\tau)\in I}$ for $i\in\{1,2,\ldots,\D{A}\}$, reset the values of those coefficients $\lambda_{\tilde{m},\tilde{\tau}}^{(i)}$, where $(\tilde{m},\tilde{\tau})$ are in the image of $(h_0,\xi_0)$, in the following way: $\lambda_{\tilde{m},\tilde{\tau}}^{(i)}\mapsto \lambda_{\tilde{m},\tilde{\tau}}^{(i)}-\lambda_0$. By construction, this operation does not produce any negative value. Finally, set $\mu_{h_0,\xi_0}\coloneqq\lambda_0$.
Iterate the previous zeroth step. One realises that the iteration of the above procedure has a finite number of steps---say $N+1$---and eventually produces some families of sets $\{(h_n,\xi_n)\}_{n=0}^{N}$ and $\{\mu_{h_n,\xi_n}\}_{n=0}^{N}$. Choose now one arbitrary set $\{(h_n^*,\xi_n^*)\}_{n=0}^{N}$ and pose $J=\{(h_n^*,\xi_n^*)\}_{n=0}^{N}$. Using Eq.~\eqref{eq:arbitrary_transformation} one verifies that, by construction, Eq.~\eqref{eq:channel} reads as Eq.~\eqref{eq:channel2}, namely, $\ref{item:22}\Rightarrow\ref{item:11}$ holds.
\end{proof}

\begin{corollary}[Deterministic transformations]\label{cor:channels}
	Let $\T{D}\in\Transf{A}{B}$.
	Then $\T{D}\in\TransfN{A}{B}$ if and only if it admits of a reversible dilation for $\T{D}$ of the form~\eqref{eq:dilation} with $\rket\Sigma_\sys{B'}\in\StN{B'}$ and $\rbra{H}_{\sys{A'}}=\rbra{e}_{\sys{A'}}$.
\end{corollary}

\begin{lemma}[Realisation of arbitrary transformations]\label{lem:realisation_transf}
	Let $\sys{I}\neq\sys{A}\in\Sys{BCT}$, $\sys{B}\in\Sys{BCT}$, and $\T{T}\in\TransfR{A}{B}$
	Then, the following are equivalent:
	\begin{enumerate}[label=(\alph*)]
		\item\label{item:222} For every $i\in\{1,2,\ldots,\D{A}\}$ there exists a set $\{\gamma_{m,\tau}^{(i)}\}_{(m,\tau)\in I^{(i)}}$---with $I^{(i)}\subseteq\{1,2,\ldots,\D{B}\}\times\{+,-\}$, $\gamma_{m,\tau}^{(i)}>0$ for all $(m,\tau)\in I^{(i)}$, and $\sum_{(m',\tau')\in I^{(i)}}\gamma_{m',\tau'}^{(i)}\leq 1$---such the following holds that for all $\rket{(ij)_s}_{\sys{AE}}\in\PurSt{AE}$:
		\begin{align}\label{eq:arbitr_transf2}
		\begin{aligned}
		\Qcircuit @C=1em @R=1.3em
		{
			\multiprepareC{1}{(ij)_{s}}&\s{A}\qw&\gate{\T{T}}&\s{B}\qw&\\
			\pureghost{(ij)_{s}}&\qw&\s{E}\qw&\qw&
		}
		\end{aligned}=
		\sum_{(m,\tau)\in I^{(i)}}\gamma_{m,\tau}^{(i)}
		\begin{aligned}
		\Qcircuit @C=1em @R=1.3em
		{
			\multiprepareC{1}{(mj)_{\tau s}}&\s{B}\qw&\\
			\pureghost{(mj)_{\tau s}}&\s{E}\qw&
		}
		\end{aligned}.
		\end{align}
		\item\label{item:333} $\T{T}\in\Transf{A}{B}$;
		\item\label{item:111}  There exist a deterministic state $\rket{\Sigma}_{\sys{B'}}\in\StN{B'}$ such that, for every $\rket{(ij)_s}_{\sys{AE}}\in\PurSt{AE}$, $\T{T}\boxtimes\T{I}_{\sys{E}}\rket{(ij)_s}_{\sys{AE}}$ is in the refinement set of some $\T{D}\boxtimes\T{I}_{\sys{E}}\rket{(ij)_s}_{\sys{AE}}$---for $\T{D}\in\TransfN{A}{B}$ dilated as in Eq.~\eqref{eq:channel}---and an effect $\rbra{a}_{\sys{A'}}\in\Eff{A'}$, such that the following holds:
		\begin{align}\label{eq:arbitr_transf}
		\begin{aligned}
		\Qcircuit @C=1em @R=1.3em
		{
			&\s{A}&\gate{\T{T}}&\s{B}\qw&
		}
		\end{aligned}
		=
		\begin{aligned}
		\Qcircuit @C=1em @R=1.3em
		{
			&\prepareC{\Sigma}&\s{B'}\qw&\multigate{1}{\tilde{\T{R}}_{\sys{A,B}}}&\s{A'}\qw&\measureD{a}&
			\\
			&\s{A}\qw&\qw&\ghost{\tilde{\T{R}}_{\sys{A,B}}}&\qw&\s{B}\qw&\qw&
		}
		\end{aligned};
		\end{align}
	\end{enumerate}
\end{lemma}
\begin{proof}
	\textbf{$\ref{item:111}\Rightarrow\ref{item:333}\Rightarrow\ref{item:222}$.} The chain of implications holds by Postulate~\ref{postulate:dilation} and Lemma~\ref{lem:char_transf}.
	
	\textbf{$\ref{item:222}\Rightarrow\ref{item:111}$.} For every map $\T{T}\in\TransfR{A}{A}$ defined as in Eq.~\eqref{eq:arbitr_transf2}, by Lemma~\ref{lem:realisation_channels} there exists a deterministic transformation $\T{D}\in\TransfN{A}{B}$ and a deterministic state $\rket{\Sigma}_{\sys{B'}}$ such that $\T{T}\boxtimes\T{I}_{\sys{E}}\rket{(ij)_s}_{\sys{AE}}$ is in the refinement set of $\T{D}\boxtimes\T{I}_{\sys{E}}\rket{(ij)_s}_{\sys{AE}}$ for every $\rket{(ij)_s}_{\sys{AE}}\in\PurSt{AE}$. It can be verified by direct computation that such a $\rket{\Sigma}_{\sys{B'}}$ is of the form $\sum_{(h,\xi)\in J}\mu_{h,\xi}\rket{\sigma_{h,\xi}}_{\sys{B''}}\rket{0}_{\sys{B}}$---and this has been indeed shown in the proof of Lemma~\ref{lem:realisation_channels}. Now consider:
	\begin{align}\label{eq:expr}
	\sum_{(h,\xi)\in J}\mu_{h,\xi}\begin{aligned}
	\Qcircuit @C=1em @R=1.3em
	{
		&\prepareC{\sigma_{h,\xi}}&\s{B''}\qw&\multigate{2}{\tilde{\T{R}}_{\sys{A,B}}}&\s{\sys{B''}}\qw& \\
		&\prepareC{0}&\s{B}\qw&\ghost{\tilde{\T{R}}_{\sys{A,B}}}&\s{\sys{A}}\qw&\\
		&\multiprepareC{1}{(ij)_s}&\s{A}\qw&\ghost{\tilde{\T{R}}_{\sys{A,B}}}&\s{B}\qw&\\
		&\pureghost{(ij)_s}&\qw&\s{E}\qw&\qw&
	}
	\end{aligned}.
	\end{align}
	By construction and by direct inspection of expression~\eqref{eq:expr} via Eq.~\eqref{eq:atomic_atomic}, for all $i$ there exists a suitably chosen set of coefficients $\{\beta_{h,\xi}^{(i)}\}_{(h,\xi)\in J^{(i)}}$---with $J^{(i)}\subseteq J$, $\beta_{h,\xi}^{(i)}>0$, and $\sum_{(h',\xi')\in J^{(i)}}\beta_{h',\xi'}^{(i)}\leq 1$---such that one attains for every $\tilde{i}$:
	\begin{align*}
	\begin{split}
	&\sum_{(\tilde{h},\tilde{\xi})\in J^{(i)}}\beta_{h,\xi}^{(i)}
	\begin{aligned}
	\Qcircuit @C=1em @R=1.3em
	{
		&\multiprepareC{1}{\Sigma}&\s{B''}\qw&\multigate{2}{\tilde{\T{R}}_{\sys{A,B}}}&\s{\sys{B''}}\qw&\measureD{\sigma_{\tilde{h},\tilde{\xi}}}& \\
		&\pureghost{\Sigma}&\s{B}\qw&\ghost{\tilde{\T{R}}_{\sys{A,B}}}&\s{\sys{A}}\qw&\measureD{\tilde{i}}& \\
		&\multiprepareC{1}{(ij)_s}&\s{A}\qw&\ghost{\tilde{\T{R}}_{\sys{A,B}}}&\qw&\s{B}\qw&\qw& \\
		&\pureghost{(ij)_s}&\qw&\qw&\s{E}\qw&\qw&\qw
	}
	\end{aligned}\!\!\!\!=\\[2.5ex]
	=&\delta_{i\tilde{i}}\sum_{(m,\tau)\in I^{(i)}}\gamma_{m,\tau}^{(i)}
	\begin{aligned}
	\Qcircuit @C=1em @R=1.3em
	{
		\multiprepareC{1}{(mj)_{\tau s}}&\s{B}\qw&\\
		\pureghost{(mj)_{\tau s}}&\s{E}\qw&
	}
	\end{aligned}.
	\end{split}
	\end{align*}
	Define now
	$\rbra{a}_{\sys{A'}}\coloneqq\sum_{i=1}^{\D{A}}\sum_{h,\xi\in J^{(i)}}\beta_{h,\xi}^{(i)}\rbra{\sigma_{h,\xi}}_{\sys{B''}}\rbra{i}_{\sys{A}}\in\EffR{A'}$. Then one obtains, for all $\rket{(ij)_s}_{\sys{AE}}\in\PurSt{AE}$:
	\begin{align*}
	\begin{aligned}
	\Qcircuit @C=1em @R=1.3em
	{
		\prepareC{\Sigma}&\s{B'}\qw&\multigate{1}{\tilde{\T{R}}_{\sys{A,B}}}&\s{\sys{A'}}\qw&\measureD{a}& \\
		\multiprepareC{1}{(ij)_s}&\s{A}\qw&\ghost{\tilde{\T{R}}_{\sys{A,B}}}&\qw&\s{B}\qw&\qw& \\
		\pureghost{(ij)_s}&\qw&\qw&\s{E}\qw&\qw&\qw
	}
	\end{aligned}
	\!\!\!=\!\!\begin{aligned}
	\Qcircuit @C=1em @R=1.3em
	{
		&\multiprepareC{1}{(ij)_{s}}&\s{A}\qw&\gate{\T{T}}&\s{B}\qw&\\
		&\pureghost{(ij)_{s}}&\qw&\s{E}\qw&\qw&
	}
	\end{aligned}\!\!\!.
	\end{align*}
	Finally, by construction one observes that $\rbra{a}_{\sys{A'}}\in\Eff{A'}$ (see classification of BCT's effects in Subsec.~\ref{subsec:characterisation})---i.e.~Eq.~\eqref{eq:arbitr_transf} holds---and this concludes the proof.
\end{proof}
We stress that the statement of condition~\ref{item:111} requires a unique deterministic state $\rket\Sigma_\sys{B'}$ for every transformation $\T T$ such that $\T{T}\boxtimes\T{I}_{\sys{E}}\rket{(ij)_s}_{\sys{AE}}$ 
refine $\T{D}\boxtimes\T{I}_{\sys{E}}\rket{(ij)_s}_{\sys{AE}}$ for a fixed deterministic $\tD\in\TransfN{A}{B}$. This is true, in particular, if $\T T$ refines $\T D$. The latter property will be crucially exploited in the proof of Proposition~\ref{prop:dilation_tests} (see Appendix~\ref{app:BCT_no-restriction}).
\begin{corollary}\label{cor:class_transf}
	Let $\sys{I}\neq\sys{A}\in\Sys{BCT}$ and $\sys{B}\in\Sys{BCT}$. Then every conical combination of elements of the form~\eqref{eq:atomic} that maps states to states is a transformation of BCT.
\end{corollary}
\begin{proof}
	This is a straightforward consequence of the fact that every conical combination of generalised transformations of the form~\eqref{eq:atomic} can be expressed as in Eq.~\eqref{eq:arbitr_transf2}, combined with implication $\ref{item:222}\Rightarrow\ref{item:333}$ in Lemma~\ref{lem:realisation_transf}.
\end{proof}
Lemma~\ref{lem:char_transf} and Corollary~\ref{cor:class_transf} provide a classification of BCT's transformations.
\begin{lemma}[Atomic transformations]\label{lem:atomic}
	Let $\sys{I}\neq\sys{A}\in\Sys{BCT}$ and $\sys{B}\in\Sys{BCT}$. A map $\T{A}\in\TransfR{A}{B}$ is an atomic transformation if and only if $\T{A}\boxtimes\T{I}_{\sys{E}}$ is of the form~\eqref{eq:atomic} for every $\sys{E}\in\Sys{BCT}$.
\end{lemma}
\begin{proof}
	By Lemma~\ref{lem:realisation_transf}, every map of the form~\eqref{eq:atomic} is an admissible transformation of the theory. First, we show that every transformation of $\T{A}\in\Transf{A}{B}$ such that $\T{A}\boxtimes\T{I}_{\sys{E}}$, for every $\sys{E}\in\Sys{BCT}$, is of the form~\eqref{eq:atomic}, satisfies the definition of atomicity (see Eq.~\eqref{eq:atomicity}). Indeed, we have already proven that the action of an arbitrary transformation $\T{T}\in\Transf{A}{B}$ is given by Eq.~\eqref{eq:arbitrary_transformation}. Let us then pose:
	\begin{align*}
	&\begin{aligned}
	\Qcircuit @C=1em @R=1.3em
	{
		\multiprepareC{1}{(ij)_{s}}&\s{A}\qw&\gate{\T{A}}&\s{B}\qw&\\
		\pureghost{(ij)_{s}}&\qw&\s{E}\qw&\qw&
	}
	\end{aligned}
	=
	\lambda\delta_{ii_0}\begin{aligned}
	\Qcircuit @C=1em @R=1.3em
	{
		\multiprepareC{1}{(lj)_{\tau s}}&\s{B}\qw&
		\\
		\pureghost{(lj)_{\tau s}}&\s{E}\qw&
	}
	\end{aligned}
	=\\[2.5ex]=\ \ &
	\begin{aligned}
	\Qcircuit @C=1em @R=1.3em
	{
		\multiprepareC{1}{(ij)_{s}}&\s{A}\qw&\gate{\T{A}_1+\T{A}_2}&\s{B}\qw&\\
		\pureghost{(ij)_{s}}&\qw&\s{E}\qw&\qw&
	}
	\end{aligned}
	=\\[2.5ex]=\ \ &
	\frac{1}{2}\sum_{n=1,2}\sum_{k,s'}\alpha_{k}^{n}\beta_{\kappa_{n,s'}^{ki}}^{n}
	\begin{aligned}
	\Qcircuit @C=1em @R=1.3em
	{
		\multiprepareC{1}{(\iota_{n,s'}^{ki}j)_{r_{n,s'}^{ki}t_{n,s'}^{ki} ss'}}&\s{B}\qw&\\
		\pureghost{(\iota_{n,s'}^{ki}j)_{r_{n,s'}^{ki}t_{n,s'}^{ki} ss'}}&\s{E}\qw&
	}
	\end{aligned}.
	\end{align*}
	Now, since all the coefficients $\lambda,\alpha_{k}^{n},\beta_{\kappa_{n,s'}^{ki}}^{n}$ are non-negative, it must be for all $\rket{(ij)_s}_\sys{AE}\in\PurSt{AE}$:
	\begin{align*}
	(\T{A}_n\boxtimes\T{I}_{\sys{E}})\rket{(ij)_s}_{\sys{AE}}=\begin{cases}
	0, & i\neq i_0, \\
	\lambda_n\rket{(lj)_{\tau s}}_{BE}, & i=i_0,
	\end{cases}	
	\end{align*}
	for $n=1,2$ and non-negative coefficients $\lambda_n$ such that $\lambda_1+\lambda_2=\lambda$. It follows that $\T{A}_1\propto\T{A}_2$, namely those transformations $\T{A}\in\Transf{A}{B}$, such that $\T{A}\boxtimes\T{I}_{\sys{E}}$ for every $\sys{E}\in\Sys{BCT}$ is of the form~\eqref{eq:atomic}, are atomic. Conversely, let $\T{A}\boxtimes\T{I}_{\sys{E}}$ be a transformation from $\sys{AE}$ to $\sys{BE}$. By Lemma~\ref{lem:char_transf}, we can write, without loss of generality, $\T{A}=\sum_{n\in N}\T{A}_{n}$, where $\T{A}_n\boxtimes\T{I}_{\sys{E}}$ is of the form~\eqref{eq:atomic} for all $\sys{E}\in\Sys{BCT},n\in N$ and $\T{A}_{n_1}\not\propto\T{A}_{n_2}$ for all $n_1\neq n_2$. Accordingly, by Lemma~\ref{lem:realisation_transf}, for any $\tilde{n}\in N$ both $\T{A}_{\tilde{n}}\boxtimes\T{I}_{\sys{E}}$ and $\left(\T{A}-\T{A}_{\tilde{n}}\right)\boxtimes\T{I}_{\sys{E}}$ are transformations of BCT. Now, suppose that $\T{A}$ is atomic, namely, $\T{A}_{\tilde{n}}\propto\left(\T{A}-\T{A}_{\tilde{n}}\right)$. This implies that $\T{A}\propto\T{A}_{\tilde{n}}$, i.e., $\T{A}\boxtimes\T{I}_{\sys{E}}$ is of the form~\eqref{eq:atomic}.
\end{proof}
Proposition~\ref{prop:atomic_conclusive} follows combining Lemmas~\ref{lem:char_transf} and~\ref{lem:atomic}, and Corollaries~\ref{cor:channels} and~\ref{cor:class_transf}.

\section{Operational realisation scheme for arbitrary instruments in BCT [Proof of Proposition~\ref{prop:dilation_tests}]}\label{app:BCT_no-restriction}
In the present appendix, in order to make sense of expressions of the form $\rket{(l_1l_2)_{u}}_{\sys{S}_1\sys{S}_2}$ with $\sys{S}_1=\sys{I}$ or $\sys{S}_2=\sys{I}$, these can be safely replaced by the expressions $\rket{(*l_2)_{+}}_{\sys{S}_2}$ or $\rket{(l_1*)_{+}}_{\sys{S}_1}$, respectively.
\begin{lemma}\label{lem:pseudo_no-restriction}
	Let $\sys{A},\sys{B}\in\Sys{BCT}$ and $\T{T}\in\TransfR{A}{B}$ such that $\T{T}\boxtimes\T{I}_{\sys{E}}$ maps $\St{AE}$ to $\St{BE}$ for all $\sys{E}\in\Sys{BCT}$. Then $\T{T}\in\Transf{A}{B}$.
\end{lemma}
\begin{proof}
	Let $\T{T}\in\TransfR{A}{C}$ be a generalised transformation such that $\T{T}\boxtimes\T{I}_{\sys{E}}$ maps $\St{AE}$ to $\St{BE}$ for all $\sys{E}\in\Sys{BCT}$. By simpliciality, $\T{T}$ may take the following general form:
	\begin{align}\label{eq:transf_generic}
	\T{T}\boxtimes\T{I}_{\sys{B}}\rket{(ij)_s}_{\sys{A}\sys{B}} = \sum_{m,l,\tau}\lambda_{m,l,\tau}^{i,j,s}\rket{(ml)_{\tau s}}_{\sys{C}\sys{B}},
	\end{align}
	with $\lambda_{m,l,\tau}^{i,j,s}>0$ for all $i,j,s,m,l,\tau$ and $\sum_{m,l,\tau}\lambda_{m,l,\tau}^{i,j,s}\leq 1$. 
	Let now $\rbra{j'}_{\sys{B}}\in\Eff{B}$ be such that $\rbraket{j'}{j}_\sys{B}=\delta_{j'j}$ for all $\rket{j}_\sys{B}\in\PurSt{B}$.
    Imposing Eq.~\eqref{eq:local_commute}, namely,
	\begin{align*}
	\left[\T{I}_\sys{C}\boxtimes\rbra{j'}_\sys{B}\right]\left(\T{T}\boxtimes\T{I}_{\sys{B}}\right)=\left(\T{T}\boxtimes\T{I}_{\sys{B}}\right)\left[\T{I}_\sys{A}\boxtimes\rbra{j'}_\sys{B}\right],
	\end{align*}
	and using Eqs.~\eqref{eq:steering}, one has for all $i,j,j',s$:
	\begin{align*}
	\begin{split}
	&\left[\T{I}_\sys{C}\boxtimes\rbra{j'}_\sys{B}\right]\left(\T{T}\boxtimes\T{I}_{\sys{B}}\right)\rket{(ij)_s}_{\sys{A}\sys{B}} =\\
	&\left[\T{I}_\sys{C}\boxtimes\rbra{j'}_\sys{B}\right]\sum_{m,l,\tau}\lambda_{m,l,\tau}^{i,j,s}\rket{(ml)_{\tau s}}_{\sys{C}\sys{B}}
	=\sum_{m,\tau}\lambda_{m,j',\tau}^{i,j,s}\rket{m}_{\sys{C}}=\\
	&\left(\T{T}\boxtimes\T{I}_{\sys{B}}\right)\left[\T{I}_\sys{A}\boxtimes\rbra{j'}_\sys{B}\right]\rket{(ij)_s}_{\sys{A}\sys{B}}=\delta_{j'j}\T{T}\rket{i}_\sys{A}.
	\end{split}
	\end{align*}
	The above equation implies that Eq.~\eqref{eq:transf_generic} takes the form:
	\begin{align}\label{eq:transf_generic2}
	\T{T}\boxtimes\T{I}_{\sys{B}}\rket{(ij)_s}_{\sys{A}\sys{B}} = \sum_{m,\tau}\lambda_{m,\tau}^{i,j,s}\rket{(mj)_{\tau s}}_{\sys{C}\sys{B}}.
	\end{align}
	Accordingly, imposing now
	\begin{align*}
	\left(\T{I}_\sys{C}\boxtimes\T{R'}\right)\left(\T{T}\boxtimes\T{I}_{\sys{B}}\right)=\left(\T{T}\boxtimes\T{I}_{\sys{B}}\right)\left(\T{I}_\sys{A}\boxtimes\T{R'}\right)
	\end{align*}
	for all $\T{R'}\in\RevTransf{B}{B}$ (see Postulate~\ref{postulate:reversible}), one has for all $i,j,s,\pi',\sigma'_m$:
	\begin{align*}
	\begin{split}
	&\left(\T{I}_\sys{C}\boxtimes\T{R}'\right)\left(\T{T}\boxtimes\T{I}_{\sys{B}}\right)\rket{(ij)_s}_{\sys{A}\sys{B}} =\\
	&=\left(\T{I}_\sys{C}\boxtimes\T{R}'\right)\sum_{m,\tau}\lambda_{m,\tau}^{i,j,s}\rket{(mj)_{\tau s}}_{\sys{C}\sys{B}}=\\
	&=\sum_{m,\tau}\lambda_{m,\tau}^{i,j,s}\rket{(m\pi'(j))_{\sigma'_m \tau s}}_{\sys{C}\sys{B}}=\\
	&=\left(\T{T}\boxtimes\T{I}_{\sys{B}}\right)\left(\T{I}_\sys{C}\boxtimes\T{R}'\right)\rket{(ij)_s}_{\sys{A}\sys{B}}=\\
	&=\sum_{m,\tau}\lambda_{m,\tau}^{i,\pi'(j),\sigma'_ms}\rket{(m\pi'(j))_{\tau \sigma'_m s}}_{\sys{C}\sys{B}}.
	\end{split}
	\end{align*}
	The above equation implies that the coefficients $\lambda_{m,\tau}^{i,j,s}$ cannot depend on $j,s$ for all $i,m,\tau$, namely Eq.~\eqref{eq:transf_generic2} takes the following form:
	\begin{align}\label{eq:transf_generic3}
	\T{T}\boxtimes\T{I}_{\sys{B}}\rket{(ij)_s}_{\sys{A}\sys{B}} = \sum_{m,\tau}\lambda_{m,\tau}^{i}\rket{(mj)_{\tau s}}_{\sys{C}\sys{B}}.
	\end{align}
	Using Proposition~\ref{prop:atomic_conclusive}, one realises from Eq.~\eqref{eq:transf_generic3} that $\T{T}\in\TransfA{BCT}$, and this concludes the proof.
\end{proof}
We now use Lemmas~\ref{lem:realisation_channels},~\ref{lem:realisation_transf}, and~\ref{lem:pseudo_no-restriction} in order to prove Proposition~\ref{prop:dilation_tests}.
\paragraph{\textbf{$\ref{item:3}\Rightarrow\ref{item:2}\Rightarrow\ref{item:1}$.}} Recalling Postulate~\ref{postulate:instruments}, the chain of implications is trivial.
\paragraph{\textbf{$\ref{item:1}\Rightarrow\ref{item:3}$.}} Let $\InstrC{T}{X}{A}{B}\subset\TransfR{A}{B}$ such that $\InstrC{T}{X}{A}{B}\boxtimes\InstrC{I}{\star}{E}{E}$ maps preparation-instruments of $\sys{AE}$ to preparation-instruments of $\sys{BE}$ for all $\sys{E}\in\Sys{BCT}$.

Let us first prove the implication in the case where $\sys{A}=\sys{I}$. Choosing $\sys{E}=\sys{I}$, by hypothesis we have that $\InstrC{T}{X}{A}{B}$ maps the unity $1\in\PurSt{I}$ to a preparation instrument of $\sys{B}$. Namely, $\InstrC{T}{X}{A}{B}$ is a preparation instrument of BCT. Let us denote $\srho_{\mathsf{X}}^{\sys{I}\!\rightarrow\!\sys{B}}\coloneqq\TestC{T}{X}{A}{B}$ and $\rho=\sum_{i\in I}\lambda_{i}\rket{i}_{\sys{B'}}\coloneqq\sum_{x\in\Out{X}}\rho_x\in\StN{B}$ with $\lambda_{i'}> 0$ for all $i'\in I$ and $\sum_{i\in I}\lambda_{i}=1$.  Then, the set $\{\rho_x\}_{x\in\Out{X}}$ is a refinement of $\rho$. Take $\sys{B'}=\sys{BB}$, $\sys{A'}=\sys{B}$, $\T{R}=\T{I}_{\sys{BB}}\in\RevTransf{B'A}{A'B}$, and $\rket{\Sigma}_{\sys{B'}}=\sum_{i\in I}\lambda_{i}\rket{i}_{\sys{B}}\rket{i}_{\sys{B}}$, so that $\rho=[\rbra{e}_{\sys{A'}}\boxtimes\T{I}_{\sys{B}}]\T{R}\rket{\Sigma}_{\sys{B'}}$. By simpliciality, every $\rho_x$ can be rewritten as $\sum_{i\in I}\gamma_{i}^{(x)}\rket{i}_{\sys{B}}$, for suitably defined $\gamma_{i}^{(x)}$ such that $0\leq\gamma_{i}^{(x)}\leq\lambda_{i}$ and $\sum_{x\in\Out{X}}\gamma_{i}^{(x)}=\lambda_{i}$ for every $i\in I$. Define now, for every $x\in\Out{X}$ and $i\in I$, $\beta_{i}^{(x)}\coloneqq \gamma_{i}^{(x)}/\lambda_{i}$. It is clear that, by definition, it must be $0\leq\beta_{i}^{(x)}\leq 1$ and $\sum_{x'\in\Out{X}}\beta_{i}^{(x')}=1$ for every $x\in\Out{X}$ and $i\in I$. Now, defining $\rbra{a_x}_{\sys{A'}}\coloneqq\sum_{i}\beta_{i}^{(x)}\rbra{i}_{\sys{A'}}\in\EffR{A'}$ for every $x\in\Out{X}$, one has $\rho_x=[\rbra{a_x}_{\sys{A'}}\boxtimes\T{I}_{\sys{B}}]\T{R}\rket{\Sigma}_{\sys{B'}}$. Finally, since by the classification of BCT's effects (see Subsec.~\ref{subsec:characterisation}) one has $\rbra{a_{x'}}_{\sys{A'}}\in\Eff{A'}$ for all $x'\in\Out{X}$, and being $\sum_{x\in\Out{X}}\rbra{a_x}_{\sys{A'}}=\rbra{e}_{\sys{A'}}$, by Postulate~\ref{postulate:instruments} one can conclude $\{\rbra{a_x}_{\sys{A'}}\}_{x\in\Out{X}}\in\InstrA{BCT}$.

Let now be $\sys{A}\neq\sys{I}$ and $\tilde{\T{R}}_{\sys{A,B}}\in\RevTransf{B'A}{A'B}$ defined as in Eq.~\eqref{eq:atomic_atomic}. By hypothesis, $\T{T}_x\boxtimes\T{I}_{\sys{E}}$ maps $\St{AE}$ to $\St{BE}$ for every $x\in\Out{X}$ and $\sys{E}\in\Sys{BCT}$. Accordingly, by Lemma~\ref{lem:pseudo_no-restriction}, $\T{T}_x\in\Transf{A}{B}$ for every $x\in\Out{X}$. On the one hand, condition~\ref{item:22} of Lemma~\ref{lem:realisation_channels} holds for $\T{D}=\sum_{x\in\Out{X}}\T{T}_x$. On the other hand, condition~\ref{item:333} of Lemma~\ref{lem:realisation_transf} holds for $\T{T}=\T{T}_x$, and for every $x\in\Out{X}$. Thus, by construction, and invoking implications  $\ref{item:22}\Rightarrow\ref{item:33}$ in Lemma~\ref{lem:realisation_channels} and $\ref{item:333}\Rightarrow\ref{item:111}$ in Lemma~\ref{lem:realisation_transf}, there exists a state $\rket{\Sigma}_{\sys{B'}}\in\StN{B'}$
such that the following holds:
\begin{align*}
\begin{aligned}
\Qcircuit @C=1em @R=1.3em
{
	&\s{A}&\gate{\sum_{x\in\Out{X}}\T{T}_x}&\s{B}\qw&
}
\end{aligned}
=
\begin{aligned}
\Qcircuit @C=1em @R=1.3em
{
	&\prepareC{\Sigma}&\s{B'}\qw&\multigate{1}{\tilde{\T{R}}_{\sys{A,B}}}&\s{A'}\qw&\measureD{e}&
	\\
	&\s{A}\qw&\qw&\ghost{\tilde{\T{R}}_{\sys{A,B}}}&\qw&\s{B}\qw&\qw&
}
\end{aligned}.
\end{align*}
We have thus shown that the collection $\InstrC{T}{X}{A}{B}\subset\Transf{A}{B}$ is a refinement for the deterministic transformation $\T{D}\coloneqq\sum_{x\in\Out{X}}\T{T}_x$. 
Considering Eq.~\eqref{eq:expr}, we know that $\rket{\Sigma}_{\sys{B'}}$ can be taken of the form $\sum_{(h,\xi)\in J}\mu_{h,\xi}\rket{\sigma_{h,\xi}}_{\sys{B''}}\rket{0}_{\sys{B}}$ with positive 
coefficients $\mu_{h,\xi}$.
Moreover, $\T{D}$ can be decomposed into conical combinations of atomic 
maps 
\begin{align}\label{eq:atomic_programmed}
\begin{split}
&\delta_{(h,\xi)(\tilde{h},\tilde{\xi})}
\begin{aligned}
\Qcircuit @C=1em @R=1.3em
{
	\multiprepareC{1}{(ij)_{s}}&\s{A}\qw&\gate{\T{A}_{\tilde{i}}^{(h,\xi)}}&\s{B}\qw&\\
	\pureghost{(ij)_{s}}&\qw&\s{E}\qw&\qw&
}
\end{aligned}
\coloneqq\\[2.5ex]\coloneqq
&
\begin{aligned}
\Qcircuit @C=1em @R=1.3em
{
	&\prepareC{\sigma_{h,\xi}}&\s{B''}\qw&\multigate{2}{\tilde{\T{R}}_{\sys{A,B}}}&\s{\sys{B''}}\qw&\measureD{\sigma_{\tilde{h},\tilde{\xi}}}& \\
	&\prepareC{0}&\s{B}\qw&\ghost{\tilde{\T{R}}_{\sys{A,B}}}&\s{\sys{A}}\qw&\measureD{\tilde{i}}& \\
	&\multiprepareC{1}{(ij)_s}&\s{A}\qw&\ghost{\tilde{\T{R}}_{\sys{A,B}}}&\qw&\s{B}\qw&\qw& \\
	&\pureghost{(ij)_s}&\qw&\qw&\s{E}\qw&\qw&\qw
}.
\end{aligned}
=\\[2.5ex]=
&
\delta_{(h,\xi)(\tilde{h},\tilde{\xi})}\delta_{i\tilde{i}}\begin{aligned}
\Qcircuit @C=1em @R=1.3em
{
	\multiprepareC{1}{(h(i)j)_{\xi(i) s}}&\s{B}\qw&
	\\
	\pureghost{(h(i)j)_{\xi(i) s}}&\s{E}\qw&
}
\end{aligned}.
\end{split}
\end{align}
In other terms, the following decomposition holds:
\begin{align}
\T{D} = \sum_{\tilde{i}=1}^{\D{A}}\sum_{(h,\xi)\in J}\mu_{h,\xi}\T{A}_{\tilde{i}}^{(h,\xi)}.
\label{eq:decomp_D}
\end{align}
Now, in the light of Lemma~\eqref{lem:atomic}, the decomposition of a transformation into atomic transformations is unique up to trivial refinements---namely, refinements where the elements are proportional to each other.
This implies that every refinement of $\tD$ must consist in a trivial refinement and subsequent coarse-graining of the decomposition in Eq.~\eqref{eq:decomp_D}. Consequently, for all $x'\in\Out{X},(h',\xi')\in J,\tilde{i}'\in\{1,2,\ldots,\D{A}\}$, it must be:
\begin{align*}
\T{T}_{x'} = \sum_{\tilde{i}=1}^{\D{A}}\sum_{(h,\xi)\in J}\nu_{(h,\xi),\tilde{i}}^{(x')}\T{A}_{\tilde{i}}^{(h,\xi)},\quad \nu_{(h',\xi'),\tilde{i}'}^{(x')}\in[0,1], \\ 
\sum_{x\in\Out{X}} \nu_{(h',\xi'),\tilde{i}'}^{(x)} =
\mu_{h',\xi'}, \quad
\sum_{(h,\xi)\in J} \nu_{(h,\xi),\tilde{i}'}^{(x')} 
\leq 1.
\end{align*}
Now, let us define the following collection of coefficients:
\begin{align*}
    \zeta_{(h,\xi),\tilde{i}}^{(x)}\coloneqq \frac{\nu_{(h,\xi),\tilde{i}}^{(x)}}{\mu_{h,\xi}} \in[0,1],\\ \forall x\in\Out{X},\ (h,\xi)\in J,\ \tilde{i}\in\{1,2,\ldots,\D{A}\}.
\end{align*}
Accordingly, each $\T{T}_x$ can be achieved as follows:
\begin{align*}
&\rbra{a_x}_{\sys{A'}} \coloneqq \sum_{\tilde{i}=1}^{\D{A}}\sum_{(\tilde{h},\tilde{\xi})\in J} \zeta_{(\tilde{h},\tilde{\xi}),\tilde{i}}^{(x)} \rbra{\sigma_{\tilde{h},\tilde{\xi}}}_{\sys{B''}}\rbra{\tilde{i}}_{\sys{A}}, \\[2.5ex]
&
\begin{aligned}
\Qcircuit @C=1em @R=1.3em
{
	&\prepareC{\Sigma}&\s{B'}\qw&\multigate{1}{\tilde{\T{R}}_{\sys{A,B}}}&\s{A'}\qw&\measureD{a_x}&
	\\
	&\s{A}\qw&\qw&\ghost{\tilde{\T{R}}_{\sys{A,B}}}&\qw&\s{B}\qw&\qw&
}
\end{aligned} =\\[2.5ex]=&
\sum_{\tilde{i}=1}^{\D{A}}\sum_{(h,\xi)\in J}\mu_{h,\xi}\zeta_{(h,\xi),\tilde{i}}^{(x)}\T{A}_{\tilde{i}}^{(h,\xi)} \equiv
\begin{aligned}
\Qcircuit @C=1em @R=1.3em
{
	&\s{A}&\gate{\T{T}_x}&\s{B}\qw&
}
\end{aligned}.
\end{align*}
By construction, and by the classification of BCT's effects given in Subsec.~\ref{subsec:characterisation}, $\{\rbra{a_x}_{\sys{A'}}\}_{x\in\Out{X}}\subset\Eff{A'}$ and $\sum_{x\in\Out{X}}\rbra{a_x}_{\sys{A'}}=\sum_{(\tilde{h},\tilde{\xi})\in J} \rbra{\sigma_{\tilde{h},\tilde{\xi}}}_{\sys{B''}}\rbra{e}_{\sys{A}}$ hold. In the case where $\sum_{(\tilde{h},\tilde{\xi})\in J} \rbra{\sigma_{\tilde{h},\tilde{\xi}}}_{\sys{B''}}\neq\rbra{e}_{\sys{B''}}$, one could complete the collection $\{\rbra{a_x}_{\sys{A'}}\}_{x\in\Out{X}}$ adding the effect $\sum_{(\tilde{h},\tilde{\xi})\in \tilde{J}} \rbra{\sigma_{\tilde{h},\tilde{\xi}}}_{\sys{B''}}\rbra{e}_{\sys{A}}$, where $\tilde{J}$ collects all the pairs $(h,\xi)\not\in J$, to any of the effects in the collection, say e.g.~$\rbra{a_{x_0}}_{\sys A'}$. This simply amounts to adding 
the associated null transformation to the corresponding transformation $\T{T}_{x_0}$, since $\rbraket{\sigma_{\tilde{h},\tilde{\xi}}}{\sigma_{h',\xi'}}_{\sys{B''}}=0$ for every $(\tilde{h},\tilde{\xi})\in \tilde{J}$ and $(h',\xi')\in J$. Then, by the first part of Postulate~\ref{postulate:instruments}, the collection of effects $\{\rbra{a_x}_{\sys{A'}}\boxtimes\T{I}_{\sys{C}}\}_{x\in\Out{X}}$ maps preparation-instruments of $\sys{AC}$ to preparation-instruments of $\sys{C}$ for all $\sys{C}\in\Sys{BCT}$.  Finally, by the second part of Postulate~\ref{postulate:instruments}, one can conclude that $\{\rbra{a_x}_{\sys{A'}}\}_{x\in\Out{X}}\in\InstrA{BCT}$.

\section{Conditional instruments in theories with a unique deterministic effect}\label{app:conditional}
We characterise those causal theories (see Property~\ref{prope:causality}) which satisfy Postulate~\ref{postulate:instruments} and Property~\ref{prope:unrestricted_instruments}, proving that they contain every possible conditional instrument, namely, they satisfy Property~\ref{postulate:conditional}.
\begin{theorem}\label{thm:conditional}
	Let $\Theta$ be an OPT satisfying Property~\ref{prope:causality}, Postulate~\ref{postulate:instruments}, and Property~\ref{prope:unrestricted_instruments}. Then the theory $\Theta$ also satisfies Property~\ref{postulate:conditional}.
\end{theorem}
\begin{proof}
	We denote the unique deterministic effect of each $\sys{A}\in\Sys{\Theta}$ by $\rbra{e}_{\sys{A}}$. Suppose that $\srho_{\mathsf{X}}^{\sys{I}\!\rightarrow\!\sys{AE}}$ and $\InstrC{A}{Y}{A}{B}$ are instruments of $\Theta$, and let $\{\mathsf{B}_{\mathsf{Z}^{(y)}}^{(y)}\}_{y\in\Out{Y}}\subset\InstrAB{B}{C}$
	a collection of instruments labelled by $y\in\Out{Y}$. Consider now the following collection of transformations:
	\begin{align*}
	\mathsf{P}\coloneqq\bigcup_{y\in\mathsf{Y}}\lbrace(\T{B}_{z}^{(y)}\T{A}_y\boxtimes\T{I}_{\sys{E}})\rho_{x}\rbrace_{(x,z)\in\Out{X}\times\mathsf{Z}^{(y)}}\in\InstrABR{I}{C}.
	\end{align*}
	Now, $(\T{B}_{z}^{(y)}\T{A}_y\boxtimes\T{I}_{\sys{E}})\rho_{x}\in\St{AE}$ for all $x\in\Out{X},y\in\Out{Y},z\in\Out{Z}^{(y)}$ by hypothesis. In addition, by Property~\ref{prope:causality} and using the characterisation of deterministic transformations given in Subsec.~\ref{subsec:probabilistic}, one has:
	\begin{align*}
	&\rbra{e}_{\sys{CE}}\sum_{z\in\Out{Z}^{(y)}}\T{B}_{z}^{(y)}\boxtimes\T{I}_{\sys{E}} = \rbra{e}_{\sys{BE}},\quad \forall y\in\Out{Y}, \\
	&\rbra{e}_{\sys{BE}}\sum_{y\in\Out{Y}}\T{A}_{y}\boxtimes\T{I}_{\sys{E}} = \rbra{e}_{\sys{AE}}, \\
	&\rbra{e}_{\sys{AE}}\sum_{x\in\Out{X}}\rho_{x} = 1.
	\end{align*}
	Accordingly:
	\begin{align*}
		\rbra{e}_{\sys{CE}}\sum_{x\in\Out{X}}\sum_{y\in\Out{Y}}\sum_{z\in\Out{Z}^{(y)}}(\T{B}_{z}^{(y)}\T{A}_y\boxtimes\T{I}_{\sys{E}})\rho_{x} = 1,
	\end{align*}
	and then, by Postulate~\ref{postulate:instruments}, one has that $\mathsf{P}\in\InstrAB{I}{C}$. Finally, by Property~\ref{prope:unrestricted_instruments}, one concludes that the conditional generalised instrument	$\bigcup_{y\in\mathsf{Y}}\lbrace\T{B}_{z}^{(y)}\T{A}_y\rbrace_{z\in\mathsf{Z}^{(y)}}$ is an instrument of the theory $\Theta$.
\end{proof}

\section{Homogeneous strict bilocal discriminability and essential uniqueness of purification imply postulates~\eqref{eq:dimensions} and~\eqref{eq:state_composition} in a simplicial theory}\label{app:simplicial}
Let $\Theta$ be a simplicial theory, and define, for all pure states $\rket{i}_{\sys{A}}\in\PurSt{A},\rket{j}_{\sys{B}}\in\PurSt{B}$, the set $\PurRef{ij}{AB}\subseteq\PurSt{AB}$ collecting those pure states of $\sys{AB}$ which convexly refine the product state $\rket{i}_{\sys{A}}\rket{j}_{\sys{B}}$. Moreover, define:
\begin{align*}
&n_{\sys{AB}}^{ij}\coloneqq \lvert\PurRef{ij}{AB}\rvert,\\
&\PurRefA{AB}\coloneqq \bigcup_{\substack{1\leq i\leq \D{A} \\ 1\leq j\leq \D{B}}} \PurRef{ij}{AB}, \\
&\PurRefbar{AB}\coloneqq \PurSt{AB}\setminus\PurRefA{AB},\\
&l_{\sys{AB}}\coloneqq \lvert\PurRefbar{AB}\rvert.
\end{align*}
By direct inspection of the proof of Theorem 2 in Ref.~\cite{d2019classical}, one easily verifies that, for all $(i,j)\neq(i',j')$, $\PurRef{ij}{AB}\cap\PurRef{i'j'}{AB}=\emptyset$. Accordingly, in a simplicial theory $\Theta$ each $\rket{\rho}_{\sys{AB}}\in\PurRefA{AB}$ can be unambiguously labelled as follows:
\begin{align*}
\rket{\rho}_{\sys{AB}} = \rket{\left(ij\right)_\sigma}_{\sys{AB}},\quad \sigma\in \{ 1,2,\ldots, n_{\sys{AB}}^{ij}\}.
\end{align*}
Then, for all systems $\sys{A},\sys{B}$ and $i,j,\sigma$,
the following holds:
\begin{align}
&\D{AB}=\sum_{\substack{1\leq i\leq \D{A} \\ 1\leq j\leq \D{B}}}n_{\sys{AB}}^{ij} + l_{\sys{AB}},\label{eq:composition_general_general} \\
&\rket{i}_{\sys{A}}\rket{j}_{\sys{B}}= \sum_{\sigma'=1}^{n_{\sys{AB}}^{ij}} p_{\sigma'}^{ij}\rket{(ij)_{\sigma'}}_{\sys{AB}},\quad p_{\sigma}^{ij}>0,\quad  \sum_{\sigma'=1}^{n_{\sys{AB}}^{ij}} p_{\sigma'}^{ij}=1.\label{eq:classification_states_general}
\end{align}
\begin{theorem}\label{thm:classification}
	Let $\Theta$ be a simplicial theory satisfying Property~\ref{prope:uniqueness_purification}. Then, for all systems $\sys{A},\sys{B}\in\Sys{\Theta}$, there exists a positive integer $n_{\sys{AB}}$ such that the following holds. For all pure states $\rket{i}_{\sys{A}}\in\PurSt{A},\rket{j}_{\sys{B}}\in\PurSt{B}$, permutation $\pi$ of $\D{A}$ elements, there exist a reversible transformation $\T{R}\in\RevTransf{A}{A}$ and a permutation $\kappa_{ij}$ of $n_{\sys{AB}}$ elements, such that:
	\begin{align}
	&n_{\sys{AB}}^{ij}=n_{\sys{AB}},\label{eq:composition_general_reversible}\\
	&\left(\T{R}\boxtimes\T{I}_{\sys{B}}\right)\rket{\left(ij\right)_\sigma}_{\sys{AB}}= \rket{\left(\pi(i)j\right)_{\kappa_{ij}(\sigma)}}_{\sys{AB}},\label{eq:reversible1}\\
	&\D{AB}= n_{\sys{AB}}\D{A}\D{B} + l_{\sys{AB}},\label{eq:composition_general}\\
	&\rket{i}_{\sys{A}}\rket{j}_{\sys{B}}= \frac{1}{n_{\sys{AB}}}\sum_{\sigma=1}^{n_{\sys{AB}}}\rket{(ij)_{\sigma}}_{\sys{AB}}.\label{eq:composition_states}
	\end{align}
\end{theorem}
\begin{proof}
	Property~\ref{prope:uniqueness_purification}, choosing $\sys{A}=\sys{I}$, implies \emph{transitivity of reversible channels on pure states}. That is, for every system $\sys{A}$ and every permutation $\pi$ of $\D{A}$ elements, there exists a reversible transformation $\T{R}\in\RevTransf{A}{A}$ such that $\T{R}\rket{i}_\sys{A}=\rket{\pi(i)}_\sys{A}$ for all $\rket{i}_\sys{A}\in\PurSt{A}$. Let $\pi$ and $\pi'$ denote two permutations of, respectively, $\D{A}$ and $\D{B}$ elements, and define $\tilde{i}\coloneqq\pi(i),\tilde{j}\coloneqq\pi'(j)$ for all $i\in\{1,\ldots,\D{A}\},j\in\{1,\ldots,\D{B}\}$. Then, combining Property~\ref{prope:uniqueness_purification} with Proposition~\ref{prop:pure_to_pure}, for all $\pi,\pi'$ there exist reversible transformations $\T{R}\in\RevTransf{A}{A},\T{R}'\in\RevTransf{B}{B}$, such that the following holds:
	\begin{align}\label{eq:pure_to_pure}
	\begin{split}
	&\rket{\tilde{i}}_\sys{A}\coloneqq\T{R}\rket{i}_\sys{A}\in\PurSt{A},\\ &\rket{\tilde{j}}_\sys{B}\coloneqq\T{R}'\rket{j}_\sys{B}\in\PurSt{B}.
	\end{split}
	\end{align}
	Still by Proposition~\ref{prop:pure_to_pure}, we can also denote:
	\begin{align}\label{eq:pure_states_permuted}
	\left(\T{R}\boxtimes\T{R}'\right)\rket{(ij)_\sigma}_{\sys{AB}} = \rket{\left(\iota^{ij}_{\sigma}\upsilon^{ij}_{\sigma}\right)_{\kappa^{ij}_{\sigma}}}_{\sys{AB}}.
	\end{align}
	Recall now Eq.~\eqref{eq:classification_states_general}, which holds by simpliciality. Eqs.~\eqref{eq:pure_to_pure} and~\eqref{eq:pure_states_permuted}, combined with Eq.~\eqref{eq:classification_states_general}, read as
	\begin{align}
	\begin{split}\label{eq:product_perm}
	\rket{\tilde{i}}_{\sys{A}}\rket{\tilde{j}}_{\sys{B}}&= 
	\sum_{\tilde{\sigma}=1}^{n_{\sys{AB}}^{\tilde{i}\tilde{j}}} p_{\tilde{\sigma}}^{\tilde{i}\tilde{j}}\rket{\left(\tilde{i}\tilde{j}\right)_{\tilde{\sigma}}}_{\sys{AB}}=\left(\T{R}\boxtimes\T{R}'\right)\rket{i}_{\sys{A}}\rket{j}_{\sys{B}}=\\ &=\sum_{\sigma=1}^{n_{\sys{AB}}^{ij}} p_{\sigma}^{ij}\rket{\left(\iota^{ij}_{\sigma}\upsilon^{ij}_{\sigma}\right)_{\kappa^{ij}_{\sigma}}}_{\sys{AB}}.
	\end{split}
	\end{align}
	First of all, by simpliciality, Eq.~\eqref{eq:product_perm} implies $\iota^{ij}_{\sigma}=\tilde{i}, \upsilon^{ij}_{\sigma}=\tilde{j}$ for all $i\in\{1,\ldots,\D{A}\},j\in\{1,\ldots,\D{B}\},\sigma\in\{1,\ldots, n_{\sys{AB}}^{ij}\}$. Moreover, by Proposition~\ref{prop:pure_to_pure}, for all $i,j,\tilde{i},\tilde{j},k$ it must be:
	\begin{align}\label{eq:final}
	n_{\sys{AB}}^{ij}=n_{\sys{A'B'}}^{\tilde{i}\tilde{j}},\quad \kappa^{ij}_{\sigma}=\kappa_{ij}(\sigma),\quad p_{\kappa_{ij}(\sigma)}^{\tilde{i}\tilde{j}}=p_{\sigma}^{ij},
	\end{align}
	where $\kappa_{ij}$ is a permutation of $n_{\sys{AB}}^{ij}$ elements. This proves Eq.~\eqref{eq:composition_general_reversible}, Eq.~\eqref{eq:reversible1}, and---recalling Eq.~\eqref{eq:composition_general_general}---Eq.~\eqref{eq:composition_general}. For all $i,j,\sigma$, by simpliciality, $\rket{(ij)_{\sigma}}_{\sys{AB}}$ is a purification of $\rket{i}_{\sys{A}}$ and $\rket{j}_{\sys{B}}$. Thus, by Property~\ref{prope:uniqueness_purification}, in Eq.~\eqref{eq:final} the permutation $\kappa_{ij}$ in $\PurRef{ij}{AB}$ can be arbitrarily chosen when $i=\tilde{i}$ and $j=\tilde{j}$. Since $p_{\sigma}^{ij}>0$ and $\sum_{\sigma'=1}^{n_{\sys{AB}}} p_{\sigma'}^{ij}=1$ for all $i,j,\sigma$, then one has $p_{\sigma}^{ij}=1/n_{\sys{AB}}$ for all $i,j,\sigma$. This finally also proves Eq.~\eqref{eq:composition_states}.
\end{proof}
\begin{corollary}\label{cor:classification_bilocal}
	Let $\Theta$ be a simplicial theory satisfying Property~\ref{prope:uniqueness_purification}, $\sys{I}\neq\sys{A},\sys{B}\in\Sys{\Theta}$, and $i\in\{1,2,\ldots,\D{A}\},j\in\{1,2,\ldots,\D{B}\}$. Then $\Theta$ also satisfies strict bilocal discriminability, with $\D{AB}>\D{A}\D{B}$, if and only if $l_{\sys{AB}}=0$ and $n_{\sys{AB}}^{ij}=2$.
\end{corollary}
\begin{proof}
	By Theorem 4 of Ref.~\cite{d2019classical}---which holds for simplicial theories with $n$-local discriminability for some positive integer $n$---one has that $l_{\sys{AB}}=0$ for all $\sys{A},\sys{B}\in\Sys{\Theta}$. Now, it suffices to plug Eq.~\eqref{eq:composition_general} into Eq.~\eqref{eq:dimensions_bilocal} of Theorem~\ref{thm:bilocal}. Solving for $n_{\sys{AB}}$, one finds the two solutions $n_{\sys{AB}}=1$ or $n_{\sys{AB}}=2$, and then it must be $n_{\sys{AB}}=2$ for every $\sys{I}\neq\sys{A},\sys{B}\in\Sys{\Theta}$. The converse has been proven in Proposition~\ref{prop:bilocal}.
\end{proof}

\end{document}